\documentclass[12pt]{article}

\usepackage{amsmath,amsfonts,amssymb,amsthm,graphicx,cite}
\usepackage{amscd}
\usepackage{mathrsfs}
\usepackage{undertilde}
\usepackage[all,cmtip]{xy}

\numberwithin{equation}{section}

\begin{document}

\newtheorem{theorem}{Theorem}[section]
\newtheorem{corollary}[theorem]{Corollary}
\newtheorem{lemma}[theorem]{Lemma}
\newtheorem{proposition}[theorem]{Proposition}

\newcommand{\adiffop}{A$\Delta$O}
\newcommand{\adiffops}{A$\Delta$Os}

\newcommand{\beq}{\begin{equation}}
\newcommand{\eeq}{\end{equation}}
\newcommand{\bea}{\begin{eqnarray}}
\newcommand{\eea}{\end{eqnarray}}
\newcommand{\sh}{{\rm sh}}
\newcommand{\ch}{{\rm ch}}
\newcommand{\einde}{$\ \ \ \Box$ \vspace{5mm}}
\newcommand{\De}{\Delta}
\newcommand{\de}{\delta}
\newcommand{\Z}{{\mathbb Z}}
\newcommand{\N}{{\mathbb N}}
\newcommand{\C}{{\mathbb C}}
\newcommand{\Cs}{{\mathbb C}^{*}}
\newcommand{\R}{{\mathbb R}}
\newcommand{\Q}{{\mathbb Q}}
\newcommand{\T}{{\mathbb T}}
\newcommand{\re}{{\rm Re}\, }
\newcommand{\im}{{\rm Im}\, }
\newcommand{\cW}{{\cal W}}
\newcommand{\cJ}{{\cal J}}
\newcommand{\cE}{{\cal E}}
\newcommand{\cA}{{\cal A}}
\newcommand{\cR}{{\cal R}}
\newcommand{\cP}{{\cal P}}
\newcommand{\cM}{{\cal M}}
\newcommand{\cN}{{\cal N}}
\newcommand{\cI}{{\cal I}}
\newcommand{\cMs}{{\cal M}^{*}}
\newcommand{\cB}{{\cal B}}
\newcommand{\cD}{{\cal D}}
\newcommand{\cC}{{\cal C}}
\newcommand{\cL}{{\cal L}}
\newcommand{\cF}{{\cal F}}
\newcommand{\cH}{{\cal H}}
\newcommand{\cS}{{\cal S}}
\newcommand{\cT}{{\cal T}}
\newcommand{\cU}{{\cal U}}
\newcommand{\cQ}{{\cal Q}}
\newcommand{\cV}{{\cal V}}
\newcommand{\cK}{{\cal K}}
\newcommand{\intR}{\int_{-\infty}^{\infty}}
\newcommand{\intI}{\int_{0}^{\pi/2r}}
\newcommand{\limp}{\lim_{\re x \to \infty}}
\newcommand{\limn}{\lim_{\re x \to -\infty}}
\newcommand{\limpn}{\lim_{|\re x| \to \infty}}
\newcommand{\diag}{{\rm diag}}
\newcommand{\Ln}{{\rm Ln}}
\newcommand{\Arg}{{\rm Arg}}

\title{Kernel functions and B\"acklund transformations for relativistic Calogero-Moser and Toda systems}
\author{Martin Halln\"as \\School of Mathematics, \\ Loughborough University, Loughborough LE11 3TU, UK \\ and \\Simon Ruijsenaars \\ School of Mathematics, \\ University of Leeds, Leeds LS2 9JT, UK}

\date{7 October 2012}

\maketitle

\begin{abstract}
We obtain kernel functions associated with the quantum relativistic Toda systems, both for the periodic version and for the nonperiodic version with its dual. This involves taking limits of previously known results concerning kernel functions for the elliptic and hyperbolic relativistic Calogero-Moser systems. We show that the special kernel functions at issue admit a limit that yields generating functions of B\"acklund transformations for the classical relativistic Calogero-Moser and Toda systems. We also obtain the nonrelativistic counterparts of our results, which tie in with previous results in the literature. 
\end{abstract}

\tableofcontents

\section{Introduction}
This paper is primarily concerned with the relativistic generalizations of the $N$-particle Calogero-Moser and Toda systems, both on the quantum and on the classical level. A survey of these systems can be found in~\cite{Rui94}. In addition, we briefly discuss the specializations  of our results to the nonrelativistic systems,
surveyed in~\cite{OP81} and in~\cite{OP83} on the classical and quantum level, resp.

The classical relativistic Calogero-Moser and Toda systems can be defined by Poisson commuting Hamiltonians of the form
\begin{equation}\label{Hamiltonians}
	S_k(x,p) = \sum_{\substack{I\subset\lbrace 1,\ldots,N\rbrace\\ |I|=k}}V_I(x)\prod_{l\in I}\exp(\beta p_l),\quad k=1,\ldots,N,
\end{equation}
with $\beta =1/mc$, where $m>0$ is the particle rest mass and $c>0$ the speed of light. The elliptic version of the Calogero-Moser system describes $N$ interacting particles on a line or ring, with $V_I$ given by 
\begin{equation}
	V_I(x) = \prod_{\substack{m\in I\\ n\notin I}}f(x_m-x_n).
\end{equation}
Here, the function $f$ encoding the interaction is defined by
\beq\label{f}
	f(z) = \left( \frac{s(z+\rho)s(z-\rho)}{s^2(z)}\right)^{1/2},
\eeq
where $s(z)$ is essentially the Weierstrass sigma function (see \eqref{ssigrel} in Appendix \ref{gammaFuncsAppendix} for the precise relation). Throughout this paper we shall work with a fixed positive half-period
\beq
\omega =\pi/(2r),\ \ \ r>0,
\eeq
and with two purely imaginary half-periods $ia_{+},ia_{-}$ in the relativistic quantum regime (cf.~\eqref{aa} below), whereas in the classical ($\hbar=0$) and nonrelativistic ($\beta=0$) regimes we need only one purely imaginary half-period, parametrized as
\beq
\omega'=i\alpha/2,\ \ \ \alpha>0.
\eeq 
 With the `coupling constant' $\rho$ constrained by
 \beq\label{mures}
 \rho\in i(0,\alpha),
 \eeq
 it follows that $f(z)^2$ is positive on the period interval $z\in(0,\pi/r)$. By taking the positive square root, we thus obtain well-defined positive coefficients $V_I(x)$ and Hamiltonians $S_k(x,p)$ on the phase space
 \beq\label{Om}
 \Omega =\{(x,p)\in\R^{2N}\mid x\in G\},
 \eeq
 where $G$ is the configuration space
 \beq\label{config}
 G= \{ x\in \R^N\mid  x_N<\cdots <x_1, x_1-x_N\in(0,\pi/r)\}. 
 \eeq
 
 The classical relativistic version of the Calogero-Moser systems dates back to~\cite{RS86}. A quantization preserving commutativity was found in~\cite{Rui87}.
For the elliptic systems it is given by the commuting analytic difference operators (henceforth  \adiffops)
\beq\label{qSk}
		\hat{S}_k(x) = \sum_{\substack{I\subset\lbrace 1,\ldots,N\rbrace\\ |I|=k}}\prod_{\substack{m\in I\\ n\notin I}}f_-(x_m-x_n)\prod_{l\in I}\exp(-i\hbar \beta\partial_{x_l})\prod_{\substack{m\in I\\ n\notin I}}f_+(x_m-x_n),\ \ \ k=1,\ldots,N,
\eeq
where $\hbar>0$ is Planck's constant, and
\begin{equation}\label{fpm}
	f_\pm(z) = \big(s(z\pm\rho)/s(z)\big)^{1/2}.
\end{equation}
With the above constraints on the parameters, they are formally self-adjoint. To promote them to commuting self-adjoint operators on the Hilbert space $L^2(G,dx)$, however, is an open problem. 

The unexpected existence of quite special kernel functions yields a novel perspective for solving this long-standing problem. This is explained in some detail in~\cite{Rui04} and~\cite{Rui09}, but in order to render our account more self-contained, we now digress to make the notion of `kernel function' at issue in this paper more precise. 

Often, the term is loosely used for a function~$\Psi(v,w)$ depending on variables $v$ and $w$ that may vary over spaces of different dimensions. Here, however, it is used in the restricted sense that the function at hand connects a pair of operators $H_1(v)$ and $H_2(w)$ acting on it via an equation of the form
\beq\label{H1H2Id}
(H_1(v)-H_2(w))\Psi(v,w) = 0,
\eeq
which we refer to as a `kernel identity'. Therefore, $\Psi(v,w)$ can be viewed as a zero-eigenvalue eigenfunction of the difference operator~$H_1(v)-H_2(w)$.
In our setting, the operators~$H_1(v)$ and $H_2(w)$ will be either \adiffops~or PDOs.  

For the one-variable case~$v,w\in\R$, such kernel functions have a long history. For example, when~$H_1(v)$ and $H_2(w)$ are Mathieu operators (which can be viewed as Hamiltonians for the reduced two-variable nonrelativistic periodic Toda system), a host of kernel functions can be found in Section~4.1 of Arscott's monograph~\cite{Ars64}. Going yet further back, for certain pairs of Lam\'e operators  (which can be viewed as Hamiltonians for the reduced two-variable nonrelativistic elliptic Calogero-Moser system), kernel functions occur in a paper by Whittaker \cite{Whi15}, see also Section~23.6 of~\cite{WW35}. 

The appearance of kernel functions in a multi-variable setting is far more recent. For the nonrelativistic periodic Toda systems they can be found in a paper by Pasquier and Gaudin~\cite{PG92}, while for the nonrelativistic elliptic Calogero-Moser systems kernel functions emerged from Langmann's work on anyonic quantum field theory~\cite{Lan00}. We note that in the former paper explicit kernel identities of the above form do not yet appear (they are present implicitly), and in the latter only a pair of defining Hamiltonians was considered. 

Returning to the present paper, further work related to kernel functions will be cited in due course. Our starting point consists of
the elliptic kernel functions from Ref.~\cite{Rui06}. They satisfy~$N$ kernel identities of the form
\beq\label{SSPsi}
(\hat{S}_k(x)-\hat{S}_k(-y))\Psi(x,y)=0,\ \ \ \ k=1,\ldots, N,
\eeq
and we shall recall their precise definition in Subsection~2.1.

On the other hand, one key feature of the kernel functions should already be mentioned now: Their building block is the elliptic gamma function (introduced and studied in~\cite{Rui97}), and this function is symmetric under the interchange of the positive parameters $\alpha$ and $\hbar\beta$ featuring in the A$\De$Os $\hat{S}_k$. This property is now often called `modular invariance', and it entails that~\eqref{SSPsi} also holds for the A$\De$Os obtained by interchanging $\alpha$ and $\hbar\beta$. It is not hard to check that the latter A$\De$Os commute with the previous ones, and so it is natural to insist on a joint diagonalization on $L^2(G,dx)$.

In this paper, however, we are not directly concerned with joint eigenfunctions of the elliptic A$\De$Os and their hyperbolic counterparts. Rather, we obtain some new insights concerning the hyperbolic kernel functions (as detailed in Subsection~2.2), and extend the theory of kernel functions to the relativistic Toda systems (in Subsections~2.3--2.5 and Appendix~B). 

On the other hand, we do expect that the results of the present paper will be of pivotal importance to solve open problems concerning the existence and properties of joint eigenfunctions of the A$\De$Os of hyperbolic and Toda type. As a first example along these lines, we have shown in  a recent paper how the hyperbolic kernel functions can be used for a recursive construction of modular invariant joint eigenfunctions of the commuting hyperbolic A$\De$Os~\cite{HR12}. The idea that such a recursive construction might exist came from papers by Kharchev, Lebedev and Semenov-Tian-Shansky~\cite{KLS02}  and by Gerasimov, Kharchev and Lebedev~\cite{GKL04}, which in turn were inspired by the pioneering work of~Gutzwiller~\cite{Gut81}. As it turns out, the recursive constructions in~\cite{KLS02} and~\cite{GKL04} can be simplified and generalized in terms of kernel functions for the nonperiodic Toda systems and hyperbolic Calogero-Moser systems, all of which arise naturally within the context of the present paper. 

Accordingly, using the kernel functions detailed below as a starting point, we hope to arrive at a unified picture for the joint eigenfunctions of the various regimes. This will involve in particular detailed estimates on the kernel functions and their duals, which are necessary to control various analytical difficulties arising in the recursive constructions of the joint eigenfunctions and the study of their duality and Hilbert space properties. 

In order to sketch our results on hyperbolic kernel functions, let us note first that the functions $f$~\eqref{f} and $f_{\pm}$~\eqref{fpm} reduce to
\beq\label{fh}
f(z)=\big(1-\sinh^2(\pi \rho/\alpha)/\sinh^2(\pi z/\alpha)\big)^{1/2},\ \ \ \rho\in i(0,\alpha),
\eeq
\beq\label{fhpm}
f_{\pm}(z)= \big(\sinh(\pi (z\pm\rho)/\alpha)/\sinh(\pi z/\alpha)\big)^{1/2}.
\eeq
It is also immediate that the elliptic kernel functions satisfying~\eqref{SSPsi} have hyperbolic analogs: We need only replace the elliptic gamma function by its hyperbolic counterpart.
The latter is modular invariant, like its elliptic generalization. (We review the relevant features of these gamma functions in Appendix~A.) However, by contrast to the elliptic case, it is possible to obtain kernel functions relating the $N$-particle A$\De$Os $\hat{S}_k(x)$, $k=1,\ldots,N$, to sums of $(N-\ell)$-particle A$\De$Os, with $\ell=1,\ldots,N$. This is because we can take $y_N,\ldots,y_{N-\ell+1}$ to infinity, a procedure that has no elliptic analog. (For the trigonometric regime---which we do not consider---and for the case $k=1$, a similar result was obtained first by~Komori, Noumi and Shiraishi~\cite{KNS09}, among a host of other ones.)  The kernel function for the case~$\ell=1$ is the key building block for the recursive construction of modular invariant joint eigenfunctions~\cite{HR12}.

Turning to our Toda results, we first recall some general features.
For the periodic and nonperiodic Toda systems the coefficients $V_I$ in~\eqref{Hamiltonians} are given by
\begin{equation}\label{VIT}
	V_I(x) = \prod_{\substack{m\in I\\ m+1\notin I}}f_T(x_{m+1}-x_m)\prod_{\substack{m\in I\\ m-1\notin I}}f_T(x_m-x_{m-1}).
\end{equation}
The pair potential reads
\begin{equation}\label{fT}
	f_T(z) = \left(1 + \gamma^2\exp(2\pi z/\alpha)\right)^{1/2},\quad \gamma\in\mathbb{R},\quad \alpha>0,
\end{equation}
and the periodic and nonperiodic versions are encoded via the convention
\beq\label{pconv}
x_0=x_N,\ \ \ x_{N+1}=x_1, \ \ \ \ (\mathrm{periodic\  Toda}),
\eeq
\beq\label{npconv}
x_0=\infty,\ \ \ x_{N+1}=-\infty, \ \ \ \ (\mathrm{nonperiodic \ Toda}).
\eeq
In both cases the functions $S_k(x,p)$ yield positive Poisson commuting Hamiltonians on the Toda phase space
\beq
\Omega_T =\{(x,p)\in \R^{2N} \}.
\eeq

The classical relativistic Toda systems were introduced in~\cite{Rui90}, together with a quantization preserving commutativity.
The latter is given by the commuting \adiffops
\begin{multline}\label{TodaOps}
	\hat{S}_k(x) =  \sum_{\substack{I\subset\lbrace 1,\ldots,N\rbrace\\ |I|=k}}\prod_{\substack{m\in I\\ m+1\notin I}}f_T(x_{m+1}-x_m)\prod_{l\in I}\exp(-i\hbar\beta\partial_{x_l})\prod_{\substack{m\in I\\ m-1\notin I}}f_T(x_m-x_{m-1}),\\ \ k=1,\ldots,N.
\end{multline}
Since $f_T$ has period $i\alpha$, the A$\De$Os obtained by interchanging $\alpha$ and $\hbar\beta$ commute with the above ones, and so the question whether modular invariant eigenfunctions exist arises once again. We proceed with some remarks concerning this modular invariance feature, which are also meant to elucidate the context of various $q$-Toda papers where modular symmetry is not at issue, cf.~the paragraph below~\eqref{qq}.

When the relativistic Toda systems were last surveyed in~\cite{Rui94}, modular invariance was not mentioned. Indeed, at that time no eigenfunctions were known at all, so it was not clear that one should be looking for eigenfunctions with this symmetry property.  Moreover,  the Toda A$\De$Os $\hat{S}_k$ given by~\eqref{TodaOps} are not even formally self-adjoint on $L^2(\R^N,dx)$ (as pointed out in~Subsection~6.1 of~\cite{Rui94}), so that their Hilbert space status seemed quite opaque.  
 
It transpired from the above-mentioned work by Kharchev, Lebedev and Semenov-Tian-Shansky~\cite{KLS02} in the context of quantum group representation theory  that a slight modification of the quantum coupling dependence remedies the lack of formal self-adjointness of the Toda A$\De$Os. Moreover, the modular invariance property showed up in the eigenfunctions presented in~\cite{KLS02} and was tied in with Faddeev's notion of modular double of a quantum group~\cite{Fad99}. Likewise, van de Bult has shown that the modular invariance  of the `relativistic' hypergeometric function (introduced in Subsection~6.3 of~\cite{Rui94}) can be understood from this quantum group perspective~\cite{vdB06}.

Both formal self-adjointness  and  modular symmetry emerge naturally from our new results on Toda kernel functions. We expect that these results will be crucial to solve open problems concerning the modular invariant joint eigenfunctions obtained in~\cite{KLS02}, including duality properties, orthogonality and completeness.

The change in coupling dependence entailing formal self-adjointness   consists in replacing the classical pair potential  in~\eqref{TodaOps} by the quantum counterpart
\beq\label{qfT}
\hat{f}_T(z) = \left(1 + \exp\Big(\frac{\pi}{\alpha} [2z+2\eta +i\hbar\beta]\Big)\right)^{1/2}.
\eeq
This amounts to the replacement
\beq\label{gameta}
\gamma^2\to \exp(\pi(2\eta +i\hbar\beta)/\alpha)
\eeq
in~\eqref{fT}.
Clearly, this substitution does not change the A$\De$O-commutativity features mentioned above. The parameter $\eta\in\R$ plays the role of coupling constant, with the limit $\eta\to -\infty$ yielding the free theory. Moreover, in the classical limit $\hbar\to 0$ the shift into the complex plane disappears. Note in this connection that there is no classical analog of modular invariance. (More precisely, the Hamiltonians obtained from the functions $S_k(x,p)$ by interchanging $\alpha$ and $\beta$ do not Poisson commute with $S_1(x,p),\ldots,S_{N-1}(x,p)$ for $1\le k<N$ and $\alpha\ne\beta$.)

The new kernel functions $\Psi(x,y)$ for the relativistic Toda systems obtained in Subsections~2.3 and~2.4 for the periodic and nonperiodic case, resp., have the hyperbolic gamma function as their building block. Just as in the quantum elliptic and hyperbolic cases, we switch in the quantum Toda case to positive parameters
\beq\label{aa}
a_{+}=\alpha,\ \ \ a_{-}=\hbar \beta,
\eeq
in terms of which modular-invariant formulas involving the hyperbolic and elliptic gamma functions are more readily expressed. 

We obtain the Toda kernel functions in a somewhat tortuous way. Basically, we exploit the previous elliptic and hyperbolic Calogero-Moser results to arrive at them. The connection between the hyperbolic Calogero-Moser Hamiltonians and their nonperiodic Toda counterparts was already observed and used by the second-named author in 1985, which led him to the relativistic Toda systems via the relativistic  hyperbolic Calogero-Moser systems~\cite{Rui90}. The relation between the defining Hamiltonian of the nonrelativistic elliptic Calogero-Moser system and its periodic Toda counterpart was first pointed out by Inozemtsev~\cite{Ino89}, and then generalized  to all of the commuting relativistic Hamiltonians, cf.~the survey\cite{Rui94}. 

Here we need these results as well, but to control the pertinent limit for the kernel functions seems not feasible on the elliptic level. Instead, we determine the limits of the functional equations expressing the kernel function property to obtain Toda functional equations. These can then be viewed as corresponding to kernel functions for periodic Toda A$\De$Os related to the above ones by a similarity transformation. The details can be found in~Subsection~2.3.

In Subsection~2.4 we first show how nonperiodic Toda kernel functions can be obtained as a limit of  the periodic ones.  Just as in Subsection~2.2, we can also obtain kernel functions connecting $N$-particle to $M$-particle A$\De$Os, but here this seems only feasible for the case $|N-M|\le 1$. At the end of this subsection we detail how the previous nonperiodic results and a few new ones follow directly from their hyperbolic Calogero-Moser counterparts. Some readers might prefer this avenue, since it does not involve the periodic Toda and elliptic regimes. 

At this point we would like to mention that the limit transitions from Calogero-Moser to Toda type systems have recently become important from the viewpoint of quantum groups and Cherednik algebras, cf.~the lecture notes~\cite{CM09} and various references given there. By contrast to our perspective and that of the paper~\cite{KLS02} cited earlier, this work involves a single deformation parameter $q$ not on the unit circle. Here we are dealing with two parameters
\beq\label{qq}
q_{+}=\exp(i\pi a_{+}/a_{-}),\ \ \ q_{-}=\exp(i\pi a_{-}/a_{+}),
\eeq
and for Hilbert space/quantum mechanical purposes it is of pivotal importance that $a_{+}$ and $a_{-}$ be positive, so that $|q_{\pm}|=1$. Moreover, the kernel functions and eigenfunctions are invariant under the interchange of $a_{+}$ and $a_{-}$, which arises from the hyperbolic gamma function $G(a_{+},a_{-};z)$ featuring as a building block. 

For the case where the building block is the trigonometric gamma function (better known as the $q$-gamma function), the Hilbert space status of the $q$-Toda A$\De$Os and their joint eigenfunctions is opaque, but in this case there are intimate connections to various issues in representation theory and algebraic geometry. Some early references include ~\cite{Eti99}, \cite{Sev00}, \cite{OR02}, \cite{GL03}. In particular, in a series of papers by Olshanetsky and Rogov (which can be traced from~\cite{OR02}), the shift in the imaginary direction featuring in \eqref{TodaOps} was for the  first time traded for a shift by $\hbar$ in the real direction, hence yielding a parameter $q=\exp(-\hbar)<1$. Their work concerns the rank-1 (2-particle) case, whereas Etingof's paper~\cite{Eti99} appears to be the first where the arbitrary-rank case is dealt with. 

The notion of `dual relativistic Toda systems' at issue in Subsection~2.5 is not widely known. On the classical level these systems emerged from the explicit construction of an action-angle map for the nonperiodic Toda systems~\cite{Rui90}. They are integrable systems for which the actions $\hat{p}$  and angles $\hat{x}$ play the role of the positions and momenta in the original system, respectively. Specifically, the Poisson commuting Hamiltonians can be chosen as
\beq\label{Hdual}
H_k(\hat{p},\hat{x})= \sum_{\substack{I\subset\lbrace 1,\ldots,N\rbrace\\ |I|=k}}
 \prod_{\substack{m\in I\\ n\notin I}}\frac{\beta/2}{|\sinh(\beta(\hat{p}_m-\hat{p}_n)/2)|}
\prod_{l\in I}\exp(2\pi \hat{x}_l/\alpha),\quad k=1,\ldots,N.
\eeq

The classical hyperbolic relativistic Calogero-Moser systems are self-dual, since the action-angle map is essentially an involution~\cite{Rui88}. More specifically, the symmetric functions of the \lq dual Lax matrix\rq
\beq
A(x)=\diag (\exp(2\pi x_1/\alpha),\ldots,\exp(2\pi x_N/\alpha)),
\eeq
turn into the dual Hamiltonians
\beq\label{Sdual}
S_k(\hat{p},\hat{x})= \sum_{\substack{I\subset\lbrace 1,\ldots,N\rbrace\\ |I|=k}}
 \prod_{\substack{m\in I\\ n\notin I}}
 \big(1-\sinh^2(\pi \rho/\alpha)/\sinh^2(\beta(\hat{p}_m-\hat{p}_n)/2)\big)^{1/2}
 \prod_{l\in I}\exp(2\pi \hat{x}_l/\alpha).
\eeq
The limit transition from the hyperbolic Calogero-Moser to the nonperiodic Toda systems has a counterpart for the classical and quantum duals.
Just as for the original system, one needs to quantize the Hamiltonians  \eqref{Sdual} via the analog of the reordering in \eqref{qSk} (cf.~\eqref{fh}--\eqref{fhpm}) before substituting $\hat{x}_l\to-i\hbar \partial_{\hat{p}_l}$. This yields again commuting hyperbolic A$\De$Os, and then the desired commuting A$\De$Os  for the dual Toda case follow from the pertinent limit.
 Thus we obtain  quantum versions $\hat{H}_k$ of~\eqref{Hdual} and
corresponding kernel functions that are once again built from the hyperbolic gamma function. 

Without a change in notation, however, the procedure just sketched would lead to awkward formulas. Indeed, when we switch to the parameters $a_{\pm}$ via \eqref{aa}, then the factor $\sinh (\beta(\hat{p}_n-\hat{p}_m)/2)$ (for example) becomes
$\sinh (a_{-}(\hat{p}_n-\hat{p}_m)/2\hbar)$. Physically speaking, the unpleasant occurrence of $\hbar$ can be understood from the parameters $a_{+}$ and $a_{-}$ having the dimension [position], whereas $\hat{p}$ has the dimension [momentum]. To obtain the desired dual modular symmetry, we should trade $\hat{p}$ for a dual variable $\alpha\beta \hat{p}/2\pi$ with dimension [position]. We denote this new position by $v$, so that we need the substitution
\beq\label{phsub}
\hat{p}=2\pi v/(\alpha\beta).
\eeq
Using this variable, the dual hyperbolic A$\De$Os again take the form \eqref{qSk} with $f_{\pm}$ given by \eqref{fhpm} and $x$ replaced by $v$.

With this change of notation in place, the dual Toda kernel functions are symmetric under the interchange of $a_{+}$ and $a_{-}$, so that they are also kernel functions for the modular transforms of the $\hat{H}_k$, obtained by  interchanging $a_{+}$ and $a_{-}$. Somewhat surprisingly, for the dual nonperiodic Toda case we easily obtain kernel functions connecting the dual $N$-particle A$\De$Os to their $M$-particle versions for any $M<N$, whereas we can only handle the $M=N-1$ case in Subsection~2.4.

Our results for the dual Toda case are collected in Subsection~2.5. A close relative of the kernel function connecting the $N$-particle and $(N-1)$-particle dual A$\De$Os has appeared in the above-mentioned work by Kharchev et al.~\cite{KLS02}. It is used in a recursive construction of joint eigenfunctions for the nonperiodic Toda A$\De$Os, without a consideration of duality and Hilbert space properties. 

Before sketching the results of Section~3, we add an important remark concerning the A$\De$Os and kernel functions at issue in Section~2. The A$\De$Os are invariant when all of their coordinates~$x_n$ are shifted to $x_n+\xi$. This entails that the kernel function property is preserved under such coordinate shifts. We make use of this freedom to choose convenient kernel functions. Another common feature is invariance of the kernel function property under multiplication  by any function of the form
\begin{equation}\label{ambig}
	\phi\left(\sum_{n=1}^N(x_n-y_n)\right),\quad \phi~\text{meromorphic}.
\end{equation}
Hence, once we have identified one kernel function we immediately obtain an infinite-dimensional family of kernel functions.

Section~3 is concerned with so-called B\"acklund transformations for the classical relativistic Calogero-Moser and Toda systems. These are canonical transformations $(x,p)\mapsto (y,q)$   that preserve the Poisson commuting Hamiltonians, derived from a generating function $F(x,y)$ via
\beq\label{Fgen}
p_j= -\frac{\partial F}{\partial x_j},\ \ \ \  q_j= \frac{\partial F}{\partial y_j},\ \ \ j=1,\ldots,N.
\eeq
For the nonrelativistic Calogero-Moser systems such transformations appear to date back to work by Wojchiechowski~\cite{Woj82}. For the nonrelativistic infinite Toda chain a B\"acklund transformation can already be found in Toda's monograph~\cite{Tod81}. It seems Gaudin was the first to realize that it can also be applied to the finite Toda systems, and that it can be tied in with the classical limit of a kernel function for their quantum versions, cf.~Ch.~14 in his monograph~\cite{Gau83}.  Pasquier and Gaudin~\cite{PG92} then used the kernel function to study eigenvalues and eigenfunctions.

For the nonrelativistic rational Calogero-Moser system a B\"acklund transformation was obtained via special solutions of the KP equation by Nijhoff and Pang~\cite{NP94}, \cite{NP96}. They reinterpreted the generating function as a Lagrangian for a discrete map, which they viewed as a time-discretization of the defining Hamiltonian. In the same spirit, in Nijhoff/Ragnisco/Kuznetsov~\cite{NRK96} B\"acklund transformations (alias \lq time-discretizations\rq) for the relativistic Calogero-Moser systems were introduced and studied.

Later on, Kuznetsov and Sklyanin elaborated on the general theory of B\"acklund transformations~\cite{KS98}. In particular, they reconsidered the relation between kernel functions on the quantum level and generating functions on the classical level. Gaudin already pointed out this relation in the special case of the nonrelativistic periodic Toda system~\cite{Gau83}, but in~\cite{KS98} it was suggested more generally that the semi-classical behavior of a kernel function $\Psi(x,y)$ of the type we consider should be given by a formula of the form
\beq\label{PsiF}
\Psi(\hbar;x,y)\sim \exp(-iF(x,y)/\hbar),\ \ \ \ \hbar\to 0,
\eeq
where $F(x,y)$ generates a B\"acklund transformation for the classical version via~\eqref{Fgen}.

Now it seems quite unlikely that this is generally true, as kernel functions exist in profusion. Indeed, assuming one has found an orthonormal base $\{ \phi_n(x)\}_{n=0}^{\infty}$ of joint eigenfunctions for the commuting elliptic Hamiltonians (say), a function of the form
\beq
K((a_0,a_1,\ldots);x,y)=\sum_{n=0}^{\infty}a_n\phi_n(x)\overline{\phi_n(y)},
\eeq
is a Hilbert-Schmidt kernel function for any $(a_0,a_1,\ldots)\in \ell^2(\N)$. Since the numbers $a_n$  are arbitrary, it is not even clear what one would mean by the semi-classical behavior of such a general kernel function.

We are, however, dealing with very special kernel functions, which can be expressed in terms of the elliptic gamma function and its specializations. In particular, there is a notion of \lq classical limit\rq\ of the hyperbolic gamma function, which is tied to its appearance in the quantum scattering of the relativistic hyperbolic Calogero-Moser system. Indeed, in order to obtain the classical scattering (position shift) for $\hbar \to 0$ via a coherent state correspondence (which goes back to Hepp's fundamental paper~\cite{Hep74}), a quite special limit is required. This is detailed in Eq.~(4.73) of~\cite{Rui97}, and in terms of the hyperbolic gamma function it amounts to a certain zero step size limit, cf.~Prop.~III.7 in~\cite{Rui97}. The elliptic counterparts of these limits are Eq.~(4.98) and Prop.~III.13 in~\cite{Rui97}. (In Appendix~A we have recalled the two pertinent  limits, cf.~\eqref{Ghypcl} and~\eqref{Gellcl}.)

The point is now that with the associated $\hbar$-dependence in force, the asymptotic behavior encoded in~\eqref{PsiF} does yield the generating function of a B\"acklund transformation, as we shall show in Section~3 for each of the different cases at issue. More is true: For the relativistic Calogero-Moser case these generating functions are basically the ones arrived at in~\cite{NRK96}. Also, the B\"acklund transformations for the relativistic Toda regimes can be tied in with results by Suris~\cite{Sur96}, and their nonrelativistic limits yield the ones already known from the papers cited earlier.

There is however an unsettling phenomenon associated with these B\"acklund transformations, which seems not to have been pointed out before: They correspond to Calogero-Moser and Toda systems of an unphysical nature, inasmuch as there seems to be no choice of parameters that yields complete flows and phase space coordinates that stay real for all times. For the Toda regimes this disease can be remedied by an analytic continuation, a state of affairs that was already noted and used by Gaudin in the nonrelativistic case~\cite{Gau83}. (More precisely, he starts from the B\"acklund transformation with the physical positive coupling, and then finds an associated quantum kernel with the `wrong' coupling; this can then be remedied by analytic continuation of positions. Since we start with a positive coupling on the quantum level, we need to reverse this procedure.) 

For the Calogero-Moser case, however, this is no option. Indeed, even for the very simplest degeneration, namely, the nonrelativistic rational $N=2$ Calogero-Moser system, it seems impossible to avoid the \lq negative coupling\rq\ behavior. For the discrete map at issue it shows up in real initial positions becoming complex after a number of discrete time steps.
Even so, the circumstance that the pertinent quantum kernel functions give rise to generating functions of B\"acklund transformations  is highly remarkable and deserves a further scrutiny from the viewpoint of global analysis.

The relevant limits and their B\"acklund features yield somewhat unwieldy formulas, which is why we shall not detail them here. In Section~3 we reconsider successively the same regimes as in Section~2, omitting details whenever there is considerable similarity to previous cases.

Section~4 is concerned with the nonrelativistic version of our results. Here, too, we leave out details when they can be readily supplied by specialization. Subsection~4.1 deals with the nonrelativistic counterparts of the kernel functions of Section~2, whereas Subsection~4.2 is concerned with the nonrelativistic limits of the B\"acklund transformations of Section~3. As already mentioned, the kernel functions and B\"acklund transformations we arrive at in Section~4 are not new. On the other hand, we arrive at the relevant features in a novel way. 

In Appendix~A we review some properties of the hyperbolic and elliptic gamma functions from~\cite{Rui97} we have occasion to use. In Appendix~B we prove that the center-of-mass relativistic periodic Toda kernel functions that connect the $N$-particle A$\De$Os give rise to Hilbert-Schmidt operators. We anticipate that this result will be the key to proving that the commuting A$\De$Os can be promoted to commuting self-adjoint operators on~$L^2(\R^N)$, a problem that has been wide open for several decades.

To conclude this Introduction, a summary of our notation changes is in order. Indeed, as transpires from the above, we use several notational conventions that depend on the regime at issue. Our choices not only simplify formulas, but also reflect physical distinctions.  Specifically, in all regimes we have a length scale $\alpha$ coming from the interaction in the defining Hamiltonian. In the quantum relativistic cases, however, we have an additional length scale, namely $\hbar \beta$ (physically speaking, the Compton wave length of the particles under consideration). As explained above, modular symmetry interchanges these two parameters, which is why it is convenient to work with two equivalent length scales~$a_{\pm}$ in Section~2, cf.~\eqref{aa}. Now in the hyperbolic and nonperiodic Toda regimes we also have a notion of dual system, with the `spectral variables' $\hat{p}_1,\ldots,\hat{p}_N$ denoting asymptotic momenta. The self-duality of the relativistic  hyperbolic regime, however, makes it more natural to work with the dual position $v$ defined by \eqref{phsub}. (This has in particular the consequence that asymptotic plane waves do not have the usual dimensionless combination $x\cdot \hat{p}/\hbar$ in the exponent, but  $x\cdot v/a_{+}a_{-}$ instead.)

By contrast to Section~2, we study in Section~3 and Section~4 the classical ($\hbar=0$) and nonrelativistic ($\beta=0$) settings, so that the length scale $\hbar\beta$ disappears. In the elliptic regime we therefore revert to the parameter $\alpha$, whereas in the classical and in the nonrelativistic hyperbolic and Toda cases we trade $a_{+}=\alpha$ for a parameter
\beq\label{mu}
\mu = 2\pi/\alpha,
\eeq
with dimension [position]$^{-1}$. This change not only avoids a plethora of factors $\pi$, but is also in accord with the self-duality of the classical relativistic hyperbolic regime. Indeed, $\mu$ is the parameter naturally dual to $\beta$, as can already be gleaned by comparing the defining Hamiltonians \eqref{Hamiltonians} and the dual ones \eqref{Sdual}. (Cf.~also the Lax matrix~\eqref{Ltau} and dual Lax matrix~\eqref{Ldtau} to appreciate  this self-duality feature.)

\section{Kernel functions}\label{kernelFuncsSection}

\subsection{The elliptic case}\label{Sec21}
In this subsection we review various elliptic quantities that play a role in our study of the periodic Toda case. As explained above, it is convenient to use notation that encodes modular invariance, cf.~\eqref{aa}. To start with, we switch from the $N$ commuting  A$\De$Os $\hat{S}_k$ given by~\eqref{qSk} to the $2N$ commuting Hamiltonians 
\beq\label{Hkp}
	H_{k,\delta}(x) = \sum_{\substack{I\subset\lbrace 1,\ldots,N\rbrace\\ |I|=k}}\prod_{\substack{m\in I\\ n\notin I}}f_{\delta,-}(x_m-x_n)\prod_{m\in I}\exp(-ia_{-\delta}\partial_{x_m})\prod_{\substack{m\in I\\ n\notin I}}f_{\delta,+}(x_m-x_n),
\eeq
where $k=1,\ldots,N$, $ \de=+,-$, and
\begin{equation}
	f_{\delta,\pm}(z) = \left(\frac{s_\delta(z\pm\rho)}{s_\delta(z)}\right)^{1/2},\ \ \ s_\delta(z)= s(r,a_{\de};z).
\end{equation}
(See Appendix~A for the definition and properties of the functions~$s_{\pm}(z)$.) Next, we introduce $2N$ additional A$\De$Os by setting
\beq\label{Hextra}
H_{-k,\de}(x)=H_{k,\de}(-x),\ \ \ \ k=1,\ldots,N,\ \ \ \de=+,-.
\eeq
Thus we have
\beq\label{Hkm}
	H_{-k,\delta}(x) = \sum_{\substack{I\subset\lbrace 1,\ldots,N\rbrace\\ |I|=k}}\prod_{\substack{m\in I\\ n\notin I}}f_{\delta,+}(x_m-x_n)\prod_{m\in I}\exp(ia_{-\delta}\partial_{x_m})\prod_{\substack{m\in I\\ n\notin I}}f_{\delta,-}(x_m-x_n),
\eeq
and in particular
\beq
H_{-N,\de}(x)=H_{N,\de}(-x)=\prod_{m=1}^N\exp(ia_{-\de}\partial_{x_m}),\ \ \ \de=+,-.
\eeq
It is readily verified that the new A$\De$Os are also related to the previous ones via
\beq
H_{-k,\de}(x)=H_{N-k,\de}(x)H_{-N,\de}(x),\ \ \ \ k=1,\ldots,N-1,
\eeq
and when we set
\beq
H_{0,\de}={\bf 1},
\eeq
then this relation holds for $k=N$, too.

The elliptic kernel function $\Psi$ is now of the form
\beq\label{Psi}
\Psi(x,y)=W(x)^{1/2}W(y)^{1/2}\cS(x,y).
\eeq
Here, the weight function is given by
\begin{equation}\label{W}
	W(x) = \frac{1}{C(x)C(-x)},
\end{equation}
with $C$ the generalized Harish-Chandra function
\begin{equation}\label{HCell}
	C(x) = \prod_{1\leq j<k\leq N} \frac{G(x_j-x_k-\rho+ia)}{G(x_j-x_k+ia)}.
\end{equation}
The function $G(z)\equiv G(r,a_+,a_-;z)$ is the elliptic gamma function reviewed in Appendix~\ref{gammaFuncsAppendix}, and the notation
\beq
a=(a_++a_-)/2
\eeq
is used throughout Section~2. Also, the special function $\cS$ is defined by
\begin{equation}\label{defcS}
	\cS(x,y) = \prod_{j,k=1}^N\frac{G(x_j-y_k-\rho/2)}{G(x_j-y_k+\rho/2)}.
\end{equation}
Note that it satisfies
\beq
\cS(x,y)=\cS(\sigma(x),\tau(y)), \ \ \ \forall \sigma,\tau\in S_N.
\eeq
Moreover, from the reflection equation~\eqref{refl} it follows that
\beq\label{cSsymm}
\cS(x,y)=\cS(y,x)=\cS(-x,-y).
\eeq

We are now prepared to recall the kernel identities. They are given by
\begin{equation}\label{reducedEllipticIds}
	\big(H_{l,\delta}(x) - H_{-l,\delta}(y)\big)\Psi(x,y) = 0,\quad \pm l=1,\ldots,N,\quad \delta=+,-.
\end{equation}
Equivalently, the $4N$ commuting A$\De$Os
\beq
A_{l,\delta}(x) = W(x)^{-1/2}H_{l,\delta}(x)W(x)^{1/2},\ \ \ \pm l=1,\ldots,N,\quad \delta=+,-,
\eeq
 satisfy
 \begin{equation}\label{AS}
	\big(A_{l,\delta}(x) - A_{-l,\delta}(y)\big)\cS(x,y) = 0,\quad \pm l=1,\ldots,N,\quad \delta=+,-.
\end{equation}
Using the analytic difference equations \eqref{ellipticGDiffEq} obeyed by the elliptic gamma function and the formula \eqref{sR} relating $R_{\de}$ and $s_{\de}$, it follows that these A$\De$Os have meromorphic coefficients. Specifically, one readily obtains the explicit formulas
\begin{equation}\label{MeromorphEllipticOps}
	A_{\pm k,\delta}(x) =\sum_{\substack{I\subset\lbrace 1,\ldots,N\rbrace\\ |I|=k}}\prod_{\substack{m\in I\\ n\notin I}}f_{\delta,\mp}(x_m-x_n)^2\prod_{m\in I}\exp(\mp ia_{-\delta}\partial_{x_m}),\ \ \ k=1,\ldots,N, \ \ \de=+,-.
\eeq

For our purposes, it is crucial that the kernel identities~\eqref{AS} are equivalent to the following identities for the functions $s_\delta(z)$:
\begin{multline}\label{EllipticFunctionalEqs}
	\sum_{\substack{I\subset\lbrace 1,\ldots,N\rbrace\\ |I|=k}}\prod_{\substack{m\in I\\ n\notin I}}\frac{s_\delta(x_m-x_n-\rho)}{s_\delta(x_m-x_n)}\prod_{\substack{m\in I\\ n\in\lbrace 1,\ldots,N\rbrace}}\frac{s_\delta(x_m-y_n+\rho)}{s_\delta(x_m-y_n)}\\ = \sum_{\substack{I\subset\lbrace 1,\ldots,N\rbrace\\ |I|=k}}\prod_{\substack{m\in I\\ n\notin I}}\frac{s_\delta(y_m-y_n+\rho)}{s_\delta(y_m-y_n)}\prod_{\substack{m\in I\\ n\in\lbrace 1,\ldots,N\rbrace}}\frac{s_\delta(y_m-x_n-\rho)}{s_\delta(y_m-x_n)}.
\end{multline}
This equivalence can be verified by using once more the equations \eqref{ellipticGDiffEq} and  \eqref{sR}. Further details, as well as a proof of~\eqref{EllipticFunctionalEqs}, can be found in Section 2 of \cite{Rui06}.

As a preparation for our account of the relativistic Toda regimes we introduce additional avatars of the $4N$ commuting A$\De$Os. They can be defined by
\beq\label{cAOps}
\cA_{l,\de}^{\pm}(x)=C(\mp x)^{-1}A_{l,\de}(x)C(\mp x),\ \ \ l\in \{\pm 1,\ldots,\pm N\},\ \ \ \de\in\{+,-\},
\eeq
with $C$ the Harish-Chandra function~\eqref{HCell}, so they have meromorphic coefficients as well.
Alternatively, introducing the elliptic scattering function
\beq\label{ellU}
U(x)=C(x)/C(-x),
\eeq
they are given by 
\beq\label{cAH}
\cA_{l,\de}^{\pm}(x)=U(x)^{\pm 1/2}H_{l,\de}(x)U(x)^{\mp 1/2},\ \ \ l=\pm 1,\ldots,\pm N,\ \ \ \de=+,-.
\eeq
Since we have
\beq\label{Uunit}
|U(x)|=1,\ \ \ x\in\R^N,
\eeq
these operators inherit the formal self-adjointness of the A$\De$Os $H_{l,\de}(x)$. Note that
for $l=\pm N$ the four sets of commuting operators
\beq
\{ H_{l,\de}\},\ \ \{ A_{l,\de}\},\ \ \{ \cA^{+}_{l,\de}\},\ \ \{\cA^{-}_{l,\de}\},\ \ \ l=\pm 1,\ldots,\pm N,\ \ \ \de=+,-,
\eeq
yield the same A$\De$O $\exp(\mp ia_{-\de}\sum_j \partial_j)$.

\subsection{The hyperbolic case}\label{Sec22}

In the hyperbolic limit $r\downarrow 0$ the \adiffops~$H_{\pm k,\delta}$ remain of the same form, but now with $s_\delta(z)\equiv \sinh(\pi z/a_{\delta})$, cf. \eqref{sToSinhLim}. In addition, the limit \eqref{ellTohypgammaLim} from the elliptic to the hyperbolic gamma function implies that the kernel identities \eqref{reducedEllipticIds} hold true if we take $G(z)$ to be the hyperbolic gamma function. All other quantities and relations in the previous subsection have immediate hyperbolic counterparts as well, so we shall not spell them out.

In the hyperbolic case, however, we are also able to obtain kernel identities relating the \adiffops
\beq
A_{\pm k,\delta}(x),\ \ \ \ k=1,\ldots,N,\ \ \ \de=+,-,
\eeq
 in $N$ variables $x=(x_1,\ldots,x_N)$ to the following \adiffops\  in $N-\ell$ variables $y=(y_1,\ldots,y_{N-\ell})$, $\ell=0,\ldots,N$:
 \beq
A_{\mp (k-j),\delta}(y),\ \ \ \ j=0,\ldots,\ell,\ \  0\le k-j\le N-\ell,\ \ \ \ A_{0,\delta}\equiv 1.
\eeq

In the trigonometric case (which we do not consider), the two hyperbolic periods $ia_{+},ia_{-}$ are replaced by one imaginary period $ia$ and a real period $\pi/r$, and accordingly it suffices to consider \adiffops\  $A_1,\ldots, A_N$, with the coefficient building blocks $s_{\de}$ replaced by the sine function. For this case  
  Komori et al.~\cite{KNS09} first arrived at the analogs of the extra kernel relations for $A_1$, by using corresponding functional identities.  Our reasoning below yields relations for arbitrary $k$, whose trigonometric analogs (with the hyperbolic gamma function replaced by the trigonometric one) can be obtained by adapting our hyperbolic arguments.

The relations involve coefficients $c_{\ell,j}^\delta$ with $\ell\in\N, j\in\Z,\de=+,-$, given by
\beq\label{side}
c_{0,0}^\delta = 1,\ \ \ \ c_{\ell,j}^\delta =0,\ \ \ j>\ell,\ \ j<0,
\eeq
and
\beq\label{cS}
c_{\ell,j}^\delta = S_j(e_{\de}((\ell-1)\rho),e_{\de}((\ell-3)\rho),\ldots,e_{\de}(-(\ell-1)\rho)),\ \ \ j=0,\ldots,\ell,
\eeq
where $S_j(a_1,\ldots,a_\ell)$ denotes the $j$th elementary symmetric function of $a_1,\ldots,a_\ell$; also, here and below we use the abbreviation
\beq
e_{\de}(z)\equiv \exp(\pi z/a_{\de}),\ \ \ \ \de=+,-.
\eeq
 Notice that the coefficients are even in $\rho$ and satisfy
\beq
c_{\ell,0}^\delta=c_{\ell,\ell}^\delta =1,\ \ \  c_{\ell,j}^\delta=c_{\ell,\ell-j}^\delta,\ \ \ j=0,\ldots, \ell.
\eeq
Moreover, it is not hard to verify that the coefficients obey a recurrence relation
\begin{equation}\label{cRecursRel}
	c_{\ell+1,j}^\delta = e_\delta(j\rho)c_{\ell,j}^\delta + e_\delta\big((j-1-\ell)\rho\big)c_{\ell,j-1}^\delta,
\end{equation}
and that they are uniquely determined by this recurrence together with the side conditions~\eqref{side}.

With $G(z)$ denoting the hyperbolic gamma function and $A_{\pm k,\delta}$ the hyperbolic version of the elliptic \adiffops~\eqref{MeromorphEllipticOps}, we are now prepared to state and prove the pertinent relations.

\begin{theorem}\label{HypKernelIdProp}
For $\ell=0,1,\ldots,N$, let
\begin{equation}\label{cSl}
	\cS_\ell(x,y) \equiv \prod_{m=1}^N\prod_{n=1}^{N-\ell}\frac{G(x_m-y_n-\rho/2)}{G(x_m-y_n+\rho/2)}, \ \ \ \ \cS_N\equiv 1.
\end{equation}
For any $k\in \{ 1,\ldots,N\}$ and $\tau,\de\in\{ +,-\}$  we have
\beq\label{hyperbolicKernelIds}
	A_{\tau k,\delta}(x_1,\ldots,x_N)\cS_\ell(x,y) = \sum_{j=0}^{\min (k,\ell)} c_{\ell,j}^\delta A_{-\tau (k-j),\delta}(y_1,\ldots,y_{N-\ell})\cS_\ell(x,y),
\eeq
where 
\beq
A_{\pm m,\delta}(y_1,\ldots,y_{N-\ell})\equiv 0,\ \ \ m>N-\ell,\ \ \ A_{0,\delta}\equiv 1,
\eeq
 and where the coefficients $c_{\ell,j}^\delta$ are given by~\eqref{side} and~\eqref{cS}.
\end{theorem}

\begin{proof}
Since we have
\beq\label{cSeven}
\cS_\ell(x,y)=\cS_\ell(-x,-y),
\eeq
it suffices to show~\eqref{hyperbolicKernelIds} for $\tau=+$.
Our proof proceeds by induction on $\ell$. The case $\ell=0$ amounts to the hyperbolic version of~\eqref{AS}, so we now assume~\eqref{hyperbolicKernelIds} for $\ell\ge 0$ and show its validity for $\ell\to\ell+1$.

To this end we begin by deducing from the asymptotics of the hyperbolic gamma function (cf.~\eqref{GLDef} and \eqref{GRLAsA}) that we have
\begin{equation}\label{Slim}
	\lim_{\Lambda\to\infty}\phi(x,y_{N-\ell}+\Lambda)\cS_\ell(x,y_1,\ldots,y_{N-\ell-1},y_{N-\ell}+\Lambda) = \cS_{\ell+1}(x,y),
\end{equation}
where
\begin{equation}
	\phi(x,z) \equiv \exp\left(\frac{i\pi\rho}{a_+a_-}\sum_{m=1}^N(x_m-z)\right).
\end{equation}
In order to exploit this limit, we note
\begin{equation}\label{phiA}
\phi(x,y_{N-\ell})A_{k,\delta}(x) = e_\delta(-k\rho)A_{k,\delta}(x)\phi(x,y_{N-\ell}).
\end{equation}
Furthermore, we split the \adiffops~$A_{-m,\de}(y)$ into two parts, depending on whether the index set $I$ contains $N-\ell$ or not:
\begin{equation}
	A_{-m,\delta}(y) = \sum_{\substack{I\subset\lbrace 1,\ldots,N-\ell-1\rbrace\\ |I|=m}}\Big(\cdots\Big) + \sum_{\substack{I=J\cup\lbrace N-\ell\rbrace\\ J\subset\lbrace 1,\ldots,N-\ell-1\rbrace\\ |J|=m-1}}\Big(\cdots\Big).
\end{equation}
Denoting the first and second sum by $B_{m,\delta}(y)$ and $C_{m,\delta}(y)$, respectively, we then observe that
\bea\label{phiBC}
	\phi(x,y_{N-\ell})B_{m,\delta}(y) & = & B_{m,\delta}(y) \phi(x,y_{N-\ell}), \nonumber \\
	\phi(x,y_{N-\ell})C_{m,\delta}(y) &=  & e_\delta(-N\rho)C_{m,\delta}(y)\phi(x,y_{N-\ell}).
\eea

Next, we multiply both sides of \eqref{hyperbolicKernelIds} by the function $e_{\de}(k\rho)\phi(x,y_{N-\ell})$ and use the commutation relations~\eqref{phiA} and~\eqref{phiBC}. We then take $y_{N-\ell}\to y_{N-\ell}+\Lambda$, and use the readily verified limits
\begin{equation}
\lim_{\Lambda\to\infty}	B_{m,\delta}(y_1,\ldots,y_{N-\ell-1},y_{N-\ell}+\Lambda)= e_\delta(-m\rho)A_{-m,\delta}(y_1,\ldots,y_{N-\ell-1}),
\end{equation}
\bea
\lim_{\Lambda\to\infty}	C_{m,\delta}(y_1,\ldots,y_{N-\ell}+\Lambda) & =  & e_\delta((N-\ell-m)\rho) \exp(ia_{-\delta}\partial_{y_{N-\ell}})
\nonumber \\
  &  & \times A_{-m+1,\de}(y_1,\ldots,y_{N-\ell-1}),
\eea
together with the limit~\eqref{Slim}.
If we now take $j\to j-1$ in the sum coming from $C_{k-j,\de}(y)$, then it becomes clear that the coefficients of the A$\De$Os $A_{-(k-j),\de}$ on the right-hand side satisfy~\eqref{cRecursRel}, so that they are given by~\eqref{cS} with $\ell\to\ell+1$.
\end{proof}

We proceed to detail three specializations of Theorem \ref{HypKernelIdProp}. For the first we fix $k=1$, but impose no restrictions on $\ell$. From~\eqref{cS} we obtain
\beq
c^\delta_{\ell,0}=1,\ \ \ \ c^\delta_{\ell,1}=s_\delta(\ell\rho)/s_\delta(\rho),
\eeq
yielding the following special cases.
\begin{corollary}
For $\ell=0,\ldots,N$, we have
\begin{equation}\label{cor1}
	\big(A_{\pm 1,\delta}(x_1,\ldots,x_N) - A_{\mp 1,\delta}(y_1,\ldots,y_{N-\ell})\big)\cS_\ell(x,y) = \frac{s_\delta(\ell\rho)}{s_\delta(\rho)}\cS_\ell(x,y).
\end{equation}
\end{corollary}
We note that this corollary is the hyperbolic analog of Statement (1) in Theorem 2.2 of \cite{KNS09}. 

Next, we require $\ell=1$, but do not restrict $k$.  From~\eqref{cS} we have $c_{1,0}^\delta = c_{1,1}^\delta = 1$, and hence the following specialization results.

\begin{corollary}
For $k=1,\ldots,N$, we have
\begin{multline}\label{cor2}
	 A_{\pm k,\delta}(x_1,\ldots,x_N)\cS_1(x,y)\\ = \big(A_{\mp k,\delta}(y_1,\ldots,y_{N-1})+A_{\mp (k-1),\delta}(y_1,\ldots,y_{N-1})\big)\cS_1(x,y).
\end{multline}
\end{corollary}

Finally, we choose $\ell=N$, recalling $\cS_N=1$.
\begin{corollary}
The following functional identities hold true:
\beq\label{cor3}
\sum_{\substack{I\subset\lbrace 1,\ldots,N\rbrace\\ |I|=k}}\prod_{\substack{m\in I\\ n\notin I}}\frac{s_\delta(x_m-x_n\pm\rho)}{s_\delta(x_m-x_n)}=c_{N,k}^{\de}.
\eeq
\end{corollary}
These identities were obtained before in the proof of Lemma~A.5 in~\cite{Rui95}.

\subsection{The periodic Toda case}\label{Sec23}
As explained in~\cite{Rui94}, the periodic Toda \adiffops~\eqref{TodaOps} can be obtained as limits of the elliptic \adiffops~\eqref{qSk}. In this section we shall in particular recover this result as a corollary of somewhat more general limit formulas, detailed in~Lemma~2.5. More precisely, these formulas are primarily derived to obtain Toda kernel functions, but they can also be used to show that the $4N$ commuting A$\De$Os $H_{\pm k,\de}(x)$ given by~\eqref{Hkp}--\eqref{Hextra} give rise to $4N$ commuting periodic Toda counterparts, denoted by the same symbols. As it shall transpire, however, the A$\De$Os $A_{\pm k,\de}$~\eqref{MeromorphEllipticOps} have no sensible limits (the $N=2$ case being a curious exception). 

This makes it all the more surprising that  the elliptic functional identities \eqref{EllipticFunctionalEqs} corresponding to the relation between $\cS(x,y)$ and $A_{l,\de}$  (as expressed in~\eqref{AS}) do have Toda limits. Once this limit is obtained, we can easily  identify kernel functions $S^{\pm}(x,y)$ for the periodic Toda system. But here  these functions correspond to A$\De$Os that are the periodic Toda counterparts of the elliptic  operators $\cA^{\pm}_{l,\de}$ given by~\eqref{cAH}. Indeed, as in the elliptic case, they are the similarity transforms of the Toda A$\De$Os $H_{l,\de}$ with a function $U(x)$.  For a suitable choice of parameters this Toda $U$-function is unitary. Moreover, this parameter choice entails that all of the Toda A$\De$Os $H_{l,\de}$ and   $\cA^{\pm}_{l,\de}$ are formally self-adjoint.

Turning to the details, our starting point consists in making the  substitutions
\begin{equation}\label{subxy}
	x_n\to x_n - \frac{n\pi}{Nr}
	,\quad y_n\to y_n -  \frac{n\pi}{Nr}
	,\quad n=1,\ldots,N,
\end{equation}
and
\begin{equation}\label{submu}
	\rho\to\rho +  \frac{\pi}{Nr},
\end{equation}
in the elliptic quantities occurring in~Subsection~2.1. Next we consider the limit $r\to 0$. (Recall the real elliptic period $2\omega$ is parametrized as $\pi/r$.) It seems intractable to control this limit  for the quantities expressed in elliptic gamma functions, and in fact it appears likely that none of them can be renormalized so that this limit exists. Rather, we concentrate on the  functional equations~\eqref{EllipticFunctionalEqs}, which only involve the functions~$s_{\pm }(z)$. 

We first note that the product representation \eqref{sProductRep} for the function $s_{\de}(z)$
 contains an exponential factor $e_\delta(-rz^2/\pi)$. Rather than directly performing the substitutions \eqref{subxy} and~\eqref{submu} in~\eqref{EllipticFunctionalEqs},  we may and shall eliminate these factors  from the start. The point is that the identity \eqref{EllipticFunctionalEqs} still holds true if we switch from $s_\delta$ to the function
\begin{equation}\label{tildesProductRep}
\begin{split}
	\tilde{s}_\delta(z) &:= e_\delta(rz^2/\pi)s_{\de}(z)\\ &= \frac{a_\delta}{\pi}\sinh(\pi z/a_\delta)\prod_{l=1}^\infty\frac{\big(1-e_\delta(2z-2\pi l/r)\big)\big(z\to -z\big)}{\big(1-e_\delta(-2\pi l/r)\big)^2}.
\end{split}
\end{equation}
Indeed, a straightforward computation shows that the exponential factors combine to yield the same overall factor in the left-hand and right-hand side.

Accordingly, we substitute \eqref{subxy} and~\eqref{submu} in \eqref{EllipticFunctionalEqs} with $s_\delta$ replaced by $\tilde{s}_\delta$, and proceed to study the asymptotic behavior as $r\to 0$ of the resulting identity. To this end we focus on the factors $\tilde{s}_\delta(y_m-x_n-\rho)/\tilde{s}_\delta(y_m-x_n)$. Their asymptotics is given by the following lemma, which involves an auxiliary Toda building block
\beq\label{ttoda}
t_{\de}(z)=1-e_{\de}(2z+2\rho).
\eeq

\begin{lemma}\label{tildesLemma}
Let $m,n=1,\ldots,N$, and $m\neq n$. Then we have, as $r\to 0$,
\begin{multline}\label{tildesAsymptotics}
	\frac{\tilde{s}_\delta\left(y_m-x_n-\rho+\frac{\pi}{Nr}(n-m-1)\right)}{\tilde{s}_\delta\left(y_m-x_n+\frac{\pi}{Nr}(n-m)\right)}\\ \sim e_\delta\left(\frac{m-n}{|m-n|}\left(\rho+\frac{\pi}{Nr}\right)\right)\left\lbrace \begin{array}{ll}t_{\de}(x_{m+1}-y_m), & n=m+1~({\rm mod}~N),\\ 1, &{ \rm otherwise}.\end{array}\right.
\end{multline}
Moreover,
\begin{equation}\label{mmasym}
	\frac{\tilde{s}_\delta\left(y_m-x_m-\rho-\frac{\pi}{Nr}\right)}{\tilde{s}_\delta(y_m-x_m)}\sim \frac{e_\delta\left(\rho+\frac{\pi}{Nr}\right)}{t_{\delta}(y_m-x_m-\rho)},\quad r\to 0.
\end{equation}
\end{lemma}

\begin{proof}
We shall infer the statement from the product representation \eqref{tildesProductRep} for the function~$\tilde{s}_\delta$. Dealing first with the sinh-prefactor, we readily find
\begin{multline}\label{shlim}
	\frac{\sinh\frac{\pi}{a_\delta}\left(y_m-x_n-\rho+\frac{\pi}{Nr}(n-m-1)\right)}{\sinh\frac{\pi}{a_\delta}\left(y_m-x_n+\frac{\pi}{Nr}(n-m)\right)}\\ \sim e_\delta\left(\frac{m-n}{|m-n|}\left(\rho+\frac{\pi}{Nr}\right)\right)\left\lbrace\begin{array}{ll} t_{\de}(x_{m+1}-y_m), & n=m+1,\\ 1, & \text{otherwise},\end{array}\right.
\end{multline}
as $r\to 0$.
Next, we consider the terms arising from the infinite product in \eqref{tildesProductRep}. Since $|n-m|<N$ and $l\geq 1$, we have
\begin{equation}
	\lim_{r\to 0}\frac{1-e_\delta\left(2(y_m-x_n-\rho)+\frac{2\pi}{r}\left(\frac{n-m-1}{N}-l\right)\right)}{1-e_\delta\left(2(y_m-x_n)+\frac{2\pi}{r}\left(\frac{n-m}{N}-l\right)\right)} = 1, 
\end{equation}
and
\begin{multline}
	\lim_{r\to 0}\frac{1-e_\delta\left(-2(y_m-x_n-\rho)+\frac{2\pi}{r}\left(-\frac{n-m-1}{N}-l\right)\right)}{1-e_\delta\left(-2(y_m-x_n)+\frac{2\pi}{r}\left(-\frac{n-m}{N}-l\right)\right)}\\ = \left\lbrace\begin{array}{ll} t_{\de}(x_1-y_N), & m=N, n=1, l=1,\\ 1, & \text{otherwise}.\end{array}\right.
\end{multline}
Putting the pieces together, we arrive at the statement for $m\neq n$. The remaining case $m=n$ now follows easily, noting all factors in the infinite product in~\eqref{tildesProductRep} then converge to one.
\end{proof}

It is clear from \eqref{tildesProductRep} that $\tilde{s}_\delta(z)$ is an odd function. As a corollary of~\eqref{tildesAsymptotics}, we thus have 
\begin{multline}\label{asym}
	\frac{\tilde{s}_\delta\left(x_m-y_n+\rho+\frac{\pi}{Nr}(n-m+1)\right)}{\tilde{s}_\delta\left(x_m-y_n+\frac{\pi}{Nr}(n-m)\right)}\\ \sim e_\delta\left(\frac{n-m}{|n-m|}\left(\rho+\frac{\pi}{Nr}\right)\right)\left\lbrace \begin{array}{ll}t_{\de}(x_m-y_{m-1}), & n=m-1~(\text{mod}~N),\\ 1, & \text{otherwise}.\end{array}\right.
\end{multline}
In addition, if we replace $y_m $ by~$x_m$ in~ \eqref{tildesAsymptotics}, and $x_m$ by~$y_m$ in~\eqref{asym}, then we clearly obtain the contribution to the asymptotics due to the remaining factors in~\eqref{EllipticFunctionalEqs}. 

With these asymptotic formulas at our disposal, we are in the position to obtain the following Toda counterpart of the elliptic functional identities \eqref{EllipticFunctionalEqs}.

\begin{lemma}\label{TFE}
For $k=1,\ldots,N$, we have
\begin{multline}\label{TodaFunctionalEqs}
	\sum_{\substack{I\subset\lbrace 1,\ldots,N\rbrace\\ |I|=k}}\prod_{\substack{m\in I\\ m+1\notin I}}t_{\de}(x_{m+1}-x_m)\prod_{m\in I}\frac{t_{\de}(x_m-y_{m-1})}{t_{\de}(y_m-x_m-\rho)}\\ = \sum_{\substack{I\subset\lbrace 1,\ldots,N\rbrace\\ |I|=k}}\prod_{\substack{m\in I\\ m-1\notin I}}t_{\de}(y_m-y_{m-1})\prod_{m\in I}\frac{t_{\de}(x_{m+1}-y_m)}{t_{\de}(y_m-x_m-\rho)},
\end{multline}
with $t_{\de}(z)$ given by~\eqref{ttoda}.
\end{lemma}
\begin{proof}
Just as for~\eqref{EllipticFunctionalEqs}, the special case~$k=N$ of~\eqref{TodaFunctionalEqs} is obvious. Fixing~$k<N$, we substitute \eqref{subxy} and~\eqref{submu} in the identity \eqref{EllipticFunctionalEqs} with $s_{\de}\to\tilde{s}_{\de}$, and then exploit the above asymptotic formulas in the following way. We focus on a term in the sum on the left associated with a fixed subset~$I$.  First, consider a pair of indices $m\in I$ and $n\notin I$. This gives rise to a product of ratios
\beq
P_{mn}=\frac{\tilde{s}_{\de}(x_m-x_n-\rho)}{\tilde{s}_{\de}(x_m-x_n)}\frac{\tilde{s}_{\de}(x_m-y_n+\rho)}{\tilde{s}_{\de}(x_m-y_n)}.
\eeq
With the substitutions in place, we can now use~\eqref{tildesAsymptotics} for the first and~\eqref{asym} for the second ratio to deduce
\beq
\lim_{r\to 0}P_{mn}= 
\left\lbrace \begin{array}{ll}t_{\de}(x_{m+1}-x_m), & n=m+1~(\text{mod}~N),\\ t_{\de}(x_{m}-y_{m-1}), & n=m-1~(\text{mod}~N),\\1, & \text{otherwise}.\end{array}\right.
\eeq

The remaining pairs of indices $m,n\in I$ yield the product
\beq
\prod_{m,n\in I}\frac{\tilde{s}_\delta(x_m-y_n+\rho)}{\tilde{s}_\delta(x_m-y_n)}=\prod_{m\in I}T_m\prod_{\substack{m,n\in I\\ m>n}}T_{mn},
\eeq
where we have introduced the ratios
\beq
T_m =\frac{\tilde{s}_{\de}(x_m-y_m+\rho)}{\tilde{s}_{\de}(x_m-y_m)},
\eeq
and the product of ratios
\beq
T_{mn} = \frac{\tilde{s}_{\de}(x_m-y_n+\rho)}{\tilde{s}_{\de}(x_m-y_n)}\frac{\tilde{s}_{\de}(x_n-y_m+\rho)}{\tilde{s}_{\de}(x_n-y_m)}.
\eeq
In view of~\eqref{mmasym} we have
\beq
\lim_{r\to 0}e_\delta\left(-\rho-\frac{\pi}{Nr}\right)T_m=1/t_{\delta}(y_m-x_m-\rho).
\eeq
Moreover, using oddness of $\tilde{s}_{\de}$ in the second ratio of $T_{mn}$, we can use~\eqref{tildesAsymptotics} for the second and~\eqref{asym} for the first ratio to get
\beq
\lim_{r\to 0}T_{mn}= 
\left\lbrace \begin{array}{ll} t_{\de}(x_{m}-y_{m-1}), & n=m-1,\\ t_{\de}(x_1-y_N), & m=N,n=1,\\1, & \text{otherwise}.\end{array}\right.
\eeq

A moment's thought now shows that when we multiply the left-hand side of  \eqref{EllipticFunctionalEqs} with $s_{\de}\to\tilde{s}_{\de}$ by the renormalizing factor
\beq
e_\delta\left(-\rho-\frac{\pi}{Nr}\right)^k,
\eeq
then its $r\to 0$ limit yields the left-hand side of~\eqref{TodaFunctionalEqs}. Proceeding in the same way for the right-hand side of  \eqref{EllipticFunctionalEqs} with $s_{\de}\to\tilde{s}_{\de}$,  
we then deduce~\eqref{TodaFunctionalEqs}.
\end{proof}

At this point we invoke the modified hyperbolic gamma functions $G_R$ and $G_L$, cf.~Appendix \ref{gammaFuncsAppendix}. Indeed, the identities \eqref{TodaFunctionalEqs}, combined with the difference equations \eqref{GRDE} and~\eqref{GLDE} satisfied by $G_R$ and $G_L$, resp., can now be used to obtain kernel functions for the periodic Toda system. Specifically, it is readily deduced from these formulas that we have
\begin{equation}\label{idAp}
	\prod_{m\in I}\frac{t_{\delta}(x_m-y_{m-1})}{t_{\delta}(y_m-x_m-\rho)} = A^{\pm}(x,y)^{-1}\prod_{m\in I}\exp(\mp ia_{-\delta}\partial_{x_m})A^{\pm}(x,y),
\end{equation}
and 
\begin{equation}\label{idAm}
	\prod_{m\in I}\frac{t_{\delta}(x_{m+1}-y_m)}{t_{\delta}(y_m-x_m-\rho)} = A^{\pm}(x,y)^{-1}\prod_{m\in I}\exp(\pm ia_{-\delta}\partial_{y_m})A^{\pm}(x,y),
\end{equation}
where we have introduced the auxiliary functions
\begin{equation}\label{Auxp}
	A^{+}(x,y) = \prod_{m=1}^N\frac{G_R(y_{m}-x_{m+1}-ia-\rho)}{G_L(y_m-x_m-ia)},
\end{equation}
\begin{equation}\label{Auxm}
	A^{-}(x,y) = \prod_{m=1}^N\frac{G_L(y_{m}-x_m+ia)}{G_R(y_{m}-x_{m+1}+ia-\rho)}.
\end{equation}
These functions are basically the kernel functions we need. To detail this, we should first introduce various additional Toda quantities.

To begin with, we have thus far retained the coupling parameter~$\rho$ of the elliptic and hyperbolic regimes, since this yields the simplest auxiliary quantities. At this stage, however, we need to switch to parameters that are more appropriate for the Toda regimes. First, we introduce four Toda  interaction functions
\beq\label{Tfn}
T^{\pm}_{\de}(z)=1+e_{\de}(2z\pm ia_{-\de}+2\eta),\ \ \ \de=+,-,\ \ \ \eta\in\R.
\eeq
They are obtained from the four functions~$t_{\de}(z)$ and $t_{\de}(z-ia_{-\de})$ when $\rho$ is replaced by $ ia +\eta$.
(The function~$T^{+}_{+}(z)$ amounts to the interaction function~\eqref{qfT} of the Introduction, cf.~\eqref{aa}.)
The real parameter~$\eta$ plays the role of coupling constant.  We also introduce the Toda $U$-function
\beq\label{TU}
U(x)=\prod_{m=1}^N\frac{1}{G_L(x_{m+1}-x_m+\eta)}.
\eeq
It has the unitarity property
\beq\label{Uun}
|U(x)|=1,\ \ \ \ x\in\R^N,\ \ \ \eta\in\R,
\eeq
and it is invariant under cyclic permutations of $x_1,\ldots,x_N$.

Next, we define a set of $4N$ formally self-adjoint \adiffops
\beq\label{cApp}
	\cA^{+}_{k,\delta}(x) = \sum_{|I|=k}\prod_{\substack{m\in I\\ m+1\notin I}}T^{+}_{\de}(x_{m+1}-x_m)\prod_{m\in I}\exp(-ia_{-\delta}\partial_{x_m}),
	\eeq
\beq\label{cApm}
	\cA^{+}_{-k,\delta}(x) = \sum_{|I|=k}\prod_{\substack{m\in I\\ m-1\notin I}}T^{+}_{\delta}(x_m-x_{m- 1})\prod_{m\in I}\exp(ia_{-\delta}\partial_{x_m}),
\eeq
and a second set of $4N$ formally self-adjoint \adiffops
\beq\label{cAmp}
	\cA^{-}_{k,\delta}(x) = \sum_{|I|=k}\prod_{\substack{m\in I\\ m-1\notin I}}T^{-}_{\de}(x_{m}-x_{m-1})\prod_{m\in I}\exp(-ia_{-\delta}\partial_{x_m}),
	\eeq
\beq\label{cAmm}
	\cA^{-}_{-k,\delta}(x) = \sum_{|I|=k}\prod_{\substack{m\in I\\ m+1\notin I}}T^{-}_{\delta}(x_{m+1}-x_{m})\prod_{m\in I}\exp(ia_{-\delta}\partial_{x_m}).
\eeq
It is easy to verify that these definitions
  entail the relations
\beq\label{cAN}
\cA^{\tau}_{-N,\de}(x)=\cA^{\tau}_{N,\de}(x)^{-1}=\prod_{m=1}^N\exp(ia_{-\delta}\partial_{x_m}),
\eeq
\beq\label{cArel}
\cA^{\tau}_{-k,\de}(x)=\cA^{\tau}_{N-k,\de}(x)\cA^{\tau}_{-N,\de}(x),\ \ \ k=1,\ldots,N-1,
\eeq
where $\tau,\de=+,-$. Moreover, from the difference equations~\eqref{GLDE} obeyed by~$G_L$ we obtain
\beq\label{cAU}
\cA^{+}_{l,\de}(x) =U(x)\cA^{-}_{l,\de}(x)U(x)^{-1},\ \ \ \ \pm l= 1,\ldots, N,\ \ \ \de=+,-.
\eeq

Finally, we introduce a third set of $4N$ formally self-adjoint \adiffops, namely,
\beq\label{TH}
H_{l,\de}(x)=U(x)^{1/2}\cA^{-}_{l,\de}(x)U(x)^{-1/2}=U(x)^{-1/2}\cA^{+}_{l,\de}(x)U(x)^{1/2},\ \ \pm l= 1,\ldots, N,\ \ \de=+,-.
\eeq
These operators are the quantum counterparts of the classical Hamiltonians given by~\eqref{Hamiltonians}, \eqref{VIT} and ~\eqref{fT}, with modular invariance and formal self-adjointness taken into account. Explicitly, letting $k=1,\ldots,N,$ and $\de=+,-$, they read
\beq\label{HTp}
H_{k,\delta}(x) = \sum_{|I|=k}\prod_{\substack{m\in I\\ m+1\notin I}}T^{+}_{\de}(x_{m+1}-x_m)^{1/2}\prod_{m\in I}e^{-ia_{-\delta}\partial_{x_m}}\prod_{\substack{m\in I\\ m-1\notin I}}T^{+}_{\de}(x_{m}-x_{m-1})^{1/2},
	\eeq
\beq\label{HTm}
	H_{-k,\delta}(x) = \sum_{|I|=k}\prod_{\substack{m\in I\\ m+1\notin I}}T^{-}_{\delta}(x_{m+1}-x_{m})^{1/2}\prod_{m\in I}e^{ia_{-\delta}\partial_{x_m}}\prod_{\substack{m\in I\\ m-1\notin I}}T^{-}_{\de}(x_{m}-x_{m-1})^{1/2},
\eeq
as is readily checked by using~\eqref{GLDE} once again. It is also clear that they satisfy~\eqref{cAN} and \eqref{cArel} with $\cA^{\tau}$ replaced by $H$. Furthermore, like the Toda A$\De$Os $\cA^{\tau}_{l,\de}$, they are invariant under cyclic permutations.

Last but not least, the $4N$ operators $H_{l,\de}$ mutually commute, so that the sets of operators $\{ \cA^{+}_{l,\de}\}$ and $\{ \cA^{-}_{l,\de}\}$ consist of mutually commuting operators as well. As explained in the Introduction, this assertion follows from~\cite{Rui90}, where a direct proof of commutativity can be found. 

On the other hand, taking commutativity of the elliptic Hamiltonians $H_{l,\de}$ in~Subsection~2.1 for granted, the commutativity of the Toda Hamiltonians $H_{l,\de}$ also follows from the latter being limits of the former.
A quick way to check these limits within the present context is as follows. First, push the third product in~\eqref{Hkp} through the shifts. Then replace $s_{\de}$ by $\tilde{s}_{\de}$, which amounts to a multiplicative renormalization. Now substitute~\eqref{subxy} and~\eqref{submu} and use \eqref{tildesAsymptotics} and \eqref{asym} (with $y$ replaced by $x$) to see that the $r\to 0$ limit yields \eqref{HTp}.

The reader who has verified these steps will easily see why this procedure fails for the A$\De$Os $A_{k,\de}(x)$ with $k<N$, unless $N=2$. The point is that  different subsets $I$ in \eqref{MeromorphEllipticOps} give rise to different powers of the factor $e_{\de}(\pi/Nr)$, so that  no nontrivial $r\to 0$ limit can be obtained by a multiplicative renormalization.

Possibly,  the substitutions \eqref{subxy} and \eqref{submu}  in the elliptic $\cA^{\tau}_{l,\de}$, along with a suitable renormalization and similarity transformation, yield A$\De$Os that converge to  the periodic Toda $\cA^{\tau}_{l,\de}$ for $r\to 0$. At any rate, the analogous substitutions in the hyperbolic $\cA^{\tau}_{l,\de}$ do yield the nonperiodic Toda $\cA^{\tau}_{l,\de}$ as limits, cf.~the end of the next subsection. We have anticipated this state of affairs in the notation we have adopted above. 

We are now prepared for the main result of this subsection.

\begin{theorem}\label{ToTheor}
Let $l\in\{ \pm 1,\ldots,\pm N\}$, and $\de\in\{ +,-\}$. We have kernel function identities 
\begin{equation}\label{periodicTidp}
	\left(\cA^{+}_{l,\delta}(x) - \cA^{+}_{-l,\delta}(y)\right)S^{+}(x,y) = 0,
\end{equation}
\begin{equation}\label{periodicTidm}
	\left(\cA^{-}_{l,\delta}(x) - \cA^{-}_{-l,\delta}(y)\right)S^{-}(x,y) = 0,
\end{equation}
\begin{equation}\label{periodicTHidp}
	\left(H_{l,\delta}(x) - H_{-l,\delta}(y)\right)U(x)^{-1/2}U(y)^{-1/2}S^{+}(x,y) = 0,
\end{equation}
\begin{equation}\label{periodicTHidm}
	\left(H_{l,\delta}(x) - H_{-l,\delta}(y)\right)U(x)^{1/2}U(y)^{1/2}S^{-}(x,y) = 0,
\end{equation}
where
\beq\label{Sp}
S^{+}(x,y)= \prod_{m=1}^N\frac{G_R(y_{m}-x_{m+1}-ia/2-\eta/2)}{G_L(y_m-x_m+ia/2+\eta/2)},
\eeq
\beq\label{Sm}
S^{-}(x,y)= \prod_{m=1}^N\frac{G_L(y_{m}-x_m-ia/2+\eta/2)}{G_R(y_{m}-x_{m+1}+ia/2-\eta/2)},
\eeq
and $U(x)$ is given by \eqref{TU}. Furthermore, the identities \eqref{periodicTidp}--\eqref{periodicTHidm} still hold when the functions $S^{\pm}(x,y)$ are replaced by $S^{\pm}(y,x)$ or by $S^{\pm}(\sigma(x),y)$, with $\sigma$ any cyclic permutation.
\end{theorem}
\begin{proof} Combining the functional equations \eqref{TodaFunctionalEqs}
with~\eqref{idAp} and \eqref{idAm}, we obtain
\begin{multline}
\sum_{\substack{I\subset\lbrace 1,\ldots,N\rbrace\\ |I|=k}}\prod_{\substack{m\in I\\ m+1\notin I}}t_{\de}(x_{m+1}-x_m)\prod_{m\in I}\exp(-ia_{-\delta}\partial_{x_m})A^{+}(x,y)\\
=\sum_{\substack{I\subset\lbrace 1,\ldots,N\rbrace\\ |I|=k}}
\prod_{\substack{m\in I\\ m-1\notin I}}t_{\de}(y_m-y_{m-1})
\prod_{m\in I}\exp(ia_{-\delta}\partial_{y_m})A^{+}(x,y).
\end{multline}
If we now replace $\rho$ by $ia+\eta$ and then shift $y_1,\ldots,y_N$ by $3ia/2+\eta/2$, then we obtain \eqref{periodicTidp} for $l=k$.
For $l=-k$, we can use~\eqref{cAN}--\eqref{cArel}, together with the relation
\beq
\cA^{+}_{-N,\de}(x)S^{+}(x,y) =\cA^{+}_{N,\de}(y)S^{+}(x,y),
\eeq
to complete the proof of \eqref{periodicTidp}.

Likewise, from the identities
\begin{multline}
\sum_{\substack{I\subset\lbrace 1,\ldots,N\rbrace\\ |I|=k}}\prod_{\substack{m\in I\\ m+1\notin I}}t_{\de}(x_{m+1}-x_m)\prod_{m\in I}\exp(ia_{-\delta}\partial_{x_m})A^{-}(x,y)\\
=\sum_{\substack{I\subset\lbrace 1,\ldots,N\rbrace\\ |I|=k}}
\prod_{\substack{m\in I\\ m-1\notin I}}t_{\de}(y_m-y_{m-1})
\prod_{m\in I}\exp(-ia_{-\delta}\partial_{y_m})A^{-}(x,y),
\end{multline}
we obtain \eqref{periodicTidm} with $l=-k$ upon replacing $\rho$ by $-ia+\eta$ and shifting $y_1,\ldots,y_N$ by $-3ia/2+\eta/2$. Then  \eqref{periodicTidm} with $l=k$ follows as before.

Recalling \eqref{TH}, we now obtain \eqref{periodicTHidp}--\eqref{periodicTHidm} from \eqref{periodicTidp}--\eqref{periodicTidm}. Also, the last statement follows from invariance of the A$\De$Os and $U(x)$ under cyclic permutations.
\end{proof}

In the Introduction we have already pointed out that from a given kernel function we can obtain an infinity of other ones, cf.~the paragraph containing \eqref{ambig}. The kernel functions in Theorem~\ref{ToTheor}, however, are not related to each other by a coordinate translation or multiplication by a factor~\eqref{ambig}. 

It is worth pointing out that limits of translations lead to elementary kernel functions, due to the simple asymptotics of $G_R(z)$ and $G_L(z)$ for $\Re (z)\to \pm \infty$, cf.~\eqref{GRLAsA} and\eqref{GRLAsB}. To be specific, consider the substitution
\beq\label{subK}
y_m \to y_m -ia/2-\eta/2+\Lambda,\ \ \ m=1,\ldots,N,
\eeq
in $S^{+}(x,y)$. For $\Lambda\to\infty$, the dominant asymptotics is given by
\beq
\exp(2iN\chi)\prod_{m=1}^N\exp\left(\frac{i\pi}{a_{+}a_{-}}(y_m-x_m+\Lambda)^2\right).
\eeq
Thus, if we multiply by
a factor
\beq
\exp(-2iN\chi-i\pi N\Lambda^2/a_{+}a_{-})\exp\left(\frac{2i\pi\Lambda}{a_{+}a_{-}}\sum_{m=1}^N(x_m-y_m)\right),
\eeq
(which is of the form~\eqref{ambig}), then we can take $\Lambda\to\infty$ and conclude that the function
\beq\label{KR}
K_R^{+}(x,y)=\exp\left(\frac{i\pi}{a_{+}a_{-}}\sum_{m=1}^N(x_m-y_m)^2\right),
\eeq
is a kernel function for the A$\De$Os $\cA^{+}_{l,\de}$. Likewise, the left asymptotics yields an elementary kernel function
\beq\label{KL}
K_L^{+}(x,y)=\exp\left(\frac{i\pi}{a_{+}a_{-}}\sum_{m=1}^N(x_m-y_{m-1})^2\right).
\eeq
Notice that the latter results from $K_R^{+}$ by a cyclic permutation.

In the same way we obtain elementary kernel functions
\beq\label{Km}
K_R^{-}(x,y)=\exp\left(\frac{-i\pi}{a_{+}a_{-}}\sum_{m=1}^N(x_m-y_m)^2\right),
K_L^{-}(x,y)=\exp\left(\frac{-i\pi}{a_{+}a_{-}}\sum_{m=1}^N(x_m-y_{m-1})^2\right),
\eeq
for the A$\De$Os $\cA^{-}_{l,\de}$.

\subsection{The nonperiodic  Toda case}\label{Sec24}
In this subsection we first deduce kernel functions for the nonperiodic Toda system by limit transitions from the periodic case. We then take further limits to obtain kernel identities that relate \adiffops~whose number of variables differs by one. Finally, we discuss  the direct limit from the hyperbolic  quantities
to their nonperiodic Toda counterparts, which is quite easily understood for the A$\De$Os $\cA^{\tau}_{l,\de}(x)$ as well.
As a bonus, we obtain the nonperiodic Toda $U$-function as a limit of the hyperbolic one.
  Throughout this subsection we use the same symbols for the nonperiodic quantities as for their periodic counterparts.

We start by observing that if we perform the substitutions
\begin{equation}\label{subpT}
	x_n\to x_n + n\Lambda,\ \ \ n=1,\ldots,N,\quad \eta\to\eta - \Lambda,
\end{equation}
in the \adiffops~$\cA^\tau_{l,\delta}(x)$,  then only the factor $T^{\tau}_{\delta}(x_1-x_N)$ is  affected, and it converges to $1$ in the limit $\Lambda\to\infty$, cf.~\eqref{Tfn}. Since we retain the notation $\cA^\tau_{l,\delta}(x)$ for the \adiffops~thus obtained, the equations \eqref{cApp}--\eqref{cAmm} still hold true, but now with the nonperiodic Toda convention $x_0 = -x_{N+1} = \infty$. Likewise, on account of the $G_L$-asymptotics \eqref{GRLAsA}, we have the relation \eqref{cAU} between $\cA^+_{l,\delta}(x)$ and $\cA^{-}_{l,\delta}(x)$, but now with
\beq\label{TUnp}
U(x)=\prod_{m=1}^{N-1}\frac{1}{G_L(x_{m+1}-x_m+\eta)}.
\eeq
Finally,  we have nonperiodic Hamiltonians $H_{l,\de}(x)$ given by \eqref{TH}--\eqref{HTm}. It should be noted that the nonperiodic A$\De$Os are no longer invariant under cyclic permutations of the coordinates $x_1,\ldots,x_N$.

Next we discuss the effect of substituting first
\beq\label{suby}
y_n\to y_n-\eta/2 + n\Lambda,\ \ \ n=1,\ldots,N,
\eeq
and then \eqref{subpT}, on the kernel functions $S^{\pm}(x,y)$ given by \eqref{Sp}--\eqref{Sm}.
 Clearly, the resulting functions depend on $\Lambda$ only via the factors $G_R(y_N-x_1\mp ia/2-\eta+N\Lambda)$. It follows from \eqref{GRLAsA} that if we take $\Lambda\to\infty$, then these factors converge to 1. Shifting next $y_n$ back by $\eta/2$, we wind up with kernel functions
\beq\label{Spl}
S^{+}(x,y)=\frac{1}{G_L(y_N-x_N+ia/2+\eta/2)} \prod_{m=1}^{N-1}\frac{G_R(y_{m}-x_{m+1}-ia/2-\eta/2)}{G_L(y_m-x_m+ia/2+\eta/2)},
\eeq
\beq\label{Sml}
S^{-}(x,y)=G_L(y_N-x_N-ia/2+\eta/2) \prod_{m=1}^{N-1}\frac{G_L(y_{m}-x_m-ia/2+\eta/2)}{G_R(y_{m}-x_{m+1}+ia/2-\eta/2)}.
\eeq
We thus arrive at the following analog of Theorem \ref{ToTheor} for the nonperiodic Toda system.

\begin{theorem}\label{openTodaKernelProp}
With the convention
\beq\label{npTconv}
x_0=y_0 =\infty,\ \ \  x_{N+1} =y_{N+1}=- \infty,
\eeq
in force, the identities \eqref{periodicTidp}--\eqref{periodicTHidm} hold true for $S^{\pm}$ given by \eqref{Spl}--\eqref{Sml} and $U$ by~\eqref{TUnp}.
\end{theorem}

Just as in the periodic Toda case, the substitution \eqref{subK} and subsequent limit $\Lambda\to\infty$ lead to the elementary kernel identities
\beq\label{KRid}
\left(\cA^{\tau}_{l,\delta}(x) - \cA^{\tau}_{-l,\delta}(y)\right)K_R^{\tau}(x,y) = 0,\ \ \pm l=1,\ldots,N,\ \ \ \de,\tau=+,-,
\eeq
where
\beq\label{KRtau}
K_R^{\pm}(x,y)=\exp\left(\frac{\pm i\pi}{a_{+}a_{-}}\sum_{m=1}^N(x_m-y_m)^2\right).
\eeq

By contrast, in this case the limit $\Lambda\to -\infty$ does not yield the kernel functions $K_L^{\pm}$ defined by \eqref{KL} and \eqref{Km}, due to  the `missing' $G_R$-factor. In fact, since the nonperiodic A$\De$Os are not invariant under cyclic permutations, one should not expect that \eqref{KRid} also holds for $K_L^{\tau}$.

We proceed to a more significant difference between the periodic and nonperiodic Toda systems: from Theorem \ref{openTodaKernelProp} we are able to deduce kernel identities that relate the \adiffops~$A^{\tau}_{l,\delta}$ in $N$ variables $x = (x_1,\ldots,x_N)$ to  \adiffops~$A^{\tau}_{l',\delta}$ in $N-1$ variables $y = (y_1,\ldots,y_{N-1})$.

\begin{corollary}
Letting $k \in\{  1,\ldots, N\}$, and $\de,\sigma,\tau \in\{ +,-\}$, we have 
\begin{multline}\label{AS1}
	\cA_{\sigma k,\delta}^{\tau}(x_1,\ldots,x_N)S_1^{\tau}(x,y)\\ = \left(\cA_{-\sigma k,\delta}^{\tau}(y_1,\ldots,y_{N-1}) + \cA^{\tau}_{-\sigma (k-1),\delta}(y_1,\ldots,y_{N-1})\right)S_1^{\tau}(x,y),
\end{multline}
where 
\beq
\cA^{\tau}_{ \pm N,\delta}(y_1,\ldots,y_{N-1})\equiv 0,\ \ \ \cA^{\tau}_{0,\delta}\equiv 1,
\eeq
and
\beq\label{SpN}
S_1^{+}(x,y)= \prod_{m=1}^{N-1}\frac{G_R(y_{m}-x_{m+1}-ia/2-\eta/2)}{G_L(y_m-x_m+ia/2+\eta/2)},
\eeq
\beq\label{SmN}
S_1^{-}(x,y)= \prod_{m=1}^{N-1}\frac{G_L(y_{m}-x_m-ia/2+\eta/2)}{G_R(y_{m}-x_{m+1}+ia/2-\eta/2)}.
\eeq

\end{corollary}

\begin{proof}
Clearly, \eqref{GRLAsA} entails that if we substitute $y_N\to y_N-\Lambda$ in \eqref{Spl} and \eqref{Sml} and let $\Lambda\to\infty$, then we obtain $S_1^{\pm}(x,y)$. In order to determine the same limit for the \adiffop
\beq
\cA_{-\sigma k,\delta}^{\tau}(y_1,\ldots,y_N),
\eeq
we split the sum in the pertinent formula among \eqref{cApp}--\eqref{cAmm} into a sum over subsets $I$ containing the index $N$ and a second sum over $I$ not containing $N$. Now there are two cases to consider.

(1)( $\sigma=\tau$) If $k=N$, then the second sum is empty, whereas for $k<N$ it equals $\cA_{-\tau k,\delta}^{\tau}(y_1,\ldots,y_{N-1})$, cf.~\eqref{cApm}--\eqref{cAmp}. This \adiffop~is independent of $y_N$, so it is invariant under the limit at issue. On the other hand, if we perform the substitution in the first sum and let $\Lambda\to\infty$, then we obtain the A$\De$O
\beq
 \cA_{-\tau (k-1),\delta}^{\tau}(y_1,\ldots,y_{N-1})\exp(\tau ia_{- \delta}\partial_{y_N}).
 \eeq
Since $S_1^{\tau}(x,y)$ is independent of $y_N$, the \adiffop~$\exp (\tau ia_{-\delta}\partial_{y_N})$ acts as the identity on $S_1^{\tau}(x,y)$, and so we arrive at \eqref{AS1}.

(2) ($\sigma=-\tau$) Arguing as before, we see from \eqref{cApp} and \eqref{cAmm} that in this case the roles of the first and second sum are reversed, i.e., they now yield the first and second A$\De$O on the right-hand side of \eqref{AS1}, respectively.
\end{proof}

It is worth pointing out that we cannot repeat the above limit procedure for $S_1^{\tau}(x,y)$ and the variable $y_{N-1}$, so as to obtain a kernel function $S_2^{\tau}(x,y)$ that relates \adiffops~in $N$ and $N-2$ variables. Indeed, the asymptotics \eqref{GRLAsB} of $G_R$ and $G_L$ implies that if we were to renormalize $S_1^{\tau}(x,y)$ so as to obtain a finite limit, then the resulting function would contain an exponential factor that still depends on $y_{N-1}$. Likewise, it seems impossible to obtain analogs of $S_1^{\pm}(x,y)$ for the periodic Toda case via a limit of the kernel functions $S^{\pm}(x,y)$ given by \eqref{Sp}--\eqref{Sm}.

To conclude this subsection, we discuss an alternative way to arrive at the above nonperiodic Toda results, namely, via the hyperbolic quantities. The pertinent limits are far more easily controled than the elliptic to periodic Toda limits, and yield additional insights. First, we substitute
\beq\label{hypsub}
x_n\to x_n-n\Lambda,\ \ \  y_n\to y_n-n\Lambda,\ \ \ \rho\to\rho +\Lambda,
\eeq
in the hyperbolic version of the functional equations \eqref{EllipticFunctionalEqs} (obtained upon replacing $s_{\de}(z)$ by $\sinh(\pi z/a_{\de})$). The asymptotic behavior  is then given by
\beq
	\frac{\sinh\frac{\pi}{a_\delta}\left(y_m-x_m-\rho-\Lambda\right)}{\sinh\frac{\pi}{a_\delta}(y_m-x_m)}\sim \frac{e_\delta\left(\rho+\Lambda\right)}{t_{\delta}(y_m-x_m-\rho)},\quad \Lambda\to \infty,
\eeq
and by  \eqref{shlim} with $\pi/Nr$ replaced by $\Lambda$ (recall $t_{\de}$ is defined by \eqref{ttoda}).

Next, we use these formulas in the same way as in Subsection~2.3 to arrive at the functional equations \eqref{TodaFunctionalEqs} with the nonperiodic Toda convention \eqref{npTconv} in effect. More specifically, in the  proof of Lemma~\ref{TFE} we need only replace $\tilde{s}_{\de}(z)$ by $\sinh(\pi z/a_{\de})$, and the $r\to 0$ limit by the $\Lambda \to\infty$ limit. 

The developments below Lemma~\ref{TFE} can now be followed verbatim, with the convention \eqref{npTconv} ensuring that we get the same quantities as obtained via limits of the periodic Toda regime. But now we can also control the limits of the hyperbolic A$\De$Os and $U$-function.

Indeed, using the difference equations \eqref{hyperbolicgammaDiffEq} obeyed by the hyperbolic gamma function, we can readily calculate the hyperbolic A$\De$Os $\cA^{\tau}_{l,\de}$ explicitly from
\eqref{HCell}, \eqref{MeromorphEllipticOps} and \eqref{cAOps}. This yields
\beq\label{cAhyp}
\cA^+_{ k,\delta}(x) =\sum_{ |I|=k}\prod_{\substack{m\in I, n\notin I\\n>m}}f_{\delta,-}(x_m-x_n)^2f_{\delta,+}(x_m-x_n-ia_{-\de})^2\prod_{m\in I}\exp(- ia_{-\delta}\partial_{x_m}),
\eeq
\beq
\cA^+_{ -k,\delta}(x) =\sum_{ |I|=k}\prod_{\substack{m\in I, n\notin I\\n<m}}f_{\delta,+}(x_m-x_n)^2f_{\delta,-}(x_m-x_n+ia_{-\de})^2\prod_{m\in I}\exp(ia_{-\delta}\partial_{x_m}),
\eeq
\beq
\cA^{-}_{ k,\delta}(x) =\sum_{ |I|=k}\prod_{\substack{m\in I, n\notin I\\n<m}}f_{\delta,-}(x_m-x_n)^2f_{\delta,+}(x_m-x_n-ia_{-\de})^2\prod_{m\in I}\exp(- ia_{-\delta}\partial_{x_m}),
\eeq
\beq
\cA^{-}_{ -k,\delta}(x) =\sum_{ |I|=k}\prod_{\substack{m\in I, n\notin I\\n>m}}f_{\delta,+}(x_m-x_n)^2f_{\delta,-}(x_m-x_n+ia_{-\de})^2\prod_{m\in I}\exp( ia_{-\delta}\partial_{x_m}),
\eeq
where $k=1,\ldots,N$,  $\de=+,-$, and
\beq\label{fhyp}
f_{\de,\pm}(z)^2=\frac{\sinh(\pi(z\pm \rho)/a_{\de})}{\sinh(\pi z/a_{\de})}.
\eeq
Now we substitute \eqref{hypsub} and take $\Lambda$ to $\infty$. Denoting the limits by the same symbols, this yields
\beq\label{cAppT}
	\cA^{+}_{k,\delta}(x) = \sum_{|I|=k}\prod_{\substack{m\in I\\ m+1\notin I}}t_{\de}(x_{m+1}-x_m)\prod_{m\in I}\exp(-ia_{-\delta}\partial_{x_m}),
	\eeq
\beq\label{cApmT}
	\cA^{+}_{-k,\delta}(x) = \sum_{|I|=k}\prod_{\substack{m\in I\\ m-1\notin I}}t_{\delta}(x_m-x_{m- 1})\prod_{m\in I}\exp(ia_{-\delta}\partial_{x_m}),
\eeq
\beq\label{cAmpT}
	\cA^{-}_{k,\delta}(x) = \sum_{|I|=k}\prod_{\substack{m\in I\\ m-1\notin I}}t_{\de}(x_{m}-x_{m-1}-ia_{-\de})\prod_{m\in I}\exp(-ia_{-\delta}\partial_{x_m}),
	\eeq
\beq\label{cAmmT}
	\cA^{-}_{-k,\delta}(x) = \sum_{|I|=k}\prod_{\substack{m\in I\\ m+1\notin I}}t_{\delta}(x_{m+1}-x_{m}-ia_{-\de})\prod_{m\in I}\exp(ia_{-\delta}\partial_{x_m}),
\eeq
with $t_{\de}(z)$ given by \eqref{ttoda} and the convention \eqref{npTconv} in effect.

If we now set
\beq\label{muT}
\rho =ia +\eta,
\eeq
then the A$\De$Os \eqref{cAppT}--\eqref{cAmmT} turn into the nonperiodic Toda A$\De$Os already obtained via the periodic regime. In the same way, the hyperbolic Hamiltonians $H_{\pm k,\de}$ converge to their nonperiodic Toda counterparts, whereas the hyperbolic A$\De$Os $A_{\pm k,\de}$ have no sensible limit, a feature shared by the hyperbolic weight function $W(x)$.

The hyperbolic scattering function $U(x)$, however, does have a limit, provided a suitable renormalization is performed. To be specific, when we substitute  \eqref{hypsub} in $U(x)$ (given by \eqref{ellU}
and \eqref{HCell} with $G$ the hyperbolic gamma function), then we obtain via the $G$-asymptotics (which can be gleaned from \eqref{GRDef}, \eqref{GLDef} and \eqref{GRLAsA})
\beq\label{Ulim}
\lim_{\Lambda\to\infty}\cN^{N(N-1)/2}U(x)=\prod_{m=1}^{N-1}\frac{1}{G_L(x_{m+1}-x_m +\rho -ia)},
\eeq
where $\cN$  is the renormalizing constant
\beq
\cN=\exp\left( \frac{i\pi}{a_{+}a_{-}}\big( (\rho -ia +\Lambda)^2+a^2\big) \right).
\eeq
If we now again replace $\rho$ by $ia+\eta$, then the limit function on the right-hand side of \eqref{Ulim} turns into the nonperiodic Toda $U$-function given by \eqref{TUnp}, as announced.

\subsection{The dual nonperiodic Toda case}\label{Sec25}
As recalled  in the Introduction, the hyperbolic relativistic Calogero-Moser system is self-dual. In the limit leading to the nonperiodic Toda dynamics, which we have just discussed, this self-duality property is not preserved. But the dual Toda quantities can be obtained by a similar, but simpler scaling limit. Turning to the details, we substitute
\beq\label{submud}
\rho \to \rho +\Lambda,
\eeq
and study the limit $\Lambda\to \infty$. As explained in the paragraph containing~\eqref{phsub}, this should be done for the hyperbolic quantities expressed
  in the \lq spectral variables\rq\  $v_1,\ldots,v_N$ instead of the 
\lq geometric variables\rq\ $x_1,\ldots,x_N$. 

For the resulting dual hyperbolic A$\De$Os 
\beq\label{Adhyp}
A_{\pm k,\delta}(v) =\sum_{|I|=k}\prod_{\substack{m\in I\\ n\notin I}}\frac{\sinh(\pi(v_m-v_n\mp \rho)/a_{\de})}{\sinh(\pi(v_m-v_n)/a_{\de})}
\prod_{m\in I}\exp(\mp ia_{-\delta}\partial_{v_m}),
\eeq
 we get finite limits by a renormalization. Specifically, we readily deduce
 \begin{multline}\label{Alim}
\lim_{\Lambda\to\infty}e_{\de}(-k(N-k)\Lambda)A_{\pm k,\delta}(v)=\\ e_{\de}(k(N-k)\rho)
\sum_{|I|=k}\prod_{\substack{m\in I\\ n\notin I}}\frac{e_{\de}(\mp (v_m-v_n))}{2s_{\de}(\mp (v_m-v_n))}
\prod_{m\in I}\exp(\mp ia_{-\delta}\partial_{v_m}).
\end{multline}
Here and from now on we use the notation
\beq
s_{\de}(z)=\sinh(\pi z/a_{\de}),\ \ \ \ \de=+,-,
\eeq
which already occurred in Subsection~2.2.

Next, we introduce the product
\beq
P(v)=\prod_{1\le j<k\le N}\exp \left( \frac{i\pi}{2a_{+}a_{-}}(v_j-v_k)^2\right).
\eeq
Using 
\beq
P(v)^{-1}\exp (\mp ia_{-\de}\partial_{v_m})P(v)=\prod_{n\ne m}e_{\de}(\pm (v_m-v_n)-ia_{-\de}/2),
\eeq
it becomes clear that, up to a $k$-dependent multiplicative constant, the limit A$\De$Os on the right-hand side of \eqref{Alim} are the similarity transforms
under $P(v)$ of dual A$\De$Os defined by
\beq\label{defDl}
\hat{A}_{\pm k,\de}(v)=(\mp i)^{k(N-k)}\sum_{ |I|=k}\prod_{\substack{m\in I\\ n\notin I}}\frac{1}{2s_{\de}(v_m-v_n)}
\prod_{m\in I}\exp(\mp ia_{-\delta}\partial_{v_m}).
\eeq
Here we have $k=1,\ldots,N$, $\de=+,-$, and the phase choice will be clarified shortly.

Consider now the kernel function $\cS(v,w)$ given by \eqref{defcS}, with $G$ the hyperbolic gamma function. Shifting $v_j$ by $\rho/2$, substituting \eqref{submud}, and letting $\Lambda\to\infty$, it follows from the $G$-asymptotics that the dominant term is a product of
\beq\label{cSd}
\hat{\cS}(v,w)=\prod_{j,k=1}^NG(v_j-w_k),
\eeq
and $\Lambda$-dependent quadratic exponentials. Likewise, when we shift $v_j$ by $-\rho/2$, substitute \eqref{submud}, and let $\Lambda\to\infty$, we obtain $1/\hat{\cS}(v,w)$ times $\Lambda$-dependent exponentials. Remarkably, when we omit the exponentials, we wind up with kernel functions for the dual A$\De$Os, as shown in the next theorem.

\begin{theorem}\label{dualT}
Letting $l\in\{ \pm 1,\ldots,\pm N\}$, $\de\in\{ +,-\}$ and $\sigma\in\{ \pm 1\}$, we have the dual kernel function identities
\begin{equation}\label{DcS}
	\big(\hat{A}_{l,\delta}(v) - \hat{A}_{-l,\delta}(w)\big)\hat{\cS}(v,w)^{\sigma} = 0.
\end{equation}
\end{theorem}

\begin{proof}
From \eqref{defDl}  we see that
\beq\label{ltominl}
\hat{A}_{l,\de}(-v)=\hat{A}_{-l,\de}(v),
\eeq
whereas the reflection equation for the $G$-function entails
\beq
\hat{\cS}(-v,-w)=1/\hat{\cS}(v,w).
\eeq
Therefore, it suffices to show \eqref{DcS} for $\sigma=1$.

Now it also follows from \eqref{defDl} that
\beq\label{minktoNmink}
\hat{A}_{-k,\de}(v)=\hat{A}_{N-k,\de}(v)\hat{A}_{-N,\de}(v),\ \ \ k=1,\ldots,N.
\eeq
Hence we need only show \eqref{DcS} for $\sigma=1$ and $l=k>0$. To this end, we invoke 
 the difference equations \eqref{hyperbolicgammaDiffEq} for the hyperbolic gamma function to obtain
\begin{align}
	\hat{\cS}(v,w)^{-1}\exp(-ia_{-\delta}\partial_{v_m})\hat{\cS}(v,w) &= \frac{1}{(2i)^N}\prod_{n=1}^N\frac{1}{s_\delta(v_m-w_n-ia)},\\
	\hat{\cS}(v,w)^{-1}\exp(ia_{-\delta}\partial_{w_m})\hat{\cS}(v,w) &= \frac{1}{(2i)^N}\prod_{n=1}^N\frac{1}{s_\delta(v_n-w_m-ia)}.
\end{align}
From this we deduce that we are done if we can prove the functional identities
\begin{multline}\label{sId}
	\sum_{\substack{I\subset\lbrace 1,\ldots,N\rbrace\\ |I|=k}}\prod_{\substack{m\in I\\ n\notin I}}\frac{1}{s_\delta(v_n-v_m)}\prod_{\substack{m\in I\\ n\in\lbrace 1,\ldots,N\rbrace}}\frac{1}{s_\delta(v_m-w_n)}\\ = \sum_{\substack{I\subset\lbrace 1,\ldots,N\rbrace\\ |I|=k}}\prod_{\substack{m\in I\\ n\notin I}}\frac{1}{s_\delta(w_m-w_n)}\prod_{\substack{m\in I\\ n\in\lbrace 1,\ldots,N\rbrace}}\frac{1}{s_\delta(v_n-w_m)}.
\end{multline}

In order to show that these identities are valid, we substitute \eqref{submud} and $x=v,y=w$, in (the hyperbolic version of) \eqref{EllipticFunctionalEqs} and let $\Lambda\to\infty$. Then we obtain from equality of the dominant asymptotics the identities
\begin{multline}\label{sIdaux}
	\sum_{\substack{I\subset\lbrace 1,\ldots,N\rbrace\\ |I|=k}}\prod_{\substack{m\in I\\ n\notin I}}\frac{e_{\de}(v_n-v_m)}{s_\delta(v_n-v_m)}\prod_{\substack{m\in I\\ n\in\lbrace 1,\ldots,N\rbrace}}\frac{e_{\de}(v_m-w_n)}{s_\delta(v_m-w_n)}\\ = \sum_{\substack{I\subset\lbrace 1,\ldots,N\rbrace\\ |I|=k}}\prod_{\substack{m\in I\\ n\notin I}}\frac{e_{\de}(w_m-w_n)}{s_\delta(w_m-w_n)}\prod_{\substack{m\in I\\ n\in\lbrace 1,\ldots,N\rbrace}}\frac{e_{\de}(v_n-w_m)}{s_\delta(v_n-w_m)}.
\end{multline}
Now it is not hard to see that \eqref{sIdaux} is equivalent to \eqref{sId}. Indeed, consider the product of all $e_{\de}$-factors for a given $I$. Taking $I$ equal to $\{ 1,\ldots,k\}$, we get on both sides a factor
\beq
e_{\de}\left(k\sum_{j=1}^N(v_j-w_j)\right).
\eeq
Since this factor is permutation invariant, it does not depend on $I$. Hence we can cancel all exponentials and obtain \eqref{sId}.
\end{proof}

When we make the substitution \eqref{submud} in the dual hyperbolic A$\De$Os $\cA^{\tau}_{l,\de}(v)$ (given by \eqref{cAhyp}--\eqref{fhyp}), then the power of $e_{\de}(\Lambda)$ in the dominant asymptotics of each term in the sum depends on $I$. Hence we cannot renormalize these operators so as to obtain nontrivial finite limits for $\Lambda\to\infty$. 
By contrast, for the dual hyperbolic Hamiltonians $H_{\pm k,\de}(v)$ we obtain
 \begin{multline}\label{Hlim}
\lim_{\Lambda\to\infty}e_{\de}(-k(N-k)\Lambda)H_{\pm k,\delta}(v)=e_{\de}(k(N-k)\rho)\\ \times
\sum_{|I|=k}\prod_{\substack{m\in I\\ n\notin I}}\left(\frac{e_{\de}(\mp (v_m-v_n))}{2s_{\de}(\mp (v_m-v_n))}\right)^{1/2}
\prod_{m\in I}\exp(\mp ia_{-\delta}\partial_{v_m})
\prod_{\substack{m\in I\\ n\notin I}}\left(\frac{e_{\de}(\pm (v_m-v_n))}{2s_{\de}(\pm (v_m-v_n))}\right)^{1/2}.
\end{multline}
Pushing the exponentials through the shifts, we infer that these limits are constant multiples of  dual Toda Hamiltonians formally given by
\beq\label{Hd}
\hat{H}_{\pm k,\de}(v)=
\sum_{|I|=k}\prod_{\substack{m\in I\\ n\notin I}}\left(\frac{1}{2s_{\de}(\mp (v_m-v_n))}\right)^{1/2}
\prod_{m\in I}\exp(\mp ia_{-\delta}\partial_{v_m})
\prod_{\substack{m\in I\\ n\notin I}}\left(\frac{1}{2s_{\de}(\pm (v_m-v_n))}\right)^{1/2}.
\eeq

Due to the square root ambiguity, the phases of the individual terms in the sum are not well defined. To remedy this, we first note that the relevant Hilbert space is $L^2(\hat{G},dv)$, where the dual Toda configuration space is defined by
\beq\label{dTconf}
\hat{G}=\{ v\in\R^N\mid v_N<\cdots <v_1\}.
\eeq
Now we fix the phase ambiguities by defining
\beq\label{Hd2}
\hat{H}_{\pm k,\de}(v)=
\sum_{|I|=k}\prod_{\substack{m\in I\\ n\notin I}}\left|\frac{1}{2s_{\de}(v_m-v_n)}\right|^{1/2}
\prod_{m\in I}\exp(\mp ia_{-\delta}\partial_{v_m})
\prod_{\substack{m\in I\\ n\notin I}}\left|\frac{1}{2s_{\de}(v_m-v_n)}\right|^{1/2}.
\eeq
Then the coefficients are positive and real-analytic on $\hat{G}$.

Next, we define a dual Toda weight function
\beq
\hat{W}(v)=\prod_{1\le j<k\le N}G(v_j-v_k+ia)G(-v_j+v_k+ia).
\eeq
Using the reflection and difference equations for the $G$-function, we obtain
\beq\label{Wd}
\hat{W}(v)=\prod_{1\le j<k\le N}4s_{+}(v_j-v_k)s_{-}(v_j-v_k).
\eeq
Thus $\hat{W}$ is entire in $v$ and positive on $\hat{G}$. Taking positive square roots, it is now not hard to verify that on $\hat{G}$ we have 
\beq\label{DWH}
\hat{A}_{l,\de}(v)=\hat{W}(v)^{-1/2}\hat{H}_{l,\de}(v)\hat{W}(v)^{1/2},\ \ \ \pm l= 1,\ldots,  N,\ \ \ \de=+,-.
\eeq
Indeed, we have chosen the phases in \eqref{defDl} such that these relations hold true. Note that the A$\De$Os $\hat{H}_{l,\de}(v)$~\eqref{Hd2} are formally positive operators on $L^2(\hat{G},dv)$, so that the same is true for the operators $\hat{A}_{l,\de}(v)$ on $L^2(\hat{G},\hat{W}(v)dv)$. Furthermore, we have the following obvious corollary of Theorem~\ref{dualT}.

\begin{corollary}
For any $l\in\{ \pm 1,\ldots,\pm N\}$, $\de\in\{ +,-\}$ and $\sigma\in\{ \pm 1\}$, we have the dual kernel function identities
\begin{equation}\label{HcS}
	\big(\hat{H}_{l,\delta}(v) - \hat{H}_{-l,\delta}(w)\big)\hat{W}(v)^{1/2}\hat{W}(w)^{1/2}\hat{\cS}(v,w)^{\sigma} = 0.
\end{equation}
\end{corollary}

Just as for the hyperbolic case, we have found kernel functions relating dual Toda \adiffops~$\hat{A}_{l,\de}(v)$ in $N$ variables $v_1,\ldots,v_N$ to dual Toda A$\De$Os in $N-\ell$ variables $w_1,\ldots,w_{N-\ell}$ for any $\ell\in\{1,\ldots,N-1\}$. However, in this case there appear to be no kernel functions  for both signs of $\ell$ at once. Another difference with the hyperbolic case is that the kernel identities in the following theorem and its corollary only involve two \adiffops.

\begin{theorem}\label{dTell}
Define kernel functions
\beq
\hat{\cS}^{\tau}_{\ell}(v,w) = \exp \left(\frac{\tau i\pi\ell}{2a_+a_-}\left(\sum_{m=1}^Nv_m^2-\sum_{n=1}^{N-\ell}w_n^2\right)\right)\prod_{m=1}^N\prod_{n=1}^{N-\ell} G(v_m-w_n-ia),
\eeq
where  $\tau=+,-$, and $\ell =0,1,\ldots,N-1$. For any $k\in \{ 1,\ldots,N-\ell\}$ and $ \de\in\{ +,-\}$, we have
\beq\label{dToda2}
\hat{A}_{k,\delta}(v_1,\ldots,v_N)\hat{\cS}^{\tau}_{\ell}(v,w) =\hat{A}_{ -k,\delta}(w_1,\ldots,w_{N-\ell})\hat{\cS}^{\tau}_{\ell}(v,w).
\eeq
\end{theorem}

\begin{proof}
Starting from the $\ell =0$ case \eqref{sId}, induction on $\ell$ readily yields
\begin{multline}\label{sId2}
	\sum_{\substack{I\subset\lbrace 1,\ldots,N\rbrace\\ |I|=k}}\prod_{\substack{m\in I\\ n\notin I}}\frac{1}{s_\delta(v_n-v_m)}\prod_{m\in I}e_\delta(\ell v_m)\prod_{\substack{m\in I\\ n\in\lbrace1,\ldots,N-\ell\rbrace}}\frac{1}{s_\delta(v_m-w_n)}\\ = \sum_{\substack{I\subset\lbrace 1,\ldots,N-\ell\rbrace\\ |I|=k}}\prod_{\substack{m\in I\\ n\notin I}}\frac{1}{s_\delta(w_m-w_n)}\prod_{m\in I}e_\delta(\ell w_m)\prod_{\substack{m\in I\\ n\in\lbrace 1,\ldots,N\rbrace}}\frac{1}{s_\delta(v_n-w_m)}.
\end{multline}
Indeed, assuming \eqref{sId2} for some $\ell\in\{0,1,\ldots,N-2\}$, the case $\ell+1$ follows upon taking $w_{N-\ell}$ to $\infty$.

The crux is now that the difference equations \eqref{hyperbolicgammaDiffEq} for the hyperbolic gamma function imply
\beq\label{Ssm}
	\hat{\cS}^{\tau}_\ell(v,w)^{-1}\exp(-ia_{-\delta}\partial_{v_m})\hat{\cS}^{\tau}_\ell(v,w) = e_\delta\big(\tau\ell (v_m-ia_{-\delta}/2)\big)\left(\frac{i}{2}\right)^{N-\ell}\prod_{n=1}^{N-\ell}\frac{1}{s_\delta(v_m-w_n-ia_{-\delta})},
	\eeq
	\beq\label{Ssp}
	\hat{\cS}^{\tau}_\ell(v,w)^{-1}\exp(ia_{-\delta}\partial_{w_m})\hat{\cS}^{\tau}_\ell(v,w) = e_\delta\big(\tau\ell( w_m+ia_{-\delta}/2)\big)\left(\frac{i}{2}\right)^N\prod_{n=1}^N\frac{1}{s_\delta(v_n-w_m-ia_{-\delta})}.
\eeq
Hence, after shifting the variables $v_m$ to $v_m+ia_{-\delta}/2$ and the variables $w_n$ to $w_n-ia_{-\delta}/2$, it is readily seen that the kernel identity \eqref{dToda2} is equivalent to the functional identity \eqref{sId2} for the case $\tau=+$. Taking $v,w\to -v,-w$ in~\eqref{sId2}, we can pull out the signs from the $s_{\de}$'s so as to obtain the functional identity equivalent to the kernel identity \eqref{dToda2} with $\tau=-$.
\end{proof}

As a corollary, we shall now obtain two more kernel functions~$\tilde{\cS}^{\pm}_\ell(v,w)$. First, since \eqref{dToda2} only involves \adiffops~of the same order $k$, the identities remain valid if we multiply $\hat{\cS}^{\pm}_\ell(v,w)$ by a factor
\beq
\exp\left(\mp\frac{ i\pi}{2a_+a_-}\left(\sum_{m=1}^Nv_m-\sum_{n=1}^{N-\ell}w_n\right)^2\right).
\eeq
By using \eqref{ltominl} we thus obtain
\beq\label{middId}
\hat{A}_{-k,\delta}(v_1,\ldots,v_N)\tilde{\cS}^{\tau}_\ell(v,w) = \hat{A}_{k,\delta}(w_1,\ldots,w_{N-\ell})\tilde{\cS}^{\tau}_\ell(v,w),\ \ \ k=1,\ldots,N-\ell,
\eeq
where
\beq
\begin{split}
\tilde{\cS}^{\tau}_\ell(v,w)\equiv \exp\left(-\frac{\tau i\pi}{2a_+a_-}\left(\sum_{m=1}^Nv_m-\sum_{n=1}^{N-\ell}w_n\right)^2\right)\hat{\cS}^{\tau}_\ell(-v,-w).
\end{split}
\eeq
Moreover, it is readily verified that the additional exponential factor in these kernel functions entails the identity
\beq\label{AtS}
\hat{A}_{N,\delta}(v_1,\ldots,v_N)\tilde{\cS}^{\tau}_{\ell}(v,w) =\hat{A}_{-(N-\ell),\delta}(w_1,\ldots,w_{N-\ell})\tilde{\cS}^{\tau}_{\ell}(v,w).
\eeq
We now act with the \adiffop~$\hat{A}_{N,\delta}(v)$ on \eqref{middId}, and then use~\eqref{AtS} and \eqref{minktoNmink}. Finally, taking $N-k\to k$, we obtain the following corollary of Theorem \ref{dTell}.

\begin{corollary}
We have eigenfunction identities
\beq
\hat{A}_{\ell,\delta}(v_1,\ldots,v_N)\tilde{\cS}^{\tau}_\ell(v,w) = \tilde{\cS}^{\tau}_\ell(v,w),
\eeq
and kernel identities
\beq
\hat{A}_{k,\delta}(v_1,\ldots,v_N)\tilde{\cS}^{\tau}_\ell(v,w) = \hat{A}_{-(k-\ell),\de}(w_1,\ldots,w_{N-\ell})\tilde{\cS}^{\tau}_\ell(v,w),\ \ \ k=\ell+1,\ldots,N.
\eeq
\end{corollary}

We note that the four kernel functions $\hat{\cS}^{\pm}_1(v,w)$ and $\tilde{\cS}^{\pm}_1(v,w)$ are closely related to the function $\mathcal{Q}$ in \cite{KLS02}, which plays an important role in a construction of eigenfunctions of the relativistic periodic and nonperiodic Toda systems, cf.~Eqs.~(3.46), (2.8) and (A.23) in~\cite{KLS02}.

\section{B\"acklund transformations}\label{BacklundSec}

\subsection{The elliptic case}
As explained in Section~1, we are going  to study the \lq expected\rq\ classical asymptotics~\eqref{PsiF} of the elliptic kernel function $\Psi(x,y)$~\eqref{Psi} by introducing dependence on $\hbar$ in a quite special way. It is in fact easy to see how this should be done, because we have started our account in Section~1 with a description of the quantization of the classical systems that involves $\hbar$ explicitly. Specifically, we need only revert from our parametrization of the two positive step sizes $a_{+}$ and $a_{-}$ in the elliptic gamma function to the parameters $\alpha$ and $\hbar\beta$, cf.~\eqref{aa}. Taking $\hbar$ to 0 then amounts to taking $a_{-}$ to 0. (We keep $\beta=1/mc$ fixed, since we wish to stay in the relativistic setting.) Therefore we can study \eqref{PsiF} via the limit~\eqref{Gellcl}. 
Recalling the definitions of the functions $W(z)$ and $\cS(x,y)$ featuring in $\Psi(x,y)$ (cf.~\eqref{W}--\eqref{defcS}), the following lemma easily follows from~\eqref{Gg} and \eqref{Gellcl}.

\begin{lemma}\label{LimLemma}
Let $\rho\in i(0,\alpha)$. For $x$ and $y$ in the elliptic configuration space $G$~\eqref{config}, we have classical limits
\begin{equation}\label{cSlim}
	\lim_{\hbar\downarrow 0}i\hbar\ln \cS(r,\alpha,\hbar\beta;x,y) = \frac{1}{\beta}\sum_{j,k=1}^N \int_{x_j-y_k+\rho/2}^{x_j-y_k-\rho/2}dw\ln R(r,\alpha;w),
\end{equation}
\begin{equation}\label{Wlim}
	\lim_{\hbar\downarrow 0}i\hbar\ln W(r,\alpha,\hbar\beta;x) = \frac{1}{\beta} \sum_{j\neq k}\int^{x_j-x_k+i\alpha/2}_{x_j-x_k+i\alpha/2-\rho}dw\ln R(r,\alpha;w),
\end{equation}
where the integration paths stay away from the cuts given by~\eqref{cuts}.
\end{lemma}

On account of the restrictions $x,y\in G$  and $\rho\in i(0,\alpha)$, we can actually use the representation \eqref{Rrep} for $\ln R$ on the right-hand side of \eqref{cSlim} and \eqref{Wlim}.

To begin with, we now analyze  whether the function $F(x,y)$ resulting from the above limits according to \eqref{PsiF}  gives rise to a B\"acklund transformation. Thus, we study a transformation $B$ from the canonical coordinates $(x,p)\in\Omega$ (with the elliptic phase space $\Omega$ given by~\eqref{Om}--\eqref{config}) to new canonical coordinates $(y,q)$, by taking as generating function
\begin{equation}\label{defF}
	F(x,y) = \frac{1}{\beta}\big(F_W(x) + F_W(y) + F_{\cS}(x,y)\big),
\end{equation}
where
\begin{equation}\label{FW}
	F_W(x) = \frac{1}{2}\sum_{j\neq k}\int^{x_j-x_k+i\alpha/2}_{x_j-x_k+i\alpha/2-\rho}dw\ln R(r,\alpha;w),
\end{equation}
\begin{equation}\label{FS}
	F_{\cS}(x,y) = \sum_{j,k=1}^N \int^{x_j-y_k-\rho/2}_{x_j-y_k+\rho/2}dw\ln R(r,\alpha;w).
\end{equation}
By definition, this means that $y(x,p)$ is to be determined from the equations
\begin{equation}\label{pj}
\begin{split}
	p_j &=- \frac{\partial F}{\partial x_j}\\ &
	= \frac{1}{2\beta}\sum_{k\neq j}\ln\left(\frac{R(x_j-x_k-i\alpha/2)R(x_j-x_k-\rho+i\alpha/2)}{R(x_j-x_k+i\alpha/2)R(x_j-x_k+\rho-i\alpha/2)}\right)\\ &
	\quad + \frac{1}{\beta}\sum_{k=1}^N\ln\left(\frac{R(x_j-y_k+\rho/2)}{R(x_j-y_k-\rho/2)}\right),\ \ \ j=1,\ldots,N,
\end{split}
\end{equation}
and then $q(x,p)$ is given by
\begin{equation}\label{qj}
\begin{split}
	q_j &= \frac{\partial F}{\partial y_j}\\ &
	= \frac{1}{2\beta}\sum_{k\neq j}\ln\left(\frac{R(y_j-y_k+i\alpha/2)R(y_j-y_k+\rho-i\alpha/2)}{R(y_j-y_k-i\alpha/2)R(y_j-y_k-\rho+i\alpha/2)}\right)\\ &\quad + \frac{1}{\beta}\sum_{k=1}^N\ln\left(\frac{R(x_k-y_j+\rho/2)}{R(x_k-y_j-\rho/2)}\right),\ \ \ j=1,\ldots,N.
\end{split}
\end{equation}
(We used evenness of $\ln R(z)$ in these formulas, cf.~\eqref{Rrep}.) 

We have now arrived at the point where we can elaborate on the problems alluded to in the Introduction. Ideally, we would like $B$ to define a bijection on the elliptic phase space $\Omega$. Now whenever an arbitrary real-valued function $F(x,y)$ is used as a generating function for a map on the phase space $\Omega$, local canonicity is clear (at least, when one takes for granted that there exists a local solution $y$ to the implicit equations \eqref{pj}), but in general it will not yield a global symplectomorphism of $\Omega$. For a special choice of $F(x,y)$, therefore, it may be intractable to prove that it does. In the case at hand, however, this question  is easily decided negatively when one retains the parameter restrictions we have imposed: Assuming that for a given $(x,p)\in\Omega$ there exists  a solution $y$ to the system of equations \eqref{pj}, this solution cannot belong to the elliptic configuration space $G$ \eqref{config}. Indeed, assuming $y\in G$, it is immediate from \eqref{pj} that the numbers $p_1,\ldots,p_N$  are purely imaginary, a contradiction. 

To try and get around this snag,  it appears inevitable to require that $\beta$ be purely imaginary instead of positive. (The requirement that $\rho$ be real instead of purely imaginary still leads to momenta that are not real; cf.~also \eqref{Hamiltonians}--\eqref{f} to see why a real $\rho$ is troublesome.) Before analyzing this change in some detail, it is expedient to study first in what sense the map might be a B\"acklund transformation. Reasoning formally (in particular, assuming its existence at least for unspecified parameters and phase space variables), this is readily answered by using \eqref{pj} and \eqref{qj}: These equations do imply the B\"acklund property
\beq\label{SS}
S_k(x,p)=S_k(y,q), \ \ \ \ k=1,\ldots,N,
\eeq
where the Hamiltonians $S_k$ are given by \eqref{Hamiltonians}--\eqref{f}. We proceed to explain this.

First, we note that in view of \eqref{sR} and \eqref{RDiffEq} the interaction function \eqref{f} can also be written
\beq\label{falt}
f(z)=\exp(-ir\rho)\left( \frac{R(z+\rho -i\alpha/2)R(z-\rho +i\alpha/2)}{R(z-i\alpha/2)R(z+i\alpha/2)}\right)^{1/2}.
\eeq
Hence \eqref{SS} is equivalent to 
\begin{multline}\label{BR}
\sum_{\substack{I\subset\lbrace 1,\ldots,N\rbrace\\ |I|=k}}\prod_{m\in I}\exp(\beta p_m)
\prod_{\substack{m\in I\\ n\notin I}}
\left( \frac{R(x_m-x_n+\rho -i\alpha/2)R(x_m-x_n-\rho +i\alpha/2)}{R(x_m-x_n-i\alpha/2)R(x_m-x_n+i\alpha/2)}\right)^{1/2}\\
=
\sum_{\substack{I\subset\lbrace 1,\ldots,N\rbrace\\ |I|=k}}\prod_{m\in I}\exp(\beta q_m)
\prod_{\substack{m\in I\\ n\notin I}}
\left( \frac{R(y_m-y_n+\rho -i\alpha/2)R(y_m-y_n-\rho +i\alpha/2)}{R(y_m-y_n-i\alpha/2)R(y_m-y_n+i\alpha/2)}\right)^{1/2}.
\end{multline}
Next, from \eqref{pj} and \eqref{qj} we have
\beq\label{epm}
\exp(\beta p_m)=\prod_{n\ne m}
\left(\frac{R(x_m-x_n-i\alpha/2)R(x_m-x_n-\rho+i\alpha/2)}{R(x_m-x_n+i\alpha/2)R(x_m-x_n+\rho-i\alpha/2)}\right)^{1/2}
\prod_{n=1}^N \frac{R(x_m-y_n+\rho/2)}{R(x_m-y_n-\rho/2)},
\eeq
\beq\label{eqm}
\exp(\beta q_m)=\prod_{n\ne m}
\left(\frac{R(y_m-y_n+i\alpha/2)R(y_m-y_n+\rho-i\alpha/2)}{R(y_m-y_n-i\alpha/2)R(y_m-y_n-\rho+i\alpha/2)}\right)^{1/2}
\prod_{n=1}^N \frac{R(x_n-y_m+\rho/2)}{R(x_n-y_m-\rho/2)}.
\eeq

Consider now the product of the quantities $\exp(\beta p_m)$ for $m$ in a fixed index set $I$. For a pair $m_1,m_2\in I$, the two corresponding radicand terms coming from the first product in \eqref{epm} cancel, since $R(z)$ is even. For pairs $m\in I$, $n\notin I$, we can cancel two of the four radicand factors and combine the remaining two to rewrite the left-hand side of \eqref{BR}. Likewise, the right-hand side of \eqref{BR} can be simplified. The upshot is that the B\"acklund property \eqref{SS} holds, provided the following identities are valid for $k=1,\ldots,N$:
\begin{multline}
	\sum_{\substack{I\subset\lbrace 1,\ldots,N\rbrace\\ |I|=k}}\prod_{\substack{m\in I\\ n\notin I}}\frac{R(x_m-x_n-\rho+i\alpha/2)}{R(x_m-x_n+i\alpha/2)}\prod_{\substack{m\in I\\ n=1,\ldots,N}}\frac{R(x_m-y_n+\rho/2)}{R(x_m-y_n-\rho/2)}\\ = \sum_{\substack{I\subset\lbrace 1,\ldots,N\rbrace\\ |I|=k}}\prod_{\substack{m\in I\\ n\notin I}}\frac{R(y_m-y_n+\rho-i\alpha/2)}{R(y_m-y_n-i\alpha/2)}\prod_{\substack{m\in I\\ n=1,\ldots,N}}\frac{R(x_n-y_m+\rho/2)}{R(x_n-y_m-\rho/2)}.
\end{multline}
These functional equations can be reduced to \eqref{EllipticFunctionalEqs} by using \eqref{sR}; cf.~also Eqs.~(2.8)--(2.9) in~\cite{Rui06}. Therefore, we have now demonstrated \eqref{SS}.

Next, recall that the kernel property is not spoiled when all coordinates $x_1,\ldots,x_N$ are translated by $\xi$ and when $\Psi(x,y)$ is multiplied by a function of the form \eqref{ambig}. In particular, we can allow a multiplier
\beq\label{psi}
\exp(i\psi(s(x,y))/\hbar),\ \ \ \ \ \ s(x,y)\equiv\sum_{j=1}^N (x_j-y_j),
\eeq
where $\psi(z)$ is any entire function. Taking the classical limit as before, we obtain a generating function that  yields an extra term
$\psi'(s(x,y))$
on the right-hand sides of \eqref{pj}--\eqref{qj}. Since these extra terms are equal and do not depend on the index $j$, they do not spoil the argument leading to the validity of \eqref{SS}. Hence a large family of B\"acklund transformations arises.

Returning to the non-rigorous status of these developments, we first repeat that we must switch to a parameter $\beta$ that is purely imaginary to ensure that for real positions and momenta $x,p$ there might exist vectors $y,q$  that not only satisfy \eqref{pj}--\eqref{qj}, but are also real. As they stand, the Hamiltonians $S_k(x,p)$ \eqref{Hamiltonians} are then not real-valued on $\Omega$, but this is easily remedied by switching to
\beq\label{HS}
H_k(x,p) =S_k(x,p)+S_k(x,-p),\ \ \ \ k=1,\ldots,N.
\eeq
Now this looks satisfactory at face value, but in fact problems remain. The point is that it seems very unlikely that the commuting local flows generated by $H_1,\ldots,H_{N-1}$ extend to global flows (the $H_N$-flow is of course global). For $H_1$, for example, the dependence on $p_j$ is now via a factor $\cos(|\beta|p_j)$. Hence, conservation of $H_1$  is compatible with collisions after a finite time, and constant-$H_1$ hypersurfaces are not compact, in contrast to the positive-$\beta$ case. 

Put differently, one should not expect to obtain Liouville tori in $\Omega$ for $\beta$ purely imaginary. Moreover, a proof of existence and uniqueness of a solution $y$ in the elliptic configuration space $G$~\eqref{config} to the equations \eqref{pj} for a given $(x,p)\in\Omega$ is not in sight. Finally, one has to deal with the Poisson commuting Hamiltonians~\eqref{HS} on $\Omega$, which can yield singularities in finite time, hence an unphysical behavior.

Ignoring these problems, we continue by tying in the above generating functions with the ones that can be found in~\cite{NRK96}. To this end we first view the generating function $F(x,y)$ given by \eqref{defF}--\eqref{FS} as the Lagrangian of a \lq discrete-time\rq\  map. 
(In a general setting, the relation between these two viewpoints has been clarified by Veselov~\cite{Ves91}.)
For a given sequence of vectors $z(n)\in\C^N$, $n\in\Z$, and
a fixed value of $n$, we write $z=z(n)$, $\tilde{z}=z(n+1)$ and $\utilde{z}=z(n-1)$, in accord with the notation used in~\cite{NRK96}. To make the connection to~\cite{NRK96}, we shall actually start from a slightly more general generating function, namely,
\beq\label{Fxi}
F_{\xi,\gamma}(x,y)=F(x_1+\xi,\ldots,x_N+\xi,y_1,\ldots,y_N)+\gamma s(x,y)^2/\beta,
\eeq
cf.~the paragraph containing~\eqref{psi}.
Then from the discrete Euler-Lagrange equations
\begin{equation}\label{ELEqs}
	\frac{\partial F_{\xi,\gamma}}{\partial y_j}\big(\utilde{z},z\big) + \frac{\partial F_{\xi,\gamma}}{\partial x_j}\big(z,\tilde{z}\big)=0,\quad j=1,\ldots,N,
\end{equation}
and \eqref{pj}--\eqref{qj} we deduce
\begin{multline}\label{RNewt}
	P(\gamma)\prod_{m\neq j}\frac{R(z_j-z_m-\rho +i\alpha/2)R(z_j-z_m-i\alpha/2)}{R(z_j-z_m+\rho -i\alpha/2)R(z_j-z_m+i\alpha/2)} \\ = \prod_{m=1}^N\frac{R(\utilde{z}_m-z_j+\rho/2+\xi)R(z_j-\tilde{z}_m-\rho/2+\xi)}{R(\utilde{z}_m-z_j-\rho/2+\xi)R(z_j-\tilde{z}_m+\rho/2+\xi)},
	\end{multline}
where the prefactor reads	
\beq
P(\gamma)=\prod_{m=1}^N\exp(2\gamma (\utilde{z}_m+\tilde{z}_m-2z_m)).
\eeq
Next, we choose 
\beq\label{xisp}
\xi=\rho/2+i\alpha/2,
\eeq
 and use \eqref{RDiffEq} and \eqref{sR} to rewrite the result as
\beq\label{sdisc}
P(\gamma)\prod_{m\neq j}\frac{s(z_j-z_m-\rho)}{s(z_j-z_m+\rho)}=\prod_{m=1}^N\frac{s(z_j-\utilde{z}_m-\rho)s(z_j-\tilde{z}_m)}{s(z_j-\utilde{z}_m)s(z_j-\tilde{z}_m+\rho)},\ \ \ \ s(z)=s(r,\alpha;z).
\eeq

We are now in the position to compare \eqref{sdisc} to Eq.~(2.14) in~\cite{NRK96}. To this end we substitute
\beq
\rho=-\lambda,
\eeq
and use~\eqref{ssigrel} to switch from the $s$-function to the Weierstrass $\sigma$-function. Then \eqref{sdisc} becomes
\beq
P(\gamma-\lambda \eta r/\pi)\prod_{m\neq j}\frac{\sigma (z_j-z_m+\lambda)}{\sigma (z_j-z_m-\lambda)}=\prod_{m=1}^N\frac{\sigma (z_j-\tilde{z}_m)\sigma(z_j-\utilde{z}_m+\lambda)}{\sigma (z_j-\utilde{z}_m)\sigma (z_j-\tilde{z}_m-\lambda)},\ \ \ \ j=1,\ldots,N.
\eeq
This coincides with Eq.~(2.14) in~\cite{NRK96}, provided $p/\utilde{p}$ is given by Eq.~(6.20) with $\theta=0$, and $\gamma$ is suitably specialized.

Even though \eqref{sdisc} can be made to coincide with the \lq discrete-time Newton equations\rq\ (2.14) in~\cite{NRK96}, the generating functions and Lagrangians employed in~\cite{NRK96} differ from the above ones. This is due to a different choice of phase space variables implicit in~\cite{NRK96}, which leads to their functions having an asymmetric dependence on $x$ and $y$.  (The choice amounts to a canonical map of the form $(x,p)\mapsto (x,p+f(x))$.)

We should add that the existence of sequences of vectors $z(n)\in\C^N$ satisfying \eqref{sdisc} is  left open. In any case, it seems unlikely that for given  initial values $z(0), z(1)$ in the elliptic configuration space $G$~\eqref{config} there exists a solution sequence that stays in $G$.

\subsection{The hyperbolic case and its dual}

From now on we reparametrize the two length scales $a_{\pm}$ as
\beq\label{aT}
a_{+}=2\pi/\mu,\ \ \ \ a_{-}=\hbar\beta,
\eeq
except in the elliptic cases considered in Section~4, where we retain an imaginary period $i\alpha$.
(Recall we motivated this change at the end of the Introduction, cf.~the paragraph containing \eqref{mu}.)
 It is also convenient to trade the parameter $\rho$ (which has dimension [position]) for a (dimensionless) parameter
\beq\label{tau}
\tau=-i\mu \rho/2.
\eeq
With these changes, we obtain from \eqref{Ghypcl} the following counterpart of Lemma~3.1.

\begin{lemma}
Let $\tau\in (0,\pi)$. For $x$ and $y$ in the hyperbolic configuration space 
 \beq\label{confighyp}
 G_{\rm hyp}= \{ x\in \R^N\mid  x_N<\cdots <x_1\},
 \eeq
  we have classical limits
\begin{equation}\label{cSlimh}
	\lim_{\hbar\downarrow 0}i\hbar\ln \cS(2\pi/\mu,\hbar\beta;x,y) = \frac{1}{\beta\mu}\sum_{j,k=1}^N \int_{\mu(x_j-y_k)+i\tau}^{\mu(x_j-y_k)-i\tau}dw\ln (2\cosh (w/2)),
\end{equation}
\begin{equation}\label{Wlimh}
	\lim_{\hbar\downarrow 0}i\hbar\ln W(2\pi/\mu,\hbar\beta;x) = \frac{1}{\beta\mu} \sum_{j\neq k}\int^{\mu(x_j-x_k)+i\pi}_{\mu(x_j-x_k)+i\pi -2i\tau}dw\ln (2\cosh (w/2)),
\end{equation}
where the integration paths stay away from the cuts $\pm i[\pi,\infty)$.
\end{lemma}

With the change $\Omega\to\Omega_{\rm hyp}$, where
\beq
\Omega_{\rm hyp}=\{ (x,p)\in\R^{2N}\mid x\in G_{\rm hyp}\},
\eeq
 the developments leading to \eqref{HS} now apply with straightforward adaptations. In particular, 
 the asymptotics~\eqref{PsiF} entails generating functions
 \begin{equation}\label{defFhyp}
	F(x,y) = \frac{1}{\beta\mu}\big(F_W(x) + F_W(y) + F_{\cS}(x,y)\big),
\end{equation}
where
\begin{equation}\label{FWhyp}
	F_W(x) = \frac{1}{2}\sum_{j\neq k}\int^{\mu(x_j-x_k)+i\pi}_{\mu(x_j-x_k)+i\pi -2i\tau}dw\ln (2\cosh (w/2)),
\end{equation}
\begin{equation}\label{FShyp}
	F_{\cS}(x,y) = \sum_{j,k=1}^N \int_{\mu(x_j-y_k)+i\tau}^{\mu(x_j-y_k)-i\tau}dw\ln (2\cosh (w/2)),
\end{equation}
and  the B\"acklund property \eqref{SS} is equivalent to the identities 
\begin{multline}\label{Bsc}
	\sum_{\substack{I\subset\lbrace 1,\ldots,N\rbrace\\ |I|=k}}\prod_{\substack{m\in I\\ n\notin I}}\frac{\sinh(\mu(x_m-x_n)/2-i\tau)}{\sinh(\mu(x_m-x_n)/2)}\prod_{\substack{m\in I\\ n=1,\ldots,N}}\frac{\cosh((\mu(x_m-y_n)+i\tau)/2)}{\cosh((\mu(x_m-y_n)-i\tau)/2)}\\ = \sum_{\substack{I\subset\lbrace 1,\ldots,N\rbrace\\ |I|=k}}\prod_{\substack{m\in I\\ n\notin I}}\frac{\sinh(\mu(y_m-y_n)/2+i\tau)}{\sinh(\mu(y_m-y_n)/2)}\prod_{\substack{m\in I\\ n=1,\ldots,N}}\frac{\cosh((\mu(x_n-y_m)+i\tau)/2)}{\cosh((\mu(x_n-y_m)-i\tau)/2)}.
\end{multline}

Next, starting again from a modified generating function \eqref{Fxi}, we get from the discrete Euler-Lagrange equations \eqref{ELEqs} as the counterpart of \eqref{RNewt}
 \begin{multline}\label{scNewt}
	P(\gamma)\prod_{m\neq j}\frac{\sinh(\mu(z_j-z_m)/2-i\tau)}{\sinh(\mu(z_j-z_m)/2+i\tau)} \\ = \prod_{m=1}^N\frac{\cosh(\mu(\utilde{z}_m-z_j+\xi)/2+i\tau/2)\cosh(\mu(z_j-\tilde{z}_m+\xi)/2-i\tau/2)}{\cosh(\mu(\utilde{z}_m-z_j+\xi)/2-i\tau/2)\cosh(\mu(z_j-\tilde{z}_m+\xi)/2+i\tau/2)}.
	\end{multline}
Choosing
\beq
\xi=i\tau/\mu +i\pi/\mu,
\eeq
this becomes
\beq
P(\gamma)\prod_{m\neq j}\frac{\sinh(\mu(z_j-z_m)/2-i\tau)}{\sinh(\mu(z_j-z_m)/2+i\tau)} =
\prod_{m=1}^N\frac{\sinh(\mu(z_j-\utilde{z}_m)/2-i\tau)\sinh(\mu(z_j-\tilde{z}_m)/2)}{\sinh(\mu(z_j-\utilde{z}_m)/2)\sinh(\mu(z_j-\tilde{z}_m)/2+i\tau)}.
 \eeq
This set of equations amounts to the hyperbolic specialization of Eq.~(2.14) in~\cite{NRK96}.

Just as in the elliptic case, it is unlikely that any solution sequences $z(n)$ exist that stay in $G_{\rm hyp}$ for all \lq discrete times\rq\ $n\in\Z$. To be sure, the equations to be solved do not involve $\beta$, so it might appear that the issue whether $\beta$ is real or imaginary is moot. In fact, however, the difference is decisive, since only in the latter case the equations have a chance to correspond to a canonical map on $\Omega_{\rm hyp}$. But since the commuting flows are not global for $\beta$ imaginary, there exists no well-defined action-angle map, by contrast to the case $\beta\in(0,\infty)$~\cite{Rui88}.

In fact, the existence of solution sequences in $\C^N$ has not been shown beyond doubt even in the hyperbolic case. In~\cite{NRK96} there are no reality conditions specified, and although the arguments for the hyperbolic case are formally convincing, they involve tacit assumptions that are not checked.

Finally, we discuss the dual hyperbolic case. This can be quite easily handled after our notation change \eqref{aT}--\eqref{tau}. Indeed, we need only substitute
\beq
x\to \hat{p},\ \ \ p\to \hat{x},\ \ \ y\to \hat{q},\ \ \ q\to \hat{y},
\eeq
replace $G_{\rm hyp}$ by
 \beq\label{dconfighyp}
 \hat{G}_{\rm hyp}= \{ \hat{p}\in \R^N\mid  \hat{p}_N<\cdots <\hat{p}_1\},
 \eeq
and interchange
\beq
\beta \leftrightarrow \mu,
\eeq
wherever these parameters occur~\cite{Rui88}. (See also Section~4, where we have occasion to say  more about the self-duality of the relativistic hyperbolic case.) In particular, note that the resulting generating function $F(\hat{p},\hat{q})$ now gives rise to purely imaginary positions $\hat{x}$ and $\hat{y}$ unless $\mu$ is required to be purely imaginary.

\subsection{The periodic Toda case}
We proceed to obtain the classical asymptotics~\eqref{PsiF} of the periodic Toda kernel functions 
\beq\label{Psipm}
\Psi^\pm(x,y)\equiv U(x)^{\mp 1/2}U(y)^{\mp 1/2}S^\pm(x,y).
\eeq
We recall that $U(x)$ and $S^\pm(x,y)$ are given in terms of the functions $G_R$ and $G_L$, cf.~\eqref{TU} and \eqref{Sp}-\eqref{Sm}, and that we have switched to the parameters \eqref{aT}. Therefore, taking $\hbar$ to $0$ amounts to taking $a_-$ to $0$, so we infer from \eqref{Ghypcl} and \eqref{GRDef}-\eqref{GLDef} that we have
\begin{equation}\label{GRLlnLim}
	\lim_{\hbar\downarrow 0}i\hbar\ln G_{\substack{R \\ L}}
	(2\pi/\mu,\hbar\beta;z) = \frac{1}{\beta}\int_0^z dw\ln\big(1+\exp(\mp \mu w)\big).
\end{equation}
The following lemma is now easy to verify.

\begin{lemma}
For $x$ and $y$ in $\mathbb{R}^N$ we have classical limits
\begin{multline}
	\lim_{\hbar\downarrow 0}i\hbar\ln S^{\pm}(2\pi/\mu,\hbar\beta;x,y) = \frac{\mp 1}{\beta}\sum_{m=1}^N\Bigg(\int_0^{x_{m+1}-y_m\pm i\pi/2\mu+\eta/2}dw\ln\big(1+\exp(\mu w)\big)\\ + \int_0^{y_m-x_m\pm i\pi/2\mu+\eta/2}dw\ln\big(1+\exp(\mu w)\big)\Bigg),
\end{multline}
\begin{equation}
	\lim_{\hbar\downarrow 0}i\hbar\ln U(2\pi/\mu,\hbar\beta;x) = -\frac{1}{\beta}\sum_{m=1}^N\int_0^{x_{m+1}-x_m+\eta}dw\ln\big(1+\exp(\mu w)\big),
\end{equation}
where the integration paths stay away from the cuts $\pm i\lbrack\pi/\mu,\infty)$.
\end{lemma}

Taking \eqref{PsiF} as a guide, we should consider transformations $B^\pm$ from canonical coordinates $(x,p)\in\mathbb{R}^{2N}$ to new canonical coordinates $(y,q)$,  generated by
\begin{equation}
	F^\pm(x,y) = \frac{\pm 1}{\beta}\big(F_U(x)+F_U(y)-F_S^\pm(x,y)\big),
\end{equation}
where
\begin{equation}
	F_U(x) = \frac{1}{2}\sum_{m=1}^N\int_0^{x_{m+1}-x_m+\eta}dw\ln\big(1+\exp(\mu w)\big),
\end{equation}
\begin{multline}
	F_S^\pm(x,y) = \sum_{m=1}^N\Bigg(\int_0^{x_{m+1}-y_m\pm i\pi/2\mu+\eta/2}dw\ln\big(1+\exp(\mu w)\big)\\ + \int_0^{y_m-x_m\pm i\pi/2\mu+\eta/2}dw\ln\big(1+\exp(\mu w)\big)\Bigg).
\end{multline}
However, as announced in the Introduction, there is a problem with these generating functions in the Toda case as well. Indeed, we need $\eta$ to be real to ensure formal self-adjointness of the Toda A$\De$Os, but this requirement implies that the gradients of $F^{+}$ and $F^{-}$ are not real-valued on $\mathbb{R}^{2N}$. Hence we should not expect that they give rise to bijections on the Toda phase space $\mathbb{R}^{2N}$. 

To attempt to improve this state of affairs, let us first focus our attention on $F^+(x,y)$ and consider a suitable analytic continuation. Specifically, taking 
\beq\label{xac}
x_m\to x_m-i\pi/2\mu +\eta/2,\ \ \ \ m=1,\ldots,N,
\eeq
 the function $F_U(x)$ is left invariant, while $F_S^+(x,y)$ turns into
\begin{multline}\label{FS+}
	F_S^+(x,y) = \sum_{m=1}^N\Bigg(\int_0^{x_{m+1}-y_m+\eta}dw\ln\big(1+\exp(\mu w)\big)\\ + \int_0^{y_m-x_m}dw\ln\big(1-\exp(\mu w)\big)\Bigg)+NI/\mu,
\end{multline}
where $I$ is the integral
\beq
I=\int_{0}^{i\pi}dz\ln (1+\exp(z)).
\eeq

 Since additive constants are irrelevant, we have now been led to a more well-behaved generating function. Clearly, if we start from $F^-(x,y)$ rather than $F^+(x,y)$,  we would obtain the same function up to an overall sign (which is a matter of convention anyway). The key question remains, however, whether this function does give rise to  a symplectomorphism  $(x,p)\mapsto(y,q)$  on $\mathbb{R}^{2N}$ via~\eqref{Fgen}. Accordingly, reverting from now on to the dimensionless coupling $\gamma$ defined by (recall \eqref{fT} and \eqref{gameta})
 \beq\label{gamma}
 \gamma =\exp(\mu \eta/2),
 \eeq
we should study  the equations
\begin{equation}
\begin{split}\label{pm}
	p_m &=- \frac{\partial F^+}{\partial x_m}\\ &= \frac{1}{2\beta}\ln\left(\frac{1+\gamma^2\exp\big(\mu(x_{m+1}-x_m)\big)}{1+\gamma^2\exp\big(\mu(x_m-x_{m-1})\big)}\right)\\ &\quad + \frac{1}{\beta}\ln\left(\frac{1+\gamma^2\exp\big(\mu(x_m-y_{m-1})\big)}{1-\exp(\mu(y_m-x_m)\big)}\right),\ \quad m=1,\ldots,N,
\end{split}
\end{equation}
\begin{equation}\label{qm}
\begin{split}
	q_m &= \frac{\partial F^+}{\partial y_m}\\ &= \frac{1}{2\beta}\ln\left(\frac{1+\gamma^2\exp\big(\mu(y_m-y_{m-1})\big)}{1+\gamma^2\exp\big(\mu(y_{m+1}-y_m)\big)}\right)\\ &\quad + \frac{1}{\beta}\ln\left(\frac{1+\gamma^2\exp\big(\mu(x_{m+1}-y_m)\big)}{1-\exp\big(\mu(y_m-x_m)\big)}\right),\ \quad m=1,\ldots,N.
\end{split}
\end{equation}

There is still a troubling aspect about these equations: If we start from an arbitrary phase space point $(x,p)\in\R^{2N}$, then  we can only accept a solution $y(x,p)$ to the implicit equations \eqref{pm} when it satisfies
\beq\label{puz}
y_m<x_m,\ \ \ \ m=1,\ldots,N.
\eeq
Now the existence of a solution with this property cannot easily be ruled out, but also seems hard to prove. Even when we assume that such solutions exist and give rise to a well-defined symplectomorphism, there is still a problem with reinterpreting the generating function as a Lagrangian for a discrete map on $\R^N$. Indeed, in that case we would expect to have the freedom to choose arbitrary $x, y$ in $\R^N$ to obtain  a unique sequence $z(n)\in\R^N$, $n\in\Z$, with initial values $z(0)=x, z(1)=y$, whereas this freedom is at variance with the constraint \eqref{puz}.

Ignoring these dilemmas from now on, we can easily show the B\"acklund property \eqref{SS} by proceeding in the same way as in Subsection~3.1. Specifically, in this case \eqref{SS} can be written
\begin{multline}\label{BT}
\sum_{\substack{I\subset\lbrace 1,\ldots,N\rbrace\\ |I|=k}}\prod_{m\in I}\exp(\beta p_m)
\prod_{\substack{m\in I\\ m+1\notin I}}\Big(1+\gamma^2\exp\big(\mu(x_{m+1}-x_m)\big)\Big)^{1/2} \\
\times \prod_{\substack{m\in I\\ m-1\notin I}}\Big(1+\gamma^2\exp\big(\mu(x_m-x_{m-1})\big)\Big)^{1/2}
 =(x,p\to y,q).
\end{multline}
Also, from \eqref{pm}-\eqref{qm}
we have
\begin{equation}
	\exp(\beta p_m) = \left(\frac{1+\gamma^2\exp\big(\mu(x_{m+1}-x_m)\big)}{1+\gamma^2\exp\big(\mu(x_m-x_{m-1})\big)}\right)^{1/2}\frac{1+\gamma^2\exp\big(\mu(x_m-y_{m-1})\big)}{1-\exp\big(\mu(y_m-x_m)\big)},
\end{equation}
\begin{equation}
	\exp(\beta q_m) = \left(\frac{1+\gamma^2\exp\big(\mu(y_m-y_{m-1})\big)}{1+\gamma^2\exp\big(\mu(y_{m+1}-y_m)\big)}\right)^{1/2}\frac{1+\gamma^2\exp\big(\mu(x_{m+1}-y_m)\big)}{1-\exp\big(\mu(y_m-x_m)\big)}.
\end{equation}
Now fix an index set $I$ and consider the product of quantities $\exp(\beta p_m)$ with $m$ in $I$. For $m_1,m_2\in I$ such that $m_2=m_1+1~({\rm mod}~N)$, the denominator of the radicand corresponding to $m_2$ cancels the numerator of the radicand corresponding to $m_1$, and vice versa for $m_2=m_1-1~({\rm mod}~N)$. The product of quantities $\exp(\beta q_m)$ can be simplified similarly. From this it readily follows that the B\"acklund property \eqref{BT} amounts to the functional equations
\begin{multline}
	\sum_{\substack{I\subset\lbrace 1,\ldots,N\rbrace\\ |I|=k}}\prod_{\substack{m\in I\\ m+1\notin I}}\Big(1+\gamma^2\exp\big(\mu(x_{m+1}-x_m)\big)\Big)\prod_{m\in I}\frac{1+\gamma^2\exp\big(\mu(x_m-y_{m-1})\big)}{1-\exp\big(\mu(y_m-x_m)\big)}\\ = \sum_{\substack{I\subset\lbrace 1,\ldots,N\rbrace\\ |I|=k}}\prod_{\substack{m\in I\\ m-1\notin I}}\Big(1+\gamma^2\exp\big(\mu(y_m-y_{m-1})\big)\Big)\prod_{m\in I}\frac{1+\gamma^2\exp\big(\mu(x_{m+1}-y_m)\big)}{1-\exp\big(\mu(y_m-x_m)\big)}.
\end{multline}
These equations can be reduced to the  identities \eqref{TodaFunctionalEqs} by taking
\beq
\delta=+, a_{+}=2\pi/\mu, \rho=(2\ln\gamma+i\pi)/\mu.
\eeq
Thus we have established \eqref{SS} in the periodic Toda case.

Finally, we tie in the above generating function with a \lq discrete-time\rq\ map introduced by Suris \cite{Sur96} as a time-discretization of the defining Hamiltonian of the relativistic periodic Toda system.  To this end we observe that the discrete Euler-Lagrange equations
\begin{equation}
	\frac{\partial F^{+}}{\partial y_m}(\utilde{z},z) + \frac{\partial F^{+}}{\partial x_m}(z,\tilde{z}) = 0,\quad m=1,\ldots,N,
\end{equation}
for a sequence of vectors $z(n)\in\mathbb{C}^N$, $n\in\mathbb{Z}$, are equivalent to
\begin{multline}
	\frac{1-\exp\big(\mu(\tilde{z}_m-z_m)\big)}{1-\exp\big(\mu(z_m-\utilde{z}_m))} = \frac{1+\gamma^2\exp\big(\mu(z_{m+1}-z_m)\big)}{1+\gamma^2\exp\big(\mu(\utilde{z}_{m+1}-z_m)\big)}\\ \times\frac{1+\gamma^2\exp\big(\mu(z_m-\tilde{z}_{m-1})\big)}{1+\gamma^2\exp\big(\mu(z_m-z_{m-1})\big)},\quad m=1,\ldots,N.
\end{multline}
These  equations coincide with Eq.~(5.4) in \cite{Sur96} if we substitute $\mu =1,\gamma=g$.

\subsection{The nonperiodic Toda case and its dual}
 The classical asymptotics \eqref{PsiF} of kernel functions \eqref{Psipm}  with $U(x)$ given by \eqref{TUnp} and $S^\pm(x,y)$ by \eqref{Spl}--\eqref{Sml} can be studied in virtually the same way as in the periodic case. The reader who has followed us this far will have no difficulty to make the pertinent changes. In particular, with the nonperiodic Toda convention \eqref{npconv} in force, the B\"acklund property \eqref{SS} continues to hold, and the nonperiodic analog of $F^+(x,y)$ can be readily tied in with Suris' time-discretization of the defining Hamiltonian of the relativistic nonperiodic Toda system~\cite{Sur96} . For this case the action-angle map is known from~\cite{Rui90}, and so a further study might clarify the puzzling constraints~\eqref{puz}, which arise in the same way as in the periodic case. This is beyond our present scope, however.

Continuing with the dual system, we get from Subsection~2.5 a kernel function
\beq
\hat{\Psi}(\hat{p},\hat{q})=\hat{W}(\beta \hat{p}/\mu)^{1/2}\hat{W}(\beta \hat{q}/\mu)^{1/2}
\hat{\cS}(\beta\hat{p}/\mu,\beta\hat{q}/\mu),
\eeq
where $\hat{W}$ and $\hat{\cS}$ are given by \eqref{Wd} and \eqref{cSd}. Also, as explained earlier, we 
use the parameters \eqref{aT}. The classical asymptotics \eqref{PsiF} is now easily inferred from \eqref{Ghypcl}.

\begin{lemma}
For $\hat{p}$ and $\hat{q}$ in $\hat{G}$~\eqref{dTconf}, we have classical limits
\beq
\lim_{\hbar\downarrow 0}i\hbar \ln \hat{\cS}(2\pi/\mu,\hbar\beta;\beta\hat{p}/\mu,\beta\hat{q}/\mu)=\frac{1}{\beta\mu}\sum_{m,n=1}^N
\int_0^{\beta(\hat{p}_m-\hat{q}_n)}dw\ln (2\cosh(w/2)),
\eeq
\beq\label{hWlim}
\lim_{\hbar\downarrow 0}i\hbar \ln 
\hat{W}(2\pi/\mu,\hbar\beta;\beta \hat{p}/\mu)=\frac{i\pi}{\mu}\sum_{1\le m<n\le N}
(\hat{p}_m-\hat{p}_n),
\eeq
where the integration paths stay away from the cuts $\pm i[\pi,\infty)$.
\end{lemma}

Hence we should study whether the generating function
\begin{equation}
	F(\hat{p},\hat{q}) = \frac{1}{\beta\mu}\big(F_{\hat{W}}(\hat{p})+F_{\hat{W}}(\hat{q})+F_{\hat{\mathcal{S}}}(\hat{p},\hat{q})\big),
\end{equation}
where
\begin{equation}
	F_{\hat{W}}(\hat{p}) = \frac{i\pi}{2}\sum_{m=1}^N(N-2m+1)\beta\hat{p}_m,
\end{equation}
\begin{equation}
	F_{\hat{\mathcal{S}}}(\hat{p},\hat{q}) = \sum_{m,n=1}^N\int_0^{\beta(\hat{p}_m-\hat{q}_n)}dw\ln\big(2\cosh(w/2)\big),
\end{equation}
might lead to a canonical map $(\hat{x},\hat{p})\mapsto (\hat{y},\hat{q})$ on the dual Toda phase space
\beq
\hat{\Omega}=\{ (\hat{x},\hat{p})\in\R^{2N} \mid \hat{p}\in\hat{G}\}.
\eeq
This entails that $\hat{q}(\hat{x},\hat{p})$ is to be determined from the equations
\begin{equation}\label{xdT}
	\hat{x}_m = -\frac{i\pi}{2\mu}(N-2m+1)-\frac{1}{\mu}\sum_{n=1}^N\ln\big(2\cosh(\beta (\hat{p}_m-\hat{q}_n)/2)\big),
\end{equation}
and then $\hat{y}(\hat{x},\hat{p})$ is given by
\begin{equation}\label{ydT}
	\hat{y}_m = \frac{i\pi}{2\mu}(N-2m+1)-\frac{1}{\mu}\sum_{n=1}^N\ln\big(2\cosh(\beta (\hat{p}_n-\hat{q}_m)/2)\big).
\end{equation}

Clearly, for $\hat{p},\hat{q}\in\hat{G}$ the numbers $\hat{x}_m$ and $\hat{y}_m$ are not real as they stand, and this cannot be cured by switching to a parameter $\mu$ that is purely imaginary. On the other hand, we can also start from a modified generating function 
\beq\label{Fmod}
\tilde{F}(\hat{p},\hat{q})= \frac{i\pi}{2\mu}\sum_{m=1}^N (\hat{p}_m-\hat{q}_m)+F((\hat{p}_1-i\pi/\beta,\ldots,\hat{p}_N-i\pi/\beta),\hat{q}),
\eeq
which can be obtained from a modified kernel function, as explained earlier (cf.~the paragraph containing \eqref{ambig}). The point of this is that a solution $\hat{q}\in\hat{G}$ to the modified equations 
\begin{equation}\label{xdTm}
	\hat{x}_m=-\frac{\partial\tilde{F}}{\partial \hat{p}_m} = -\frac{i\pi}{2\mu}(N-2m+2)-\frac{1}{\mu}\sum_{n=1}^N\ln\big(-2i\sinh(\beta (\hat{p}_m-\hat{q}_n)/2)\big),
\end{equation}
might  then exist, provided its coordinates interlace with those of the given vector $\hat{p}\in\hat{G}$:
\beq\label{interl}
\hat{q}_N<\hat{p}_N<\hat{q}_{N-1}<\cdots <\hat{q}_1<\hat{p}_1.
\eeq
 In that case, the vector $\hat{y}$ given by the modified equations 
 \begin{equation}\label{ydTm}
	\hat{y}_m =\frac{\partial\tilde{F}}{\partial \hat{q}_m}= \frac{i\pi}{2\mu}(N-2m)-\frac{1}{\mu}\sum_{n=1}^N\ln\big(-2i\sinh(\beta (\hat{p}_n-\hat{q}_m)/2)\big),
\end{equation}
would be real. 

Again, a further study of this global analysis question is beyond our scope. Rather, we proceed to show the validity of the B\"acklund property
\begin{equation}\label{dSS}
	\hat{H}_k(\hat{x},\hat{p}) =\hat{ H}_k(\hat{y},\hat{q}),\ \quad k=1,\ldots,N,
\end{equation}
for the dual Hamiltonians, which we rewrite as
\beq
\hat{H}_k(\hat{x},\hat{p})= \sum_{\substack{I\subset\lbrace 1,\ldots,N\rbrace\\ |I|=k}}
 \prod_{\substack{m\in I\\ n\notin I}}\frac{\beta/2}{|\sinh(\beta(\hat{p}_m-\hat{p}_n)/2)|}
\prod_{l\in I}\exp(\mu \hat{x}_l),\quad k=1,\ldots,N.
\eeq

Substituting the above expressions \eqref{xdT}--\eqref{ydT} for $\hat{x}$ and $\hat{y}$, we infer that \eqref{dSS} holds if and only if the following identities are valid for $k=1,\ldots,N$:
\begin{multline}\label{dTId}
	\sum_{\substack{I\subset\lbrace 1,\ldots,N\rbrace\\ |I|=k}}\prod_{m\in I}\exp\left(-\frac{i\pi}{2}(N-2m+1)\right)\prod_{n\notin I}\left|\frac{1}{\sinh(\beta(\hat{p}_m-\hat{p}_n)/2)}\right|\\ \times\prod_{\substack{m\in I\\ n=1,\ldots,N}}\frac{1}{\cosh(\beta(\hat{p}_m-\hat{q}_n)/2)}\\ = \sum_{\substack{I\subset\lbrace 1,\ldots,N\rbrace\\ |I|=k}}\prod_{m\in I}\exp\left(\frac{i\pi}{2}(N-2m+1)\right)\prod_{n\notin I}\left|\frac{1}{\sinh(\beta(\hat{q}_m-\hat{q}_n)/2)}\right|\\ \times\prod_{\substack{m\in I\\ n=1,\ldots,N}}\frac{1}{\cosh(\beta(\hat{p}_n-\hat{q}_m)/2)}.
\end{multline}
Fix an index set $I$ and consider the corresponding term on the left-hand side. Since $\hat{p}$ belongs to $\hat{G}$~\eqref{dTconf}, we can rewrite this term as
\begin{equation}
	\prod_{m\in I}i^{-N+2m-1}\prod_{\substack{n\notin I\\ n>m}}(-1)\prod_{n\notin I}\frac{1}{\sinh(\beta(\hat{p}_n-\hat{p}_m)/2)}\prod_{n=1}^N\frac{1}{\cosh(\beta(\hat{p}_m-\hat{q}_n)/2)}.
\end{equation}
The numerical factor can be simplified using 
\begin{equation}
	\prod_{\substack{m\in I\\ n\notin I\\ n>m}}(-) = \prod_{\substack{m\in I\\ n=1,\ldots,N\\ n>m}}(-)\Bigg/\prod_{\substack{m,n\in I\\ n>m}}(-) = \prod_{m\in I}(-)^{N-m}\Bigg/ (-)^{k(k-1)/2}=i^{-k(k-1)}\prod_{m\in I}i^{2N-2m}.
\end{equation}
Doing so, we obtain
\begin{equation}
	i^{k(N-k)}\prod_{n\notin I}\frac{1}{\sinh(\beta(\hat{p}_n-\hat{p}_m)/2)}\prod_{n=1}^N\frac{1}{\cosh(\beta(\hat{p}_m-\hat{q}_n)/2)}.
\end{equation}
Likewise, since $\hat{q}\in\hat{G}$, the corresponding term on the right-hand side can be simplified to yield
\begin{equation}
	i^{k(N-k)}\prod_{n\notin I}\frac{1}{\sinh(\beta(\hat{q}_m-\hat{q}_n)/2)}\prod_{n=1}^N\frac{1}{\cosh(\beta(\hat{p}_n-\hat{q}_m)/2)}.
\end{equation}
We can thus reduce \eqref{dTId} to \eqref{sId} by substituting $\hat{p}_j\to \hat{p}_j-i\pi/\beta$, $j=1,\ldots,N$, and cancelling the $I$-independent numerical factor. This proves the B\"acklund property \eqref{dSS}.

\section{Nonrelativistic limits}

\subsection{Kernel functions}

\subsubsection{The elliptic case}

To take the nonrelativistic limit $c\to\infty$ of the various quantities in Subsection~2.1, we first need to reparametrize the two positive step sizes $a_{+}$ and $a_{-}$ in the elliptic gamma function as $\alpha$ and $\hbar\beta$, cf.~\eqref{aa}. Taking $\beta=1/mc$ to 0 then amounts to taking $a_{-}$ to 0. (We keep $\hbar$ fixed, since we wish to stay in the quantum setting.) To obtain nontrivial limits, we should also set
\beq\label{rhog}
\rho =i\beta g.
\eeq
Since we let $\beta$ go to 0, the coupling constant $g$ is allowed to vary over $(0,\infty)$. 

From \eqref{Gellnr} it now follows  that the kernel function $\cS(x,y)$ \eqref{defcS} has $\beta\to 0$ limit
\beq\label{cSnr}
\cK(x,y)=\prod_{j,k=1}^NR(x_j-y_k)^{-g/\hbar}.
\eeq
Moreover, 
(\ref{Gellnr}) yields the nonrelativistic limit
of the weight function  given by \eqref{W}--\eqref{HCell}:
\beq\label{Wnr}
W_{\rm nr}(x)=\left(\prod_{1\le j<k\le N}R(x_j-x_k+i\alpha/2)R(x_j-x_k-i\alpha/2)\right)^{g/\hbar}.
\eeq

The  nonrelativistic counterparts $H_{m,{\rm nr}}$, $m=1,\ldots,N$, of the commuting  Hamiltonians $H_{k,+}$ \eqref{Hkp} arise as $\beta\to 0$ limits of suitable linear combinations of the Hamiltonians $H_{1,+},\ldots,H_{m,+}$. This limit transition hinges on the use of the classical elliptic relativistic and nonrelativistic Lax matrices. Later on, we use these matrices in our study of classical B\"acklund transformations, cf.~\eqref{Lell} and \eqref{Kric}. On the quantum level, however, the pertinent limit is fraught with ordering problems. Although these can be resolved, the details are quite substantial and will be skipped. They can be found in Subsection~4.3 of~\cite{Rui94}. (We employ a similar method in the quantum periodic Toda case below, and make use of a corresponding formula for the Hamiltonians in the hyperbolic case discussed in the next section.)

An important ingredient of the reasoning in~\cite{Rui94} is a uniqueness result obtained by Oshima and H.~Sekiguchi~\cite{OS95}; this paper also contains explicit expressions for the commuting elliptic PDOs. The relevant limits entail the kernel identities
\beq\label{HnrP}
(H_{k,{\rm nr}}(x)-H_{k,{\rm nr}}(-y))\Psi_{\rm nr}(x,y)=0,\ \ \ \ k=1,\ldots,N,
\eeq
where 
\beq\label{Psinr}
\Psi_{\rm nr}(x,y)=W_{\rm nr}(x)^{1/2}W_{\rm nr}(y)^{1/2}\cK(x,y).
\eeq
 More explicitly, the Hamiltonians are  $N$ commuting PDOs of the form
\beq\label{H1nr}
H_{1,{\rm nr}}(x)=-i\hbar\sum_{j=1}^N\partial_{x_j},
\eeq
\beq\label{H2}
H_{2,{\rm nr}}(x)=-\hbar^2\sum_{1\le j_1< j_2\le N}\partial_{{x_j}_1}\partial_{{x_j}_2}-g(g-\hbar)\sum_{1\le j<l\le N}\wp(x_j-x_l;\pi/2r,i\alpha/2),
\eeq
\beq
H_{k,{\rm nr}}(x)=(-i\hbar)^k
\sum_{1\le j_1<\cdots < j_k\le N}\partial_{{x_j}_1}\cdots \partial_{{x_j}_k} +{\rm l.\ o.},\ \ \ \ k=3,\ldots,N,
\eeq
where l.~o.~denotes terms that are of lower order in the $x_j$-partials~\cite{OS95}. (In~\eqref{H2} we have omitted an additive constant that depends on the spectral parameter in the Lax matrix.)
In particular, for the defining Hamiltonian
\beq\label{Hnr}
H_{{\rm nr}}(x)= \frac{1}{2}(H_{1,{\rm nr}}(x))^2-H_{2,{\rm nr}}(x)=-\frac{\hbar^2}{2}\sum_{j=1}^N\partial_{x_j}^2+g(g-\hbar)\sum_{1\le j<l\le N}\wp(x_j-x_l;\pi/2r,i\alpha/2)
\eeq
of the elliptic nonrelativistic Calogero-Moser system,
this implies the kernel identity
\beq\label{HP}
(H_{{\rm nr}}(x)-H_{{\rm nr}}(y))\Psi_{\rm nr}(x,y)=0.
\eeq
(Here and below, we choose $m=1$.) This identity was first obtained by Langmann \cite{Lan00}, cf.~also~\cite{Rui04} for further details.

\subsubsection{The hyperbolic case and its dual}

Using the parameters \eqref{aT} and \eqref{rhog}, we obtain from \eqref{Ghypnr} the nonrelativistic limits
\beq\label{cSnrh} 
\cK(x,y)=\prod_{j,k=1}^N[2\cosh(\mu(x_j-y_k)/2)]^{-g/\hbar},
\eeq
\beq\label{Wnrh}
W_{\rm nr}(x)=\left(\prod_{1\le j<k\le N}4\sinh^2(\mu(x_j-x_k)/2)\right)^{g/\hbar},
\eeq
as the analogs of \eqref{cSnr}--\eqref{Wnr}. The kernel function $\Psi_{\rm nr}$ is given by \eqref{Psinr} with \eqref{cSnrh}--\eqref{Wnrh} in force. Then we obtain \eqref{HnrP}--\eqref{HP} with the replacement
\beq\label{replacement}
\wp(x_j-x_l;\pi/2r,i\alpha/2)\to \mu^2/4\sinh^2(\mu(x_j-x_l)/2).
\eeq

In contrast to the elliptic case, we are also able to obtain kernel identities relating the PDOs $H_{k,{\rm nr}}$ in $N$ variables $x=(x_1,\ldots,x_N)$ to a sum of the PDOs $H_{j,{\rm nr}}$ in $N-\ell$ variables $y=(y_1,\ldots,y_{N-\ell})$, $\ell=0,\ldots,N$. The key for doing so consists of explicit formulae for the PDOs, which involve not only the nonrelativistic Lax matrix
\beq\label{Lnrhyp}
L_{\rm nr}(x,p)_{jk} \equiv \delta_{jk}p_j + (1-\delta_{jk})\frac{i\mu g}{2\sinh\big(\mu(x_j-x_k)/2\big)},\quad j,k=1,\ldots,N,
\eeq
but also the diagonal $N\times N$ matrix
\beq
E(x)\equiv\text{diag}\big(z_1(x),\ldots,z_N(x)\big),
\eeq
with
\beq\label{zjx}
z_j(x) \equiv -\frac{i\mu g}{2}\sum_{k\neq j}\coth\big(\mu(x_j-x_k)/2\big),\ \ \ \ j=1,\ldots,N.
\eeq
We note that the substitution $x_j\to ix_j$ in $L_{\rm nr}(x,p)$ yields a trigonometric Lax matrix that is slightly different from the Lax matrix due to Moser~\cite{Mos75}. (Specifically, in Moser's Lax matrix the function $1/\sin$ is replaced by~$\cot$.) 

The canonical quantization substitution
\beq
p_j\to -i\hbar \partial_{x_j},\ \ \ \ j=1,\ldots,N,
\eeq
in $L_{\rm nr}(x,p)$ yields an operator-valued matrix whose symmetric functions $\hat{\Sigma}_k(L_{\rm nr})(x)$ are well defined, since no ordering ambiguities occur. (Indeed, a term in the expansion of a principal minor of~\eqref{Lnrhyp} that depends on $p_j$ does not depend on $x_j$.) However, $\hat{\Sigma}_2(L_{\rm nr})(x)$ is not equal to
\beq
H_{2,{\rm nr}}(x)=-\hbar^2\sum_{1\le j_1< j_2\le N}\partial_{{x_j}_1}\partial_{{x_j}_2}-g(g-\hbar)\sum_{1\le j<l\le N}\mu^2/4\sinh^2(\mu(x_j-x_l)/2),
\eeq
since the term proportional to $\hbar$ in the potential energy is missing. (In fact, there seems to be no complete proof in the literature that the $N$ PDOs $\hat{\Sigma}_k(L_{\rm nr})(x)$ commute for arbitrary $N$, cf.~\cite{Rui94}, Subsection~4.2.)

By contrast, the symmetric functions $\Sigma_k(L_{\rm nr}(x,p)+E(x))$ with $k>1$  contain products of terms that depend on $p_m$ and $x_m$, so their canonical quantization is ambiguous. As shown in~Section~4.3 of~\cite{Rui94}, the ordering choice ensuring commutativity is normal ordering: the procedure of putting $x$-dependent coefficients to the left of monomials in  the momentum operators $-i\hbar\partial_{x_1},\ldots,-i\hbar\partial_{x_N}$. We shall write $:\hat{\Sigma}_k(L_{\text{nr}}+E)(x):$ for the normal-ordered PDOs obtained from $\Sigma_k(L_{{\rm nr}}(x,p)+E(x))$ by substituting $-i\hbar\partial_{x_m}$ for $p_m$. The nonrelativistic commuting Hamiltonians $H_{k,{\rm nr}}$ are then given by the formula
\beq\label{Hhnr}
H_{k,{\rm nr}}(x) = W_{\rm nr}(x)^{1/2}:\hat{\Sigma}_k\big(L_{\rm nr}+E\big)(x):W_{\rm nr}(x)^{-1/2},\ \ \ \ k=1,\ldots,N.
\eeq
(To get a feel for what is going on here, the reader may wish to check the case $N=k=2$.)

The corresponding kernel identities have the same structure as the identities \eqref{hyperbolicKernelIds} in Theorem \ref{HypKernelIdProp}. However, in this case they involve coefficients $c_{\ell,j}$ with $\ell\in\mathbb{N}$, $j\in\mathbb{Z}$, given by
\beq\label{nrside}
c_{0,0}=1,\ \ \ c_{\ell,j}=0,\ \ \ j>\ell,\ \ \ j<0,
\eeq
and
\beq\label{nrcS}
c_{\ell,j} = \left(\frac{i\mu g}{2}\right)^jS_j(1,3,\ldots,2\ell-1),\ \ \ \ j=0,\ldots,\ell,
\eeq
where $S_j(a_1,\ldots,a_\ell)$ is the $j$th elementary symmetric function of $a_1,\ldots,a_\ell$. We note that the coefficients are uniquely determined by the recurrence relation
\beq\label{cCoeffsRel}
c_{\ell+1,j} = c_{\ell,j}+\frac{i\mu g}{2}(2\ell+1) c_{\ell,j-1}
\eeq
together with the side conditions \eqref{nrside}.

We are now ready to state and prove the pertinent identities.

\begin{theorem}
For $\ell=0,1,\ldots,N$, let
\beq\label{cKLim}
\cK_\ell(x,y)\equiv \frac{\exp\left(\frac{\mu g\ell}{2\hbar}\left(\sum_{n=1}^{N-\ell}y_n-\sum_{m=1}^Nx_m\right)\right)}{\prod_{m=1}^N\prod_{n=1}^{N-\ell}[2\cosh(\mu(x_m-y_n)/2)]^{g/\hbar}},\ \ \ \ \cK_N(x)\equiv\exp\left(-\frac{\mu gN}{2\hbar}\sum_{m=1}^Nx_m\right).
\eeq
For any $k\in\lbrace 1,\ldots,N\rbrace$, we have
\begin{multline}\label{nrhyperbolicKernelIds}
:\hat{\Sigma}_k\big(L_{\rm nr}+E\big)(x_1,\ldots,x_N):\cK_\ell(x,y)\\ = \sum_{j=0}^{\min(k,\ell)}c_{\ell,j}:\hat{\Sigma}_{k-j}\big(L_{\rm nr}+E\big)(-y_1,\ldots,-y_{N-\ell}):\cK_\ell(x,y),
\end{multline}
where
\beq
:\hat{\Sigma}_m\big(L_{\rm nr}+E\big)(-y_1,\ldots,-y_{N-\ell}):\ \equiv 0,\quad m>N-\ell,\quad :\hat{\Sigma}_0\big(L_{\rm nr}+E\big):\ \equiv 1.
\eeq
\end{theorem}

\begin{proof}
We shall prove the statement by induction on $\ell$. For $\ell=0$ the kernel identity \eqref{nrhyperbolicKernelIds} is equivalent to the hyperbolic counterpart of~\eqref{HnrP}. Hence, we now assume \eqref{nrhyperbolicKernelIds} for $\ell\geq 0$ and establish its validity for $\ell\to\ell+1$.

Setting
\beq
\phi(y)\equiv\exp\left(\frac{\mu g}{2\hbar}\left(\sum_{n=1}^{N-\ell-1}y_n-(N+\ell)y_{N-\ell}\right)\right),
\eeq
we start from the limit
\beq\label{cKhypLim}
\lim_{\Lambda\to\infty}\phi(y_1,\ldots,y_{N-\ell}-\Lambda)\cK_\ell(x,y_1,\ldots,y_{N-\ell}-\Lambda) = \cK_{\ell+1}(x,y),
\eeq
which is readily verified.
To make use of this, we note the commutation relation
\begin{multline}\label{commRel}
\phi(y):\hat{\Sigma}_k\big(L_{\rm nr} + E\big)(-y):\\ =\ :\hat{\Sigma}_k\left(L_{\rm nr} + E-\frac{i\mu g}{2}{\rm diag}(1,\ldots,1,-N-\ell)\right)(-y):\phi(y).
\end{multline}
Substituting $y_{N-\ell}\to y_{N-\ell}-\Lambda$ in the Lax matrix $L_{\rm nr}(-y)$, it is clear that the matrix elements $(L_{\rm nr})_{N-\ell,k}$ and $(L_{\rm nr})_{j,N-\ell}$, where $j,k=1,\ldots,N-\ell-1$, vanish in the limit $\Lambda\to\infty$, whereas $(L_{\rm nr})_{N-\ell,N-\ell}=i\hbar\partial_{y_{N-\ell}}$ remains the same. We also note the limits
\beq\label{zlim1}
\lim_{\Lambda\to\infty}z_{N-\ell}(-y_1,\ldots,-y_{N-\ell}+\Lambda) = -\frac{i\mu g}{2}(N-\ell-1),
\eeq
\beq\label{zlim2}
\lim_{\Lambda\to\infty}z_j(-y_1,\ldots,-y_{N-\ell}+\Lambda) = \frac{i\mu g}{2} + z_j(-y_1,\ldots,-y_{N-\ell-1}),\quad j=1,\ldots,N-\ell-1.
\eeq
We now multiply both sides of \eqref{nrhyperbolicKernelIds} by $\phi(y)$, substitute $y_{N-\ell}\to y_{N-\ell}-\Lambda$, and consider a term on the right-hand side corresponding to some $j=0,\ldots,\min(k,\ell)$. Using the commutation relation \eqref{commRel} and the limits \eqref{cKhypLim} and \eqref{zlim1}--\eqref{zlim2}, we find that the $\Lambda\to\infty$ limit of the term is given by
\begin{multline}
:\hat{\Sigma}_{k-j}\big(L_{\text{nr}}+E\big)(-y_1,\ldots,-y_{N-\ell-1}):\cK_{\ell+1}(x,y)\\ + \frac{i\mu g}{2}(2\ell+1) :\hat{\Sigma}_{k-j-1}\big(L_{\text{nr}}+E\big)(-y_1,\ldots,-y_{N-\ell-1}):\cK_{\ell+1}(x,y).
\end{multline}
If we now take $j\to j-1$ in the sum resulting from the second term, and compare the result with the recurrence relation \eqref{cCoeffsRel} that uniquely determines the coefficients $c_{\ell,j}$, then we arrive at \eqref{nrhyperbolicKernelIds} for $\ell\to\ell+1$.
\end{proof}

We continue by deducing nonrelativistic analogs of the three corollaries \eqref{cor1}--\eqref{cor3}. As a counterpart of \eqref{cor1}, we shall first consider the kernel identities involving the defining Hamiltonian
\beq
H_{{\rm nr}}(x) = W_{\rm nr}(x)^{1/2}\left(\frac{1}{2}\left(:\hat{\Sigma}_1\big(L_{\rm nr}+E\big)(x):\right)^2 - :\hat{\Sigma}_2\big(L_{\rm nr}+E\big)(x):\right)W_{\rm nr}(x)^{-1/2},
\eeq
cf. \eqref{Hnr} and \eqref{Hhnr}. From \eqref{nrcS} we infer
\beq
\begin{split}
\frac{1}{2}c_{\ell,1}^2 - c_{\ell,2} &=\left(\frac{i\mu g}{2}\right)^2 \left(\frac{1}{2}S_1(1,3,\ldots,2\ell-1)^2 - S_2(1,3,\ldots,2\ell-1)\right)\\ &= -\frac{\mu^2g^2}{8}P_2(1,3,\ldots,2\ell-1)\\ &= -\frac{\mu^2g^2}{24}\ell(4\ell^2-1),
\end{split}
\eeq
where $P_2(a_1,\ldots,a_\ell)$ denotes the second power sum symmetric function of $a_1,\ldots,a_\ell$. This yields the following corollary. 

\begin{corollary}
For $\ell=0,\ldots,N$, we have
\beq\label{defHnrP}
\big(H_{{\rm nr}}(x_1,\ldots,x_N) - H_{{\rm nr}}(y_1,\ldots,y_{N-\ell})\big)\Psi_{\ell,{\rm nr}}(x,y) = -\frac{\mu^2g^2}{24}\ell(4\ell^2-1)\Psi_{\ell,{\rm nr}}(x,y),
\eeq
where
\beq
\Psi_{\ell,{\rm nr}}(x,y)\equiv W_{{\rm nr}}(x)^{1/2}W_{{\rm nr}}(y)^{1/2}\cK_\ell(x,y).
\eeq
\end{corollary}

We note that, apart from the exponential factor in the kernel function \eqref{cKLim}, the kernel identity \eqref{defHnrP} coincides with the hyperbolic limit of an elliptic identity obtained by Langmann \cite{Lan06}. The identity can also be obtained as a special case of Corollary~2.3 in \cite{HL10}. More specifically, the latter corollary depends on two polynomials $\alpha$ and $\beta$, and one function $z(x)$, which should be fixed according to
\beq
\alpha(z)=z^2,\ \ \ \beta(z)=z,\quad z(x)=e^{x}.
\eeq
In addition, one should set $\kappa=g$, $M=N-\ell$ and $\tilde{N}=\tilde{M}=0$. This results in a kernel identity that is equivalent to \eqref{defHnrP} for $\hbar=\mu=1$. To be precise, the corresponding kernel function has a slightly different exponential factor, and the kernel identity contains an additional overall factor of $2$.

We proceed to detail a nonrelativistic analog of \eqref{cor2}. From \eqref{nrcS} we have $c_{1,0}=1$ and $c_{1,1}=i\mu g/2$. This yields the following corollary.

\begin{corollary}
For $k=1,\ldots,N$, we have
\begin{multline}
:\hat{\Sigma}_k\big(L_{\rm nr}+E\big)(x_1,\ldots,x_N): \cK_1(x,y) = \Bigg(:\hat{\Sigma}_k\big(L_{\rm nr}+E\big)(-y_1,\ldots,-y_{N-1}): \\ + \frac{i\mu g}{2}:\hat{\Sigma}_{k-1}\big(L_{\rm nr}+E\big)(-y_1,\ldots,-y_{N-1}):\Bigg) \cK_1(x,y).
\end{multline}
\end{corollary}

To obtain a counterpart of~Corollary~2.4, we take $\ell =N$ and note that~\eqref{nrhyperbolicKernelIds} implies
\begin{multline}
\Sigma_k\Big(L_{\rm nr}(x_1,\ldots,x_N,i\mu gN/2,\ldots,i\mu gN/2)+E(x_1,\ldots,x_N)\Big)\cK_N(x)\\ =\left(\frac{i\mu g}{2}\right)^kS_k(1,3,\ldots,2N-1)\cK_N(x).
\end{multline}
In this identity we can cancel~$\cK_N(x)$ and the factor~$(i\mu g/2)^k$, cf.~\eqref{Lnrhyp}--\eqref{zjx}. Then it says that the matrix
\beq\label{MQ}
M=N{\bf 1}_N+Q,\ \ \ Q_{jk}\equiv -\de_{jk}\sum_{n\ne j}\coth (z_j-z_n)+(1-\de_{jk})\sinh(z_j-z_k)^{-1},
\eeq
has the same symmetric functions as the matrix $\diag (1,3,\ldots, 2N-1)$. Shifting the latter matrix and $M$ by $N{\bf 1}_N$, this leads to the following remarkable result.

\begin{corollary}
The matrix $Q$ given by~\eqref{MQ} has spectrum $\{ -N+1,-N+3,\ldots,N-1\}$.
\end{corollary}

We note that, by exploiting the commutation relation
\beq
\varphi(x,y):\hat{\Sigma}_k\big(L_{\rm nr} + E\big)(x):\\ =\ :\hat{\Sigma}_k\left(L_{\rm nr} + E+\frac{i\mu g\ell}{2}{\bf 1}_N\right)(x):\varphi(x,y),
\eeq
where
\beq
\varphi(x,y) \equiv \exp\left(\frac{\mu g\ell}{2\hbar}\left(\sum_{m=1}^Nx_m-\sum_{n=1}^{N-\ell}y_n\right)\right),
\eeq
and the expansion
\begin{multline}
:\hat{\Sigma}_k\left(L_{\rm nr} + E+\frac{i\mu g\ell}{2}{\bf 1}_N\right)(x): \\= \sum_{l=0}^k\left(\frac{i\mu g\ell}{2}\right)^l\binom{N-k+l}{l}:\hat{\Sigma}_{k-l}\left(L_{\rm nr} + E\right)(x):,
\end{multline}
the exponential factor can be removed from \eqref{cKLim}, yielding corresponding identities.

Alternatively, these identities can be deduced from the relativistic kernel identities in Theorem \ref{HypKernelIdProp}, by using that each PDO $:\hat{\Sigma}_k(L_{\rm nr}+E):$ can be obtained as a limit of a linear combination of the identity operator and the commuting \adiffops~$A_1,\ldots,A_k$ that result from $A_{1,+},\ldots,A_{k,+}$ upon substituting \eqref{aT} and
\beq
\rho=i\beta g.
\eeq
Indeed, from Section~4.3 in \cite{Rui94} one can infer 
\beq
:\hat{\Sigma}_k\big(L_{\rm nr}+E\big): = \lim_{\beta\to 0}\beta^{-k}\sum_{j=0}^k(-1)^{k+j}\binom{N-j}{N-k}A_j,\ \ \ \ \ A_0\equiv 1.
\eeq

To illustrate this second method we continue by deducing the pertinent identities for $\ell=1$.  Then we can use Corollary~2.3 to obtain
\begin{multline}\label{sumKIds}
\sum_{j=0}^k(-1)^{k+j}\binom{N-j}{N-k}A_j(x_1,\ldots,x_N)\cS_1(x,y)\\ = \sum_{j=0}^k(-1)^{k+j}\binom{N-j}{N-k}(A_j +A_{j-1})(-y_1,\ldots,-y_{N-1})\cS_1(x,y)\\ = \sum_{n=0}^k(-1)^{k+n}A_{n}(-y_1,\ldots,-y_{N-1})\sum_{l=0}^1(-1)^{l}\binom{N-n-l}{N-k}\cS_1(x,y).
\end{multline}
Hence, for $k=N$ the coefficient of each \adiffop~$A_{n}(-y)$ vanishes. Moreover, for $k<N$ we have (`Pascal's triangle')
\beq
\binom{N-n}{N-k}-\binom{N-n-1}{N-k}= \binom{N-1-n}{N-1-k},\quad k=1,\ldots,N-1,\ \ \ n=0,\ldots,k.
\eeq
Substituting this in \eqref{sumKIds}, multiplying by $\beta^{-k}$, and taking the nonrelativistic limit $\beta\to 0$, we arrive at the following result.

\begin{proposition}
Setting
\beq
\tilde{\cK}_1(x,y) = \prod_{m=1}^N\prod_{n=1}^{N-1}[2\cosh(\mu(x_m-y_n)/2)]^{-g/\hbar},
\eeq
we have an eigenfunction identity
\beq
:\hat{\Sigma}_N\big(L_{\rm nr}+E\big)(x):\tilde{\cK}_1(x,y)=0,
\eeq
and kernel identities
\beq
\left(:\hat{\Sigma}_k\big(L_{\rm nr}+E\big)(x):-:\hat{\Sigma}_k\big(L_{\rm nr}+E\big)(-y):\right)\tilde{\cK}_1(x,y)=0,\quad k=1,\ldots,N-1.
\eeq
\end{proposition}

For $\ell>1$, the identities in question become quite unwieldy, and we have not obtained them explicitly by either of the above methods.
  
The self-duality of the relativistic case is not preserved by the nonrelativistic limit, so we proceed to study kernel functions involving the dual variables. To begin with, we need to revert to the spectral variables $\hat{p}$ via \eqref{phsub}. To ease the notation we omit the hats, so that we wind up with A$\De$Os (cf.~\eqref{Adhyp})
\beq\label{Akr}
A_{\pm k,+}=\sum_{ |I|=k}\prod_{\substack{m\in I\\ n\notin I}}\frac{\sinh(\beta(p_m-p_n\mp i\mu g)/2)}{\sinh(\beta(p_m-p_n)/2)}
\prod_{m\in I}\exp(\mp i\hbar\mu\partial_{p_m}).
\eeq
Their $\beta \to 0$ limits yield the dual nonrelativistic A$\De$Os
\beq
\hat{A}_{\pm k,{\rm nr}}(p)=\sum_{ |I|=k}\prod_{\substack{m\in I\\ n\notin I}}\frac{p_m-p_n\mp i\mu g}{p_m-p_n}
\prod_{m\in I}\exp(\mp i\hbar\mu\partial_{p_m}),\ \ \ \ k=1,\ldots,N.
\eeq
(To be sure, they can also be viewed as the A$\De$Os associated with the relativistic rational Calogero-Moser systems.)

We proceed to obtain the nonrelativistic limit of the dual version of the kernel identities in Theorem~2.1. The reparametrized kernel function \eqref{cSl} reads
\beq\label{Grepar}
\prod_{m=1}^N\prod_{n=1}^{N-\ell}\frac{G(2\pi/\mu,\hbar\beta;\beta(p_m-q_n-i\mu g/2)/\mu)}{G(2\pi/\mu,\hbar\beta;\beta(p_m-q_n+i\mu g/2)/\mu)}.
\eeq
As a preparation for the limit \eqref{GhGr}, we first use the scale invariance~\eqref{hypsc} to write
\beq\label{Gkap}
G(2\pi/\mu,\hbar\beta;\beta z/\mu)=G(1,\kappa; \kappa z/\hbar\mu),\ \ \ \ \kappa =\beta\hbar\mu/2\pi.
\eeq
Next, we note that we may multiply by a constant and shift the coordinates $p_j$ in \eqref{Grepar} without losing the kernel property. A moment's thought then shows that we can invoke \eqref{GhGr} to obtain from \eqref{Grepar} a nonrelativistic dual kernel function
\beq\label{cKd}
\hat{\cK}_{\ell}(p,q)=\prod_{m=1}^N\prod_{n=1}^{N-\ell}\frac{\Gamma(i(p_m-q_n)/\hbar\mu-g/2\hbar)}{\Gamma(i(p_m-q_n)/\hbar\mu+g/2\hbar)},\ \ \ \ \ell=0,1,\ldots,N.
\eeq
Likewise, Theorem~2.1 has the following counterpart.

\begin{theorem}
For any $k\in \{ 1,\ldots,N\}$ and $\tau\in\{ +,-\}$ we have
\beq
	\hat{A}_{\tau k,{\rm nr}}(p_1,\ldots,p_N)\hat{\cK}_\ell(p,q) = \sum_{j=0}^{\min (k,\ell)} \binom{\ell}{j}\hat{A}_{-\tau (k-j),{\rm nr}}(q_1,\ldots,q_{N-\ell})\hat{\cK}_\ell(p,q),
\eeq
where 
\beq
\hat{A}_{\pm m,{\rm nr}}(q_1,\ldots,q_{N-\ell})\equiv 0,\ \ \ m>N-\ell,\ \ \ \hat{A}_{0,{\rm nr}}\equiv 1.
\eeq
\end{theorem}

From this, the nonrelativistic versions of the three corollaries \eqref{cor1}--\eqref{cor3} of Theorem~2.1 will be clear. 

For later use, we add that the nonrelativistic limit of the dual weight function reads
\beq\label{Whypd}
\hat{W}_{\rm nr}(p)=\prod_{1\le m<n\le N}\frac{\Gamma(i (p_m-p_n)/\hbar\mu
+g/\hbar)\Gamma(-i (p_m-p_n)/\hbar\mu
+g/\hbar)}{\Gamma(i (p_m-p_n)/\hbar\mu)\Gamma(-i (p_m-p_n)/\hbar\mu)}.
\eeq
This follows from \eqref{GhGr} in the same way as for the kernel functions. As a check, note that the dual nonrelativistic Hamiltonians,
\bea\label{Hknr1}
\hat{H}_{\pm k,{\rm nr}}(p)  &  =  &  \sum_{ |I|=k}\prod_{\substack{m\in I\\ n\notin I}}\left(\frac{p_m-p_n\mp i\mu g}{p_m-p_n}\right)^{1/2}
\prod_{m\in I}\exp(\mp i\hbar\mu\partial_{p_m})
\nonumber \\
&  &  \times \prod_{\substack{m\in I\\ n\notin I}}
\left(\frac{p_m-p_n\pm i\mu g}{p_m-p_n}\right)^{1/2},
\eea
are then related to the A$\De$Os $\hat{A}_{\pm k,{\rm nr}}(p)$ via
\beq\label{AHW}
\hat{H}_{l,{\rm nr}}(p) =\hat{W}_{\rm nr}(p)^{1/2} \hat{A}_{l,{\rm nr}}(p)\hat{W}_{\rm nr}(p)^{-1/2},\ \ \ \pm l=1,\ldots,N,
\eeq
as should be the case.

\subsubsection{The periodic Toda case}\label{SecNonRelKernelPT}
In this section we obtain a kernel identity for the nonrelativistic periodic Toda system via its relativistic counterpart~\eqref{periodicTidp}. For the same reason as for the elliptic and hyperbolic limit transitions, this involves the classical relativistic and nonrelativistic Lax matrices. Like in the hyperbolic case, we trade the parameters $a_{+}$ and $a_{-}$ for $2\pi/\mu$ and $\hbar\beta$, resp.
Choosing also 
\beq\label{etag}
\eta=\frac{2}{\mu}\ln (\beta\mu g),\ \ \ \ \ g>0,
\eeq
 it follows from \eqref{GLRLimit} with $\lambda =2$ that the relativistic $U$-function \eqref{TU} satisfies
\beq\label{Ulim1}
\lim_{\beta\to 0}U(x)=1.
\eeq
Moreover, \eqref{GLRLimit} with $\lambda =1$ implies that both kernel functions \eqref{Sp} and \eqref{Sm} satisfy
\beq\label{ScK}
\lim_{\beta\to 0}S^{\pm}(x,y)=\cK(x,y),
\eeq
where
\beq\label{cK}
\cK(x,y)= \exp \left(-\frac{g}{\hbar}\sum_{m=1}^N \left(e^{\mu(x_{m+1}-y_{m})}+e^{\mu(y_m-x_m)}\right)\right).
\eeq

After the substitutions \eqref{aT} and \eqref{etag}, we obtain from \eqref{cApp}, \eqref{cAmp} and \eqref{HTp}
\beq\label{cAp1}
	\cA_{1,+}^{+} = \sum_{m=1}^N
	\big(1+\gamma^2e^{\mu(x_{m+1}-x_m+i\hbar\beta/2)}\big)\exp(-i\hbar\beta\partial_{x_m}),
	\eeq
	\beq
	\cA_{1,+}^{-} = \sum_{m=1}^N
	\big(1+\gamma^2e^{\mu(x_m-x_{m-1}-i\hbar\beta/2)}\big)\exp(-i\hbar\beta\partial_{x_m}),	
\eeq
\bea\label{H1}
H_{1,+}  &  =  &  \sum_{m=1}^N
\big(1+\gamma^2e^{\mu(x_{m+1}-x_m+i\hbar\beta/2)}\big)^{1/2}
 \nonumber \\
   &  & \times 
  \big(1+\gamma^2e^{\mu(x_m-x_{m-1}-i\hbar\beta/2)}\big)^{1/2}
  \exp(-i\hbar\beta\partial_{x_m}),
   \eea
   where we have set
   \beq\label{defgam}
   \gamma =\beta\mu g.
   \eeq
Expanding these \adiffops~in a power series in $\beta$, we get in each of the three cases
\beq\label{bexp}
N +\beta \Big(-i\hbar\sum_{j=1}^N\partial_{x_j}\Big)+\beta^2 H_{{\rm nr}}+O(\beta^3),
\eeq
where
\beq\label{HTnr}
H_{{\rm nr}}(x)=-\frac{\hbar^2}{2}\sum_{m=1}^N\partial_{x_m}^2+a^2\sum_{m=1}^N e^{\mu(x_{m+1}-x_m)},\ \ \ \ a=\mu g,
\eeq
is  the defining Hamiltonian of the nonrelativistic periodic Toda system. This implies that $\cK(x,y)$ yields a kernel function for the defining Hamiltonian. Now this can easily be checked directly, but the kernel property for the higher order commuting Hamiltonians is harder to show.

To obtain this more general property, we first need more information on the relation between the commuting relativistic A$\De$Os and the commuting nonrelativistic PDOs.
To this end we begin by recalling a Lax matrix for the classical nonrelativistic periodic Toda system.  This $N\times N$ matrix  depends on a spectral parameter $w\in\mathbb{C}^*$ and is given by
\begin{multline}\label{periodicTodaLax}
	(L_{\text{nr}})_{mn} = \delta_{mn}p_m + \delta_{m,n-1} + a^2\delta_{m,n+1}e^{\mu(x_m-x_{m-1})}\\ - (ia)^Nw\delta_{mN}\delta_{n1} - a^2(ia)^{-N}w^{-1}\delta_{m1}\delta_{nN}e^{\mu(x_1-x_N)}.
\end{multline}
Our subsequent considerations involve the symmetric functions of $L_{\text{nr}}$, given as the coefficients $\Sigma_{k,\text{nr}}$ in the expansion
\begin{equation}\label{Tdet}
	\det({\bf 1}_N+\lambda L_{\text{nr}}) = \sum_{k=0}^N \lambda^k\Sigma_{k,\text{nr}}.
\end{equation}
In particular, this yields
\begin{equation}
	\Sigma_{1,\text{nr}} = \sum_{m=1}^Np_m,\quad \Sigma_{2,\text{nr}} = \sum_{1\leq m<n\leq N}p_mp_n - a^2\sum_{m=1}^Ne^{\mu(x_m-x_{m-1})}.
\end{equation}
(For $N=2$ one should add the constant $a^2(w+1/w)$ to $\Sigma_{2,\text{nr}}$.)
 The formula~\eqref{Tdet} has an unambiguous quantum analog. Indeed, 
a term in the expansion of the determinant that depends on a given $p_m$ does not depend on $x_m$, since the only elements of $L_{\text{nr}}$ that depend on $x_m$ occur in the $m$th column and row. Denoting by $\hat{L}_{\text{nr}}$ the matrix obtained from $L_{\text{nr}}$ by substituting $-i\hbar\partial_{x_m}$ for $p_m$, $m=1,\ldots,N$, the determinant of $\hat{L}_{\text{nr}}$ can now be defined by the usual expansion, since no ordering problems arise. Specifically, we have
\begin{equation}
	\det\left({\bf 1}_N+\lambda \hat{L}_{\text{nr}}\right) = \sum_{k=0}^N \lambda^k\hat{\Sigma}_{k,\text{nr}},
\end{equation}
where $\hat{\Sigma}_{k,\text{nr}}$ is obtained from $\Sigma_{k,\text{nr}}$ via the above substitution. Next, we show that the PDO $\hat{\Sigma}_{k,\text{nr}}$ can be obtained as a limit of a certain linear combination of the identity operator and commuting \adiffops~$A^{(\tau)}_1,\ldots,A^{(\tau)}_k$, $\tau 
\in\{+,-\}$. The latter are obtained from the A$\De$Os $\cA^{\pm}_{1,+},\ldots, \cA^{\pm}_{N,+}$ by substituting \eqref{aT} and 
\beq\label{etapm}
\eta =\frac{2}{\mu}\ln \gamma \mp i\hbar\beta/2,\ \ \ \ \gamma=\beta\mu g,
\eeq
respectively. Thus we have (cf.~\eqref{cApp}, \eqref{cAmp} and \eqref{Tfn})
\beq\label{Apmk}
A^{(\pm)}_k= \sum_{|I|=k}\prod_{\substack{m\in I\\ m\pm 1\notin I}}\big(1+\gamma^2 e^{\mu (x_{m+1}-x_{m})}\big)\prod_{m\in I}\exp(-i\hbar\beta\partial_{x_m}),\ \ \ \ k=1,\ldots,N.
\eeq

\begin{lemma}
Fix $w\in\mathbb{C}^*$, and let
\begin{equation}\label{Dkt}
	D^{(\tau)}_k=(-1)^k\binom{N}{N-k}+ \sum_{j=1}^k (-1)^{k+j}\binom{N-j}{N-k}c_jA^{(\tau)}_j,
\end{equation}
where $k=1,\ldots,N$, $ \tau=+,-$, and
\begin{equation}\label{cj}
	c_j= \left\lbrace\begin{array}{ll}
		\big(1+(i\gamma)^Nw\big)^{j-1}, & j=1,\ldots,N-1,\\
		\big(1+(i\gamma)^Nw\big)^{N-1}\big(1+(i\gamma)^Nw^{-1}\big), & j=N.
	\end{array}\right.
\end{equation}
Then for all $k\in\{ 1,\ldots,N\}$ and $\tau\in \{ +,-\}$ we have
\begin{equation}\label{DSig}
	\lim_{\beta\to 0}\beta^{-k}D^{(\tau)}_k = \hat{\Sigma}_{k,{\rm nr}}.
\end{equation}
\end{lemma}

\begin{proof}
The following reasoning involves a modification of the arguments employed in~Subsection~4.3 of~\cite{Rui94} to handle the  nonrelativistic limit of the elliptic A$\De$Os. We first note that in the \adiffops~$D_k^{(\tau)}$ all partial derivatives occur to the right of the $x$-dependent coefficients. Since no ordering problems arise in the definition of the differential operators $\hat{\Sigma}_{k,\text{nr}}$, it suffices to prove the classical version of the statement, obtained by replacing $-i\hbar\nabla_{x}$ by $p$.

To this end we use the Lax matrix $L_{\text{nr}}$ introduced above, as well as a Lax matrix for the classical relativistic periodic Toda system. Specifically, let $L$ be the $N\times N$ matrix given by
\begin{equation}\label{LTp}
	L_{mn} = \beta^{n-m}b_mE_{mn},
\end{equation}
where
\begin{equation}\label{bm}
	b_m = \big(1+\gamma^2e^{\mu(x_{m+1}-x_m)}\big)^{1/2}\big(1+\gamma^2e^{\mu(x_m-x_{m-1})}\big)^{1/2}e^{\beta p_m},
\end{equation}
and
\begin{gather}
	E_{1N} = \frac{1-(i\gamma)^{-N}w^{-1}\gamma^2e^{\mu(x_1-x_N)}}{1+\gamma^2e^{\mu(x_1-x_N)}},\\
	E_{mn} = 1,\quad n-m=N-2,\ldots,1,0,\\
	E_{m,m-1} = \frac{-(i\gamma)^Nw+\gamma^2e^{\mu(x_m-x_{m-1})}}{1+\gamma^2e^{\mu(x_m-x_{m-1})}},\quad m=2,\ldots,N,\\
	E_{mn} = -(i\gamma)^Nw,\quad n-m=-2,\ldots,-N+1.
\end{gather}
By exploiting a limit from a Lax matrix for the elliptic system it can be shown that the symmetric functions of $L$ read
\begin{equation}
	\Sigma_k (x,p)= c_kS_k(x,p),\ \ \ \ k=1,\ldots,N.
\end{equation}
Here, the coefficients $c_k$ are given by~\eqref{cj} and
\beq\label{SkT}
S_k(x,p)= \sum_{|I|=k}\prod_{\substack{m\in I\\ m+1\notin I}}\big(1+\gamma^2e^{\mu(x_{m+1}-x_{m})}\big)^{1/2}
\prod_{\substack{m\in I\\ m-1\notin I}}\big(1+\gamma^2e^{\mu(x_{m}-x_{m-1})}\big)^{1/2}
\prod_{m\in I}\exp(\beta p_m).
\eeq
(The details can be found in~Subsection 3.2 of \cite{Rui94}.) Moreover, recalling
\beq
\gamma =\beta \mu g=\beta a,
\eeq
one readily checks 
\begin{equation}\label{Lexp}
	L ={\bf 1}_N + \beta L_{\text{nr}} + O(\beta^2),\quad \beta\to 0.
\end{equation}
Expanding the determinant of the matrix ${\bf 1}_N+\lambda\beta^{-1}(L-{\bf 1}_N)$ in powers of $\lambda$, we deduce the limit
\begin{multline}\label{SymFuncsPeridodicToda}
	\Sigma_{k,\text{nr}}(x,p) = \lim_{\beta\to 0}\beta^{-k}\left( (-1)^k\binom{N}{N-k}+\sum_{j=1}^k (-1)^{k+j}\binom{N-j}{N-k}c_jS_j(x,p)\right),\\ k=1,\ldots,N.
\end{multline}

To relate this limit to $A^{(\tau)}_k(x,p)$, we introduce the diagonal matrix
\begin{equation}\label{Dd}
	D = \text{diag}(d_1,\ldots, d_N),\quad d_m = \frac{\big(1+\gamma^2e^{\mu(x_{m+1}-x_m)}\big)^{1/2}}{\big(1+\gamma^2e^{\mu(x_m-x_{m-1})}\big)^{1/2}}.
\end{equation}
The principal minor of $D$ with indices $\lbrace i_1,\ldots, i_k\rbrace\equiv I$ reads
\begin{equation}
	D(I) = \prod_{m\in I}\frac{\prod_{m+1\notin I}\big(1+\gamma^2e^{\mu(x_{m+1}-x_m)}\big)^{1/2}}{\prod_{m-1\notin I}\big(1+\gamma^2e^{\mu(x_m-x_{m-1})}\big)^{1/2}},
\end{equation}
so the symmetric functions of the matrices $D^{\pm 1} L$ are given by
\begin{equation}
	\sum_{ |I|=k}\big(D^{\pm 1} L\big)(I) = \sum_{ |I|=k}D(I)^{\pm 1} L(I) = c_k A_k^{(\pm )}(x,p).
\end{equation}
Now from \eqref{Lexp} and \eqref{Dd} we clearly have
\begin{equation}
	D^{\pm 1}L={\bf 1}_N+\beta L_{\rm nr} + O(\beta^2),\quad \beta\to 0.
\end{equation}
Therefore, the limit \eqref{SymFuncsPeridodicToda} is still valid when we replace $S_j$ by $A_j^{(+)}$
or $A_j^{(-)}$. Hence the classical version of the statement follows, thus completing the proof.
\end{proof}

In particular, this lemma implies the commutativity of the $N$ PDOs $\hat{\Sigma}_{1,{\rm nr}},\ldots,\hat{\Sigma}_{N,{\rm nr}}$ (a property left open in~\cite{Rui94}). It is easy to verify that for $k=1$ and $k=2$ one still gets the limit \eqref{DSig} when $A_j^{(\tau)}$ in \eqref{Dkt} is replaced by $\cA_{j,+}^{+}$, $\cA_{j,+}^{-}$ or $H_{j,+}$ with the substitutions \eqref{aT} and \eqref{etag} in force. Already for $k=3$, however, this is no longer clear, since the $\beta$-dependence in the exponentials yields unwieldy extra terms in the expansion, cf.~e.g.~\eqref{cAp1}--\eqref{H1}. (In the elliptic case such extra terms in the A$\De$Os $H_{1,+},H_{2,+}$ do give nontrivial $\hbar$-dependent deviations from the classical expansion already for $k=2$, as exhibited in \eqref{H2} and \eqref{Hnr}.)

We are now prepared for the nonrelativistic counterpart of Theorem~2.7. 

\begin{theorem}
Setting
\begin{equation}
	D(x) = \det\left({\bf 1}_N+\lambda\hat{L}_{\rm nr}(x)\right),
\end{equation}
where $\hat{L}_{\rm nr}(x)$ is defined by \eqref{periodicTodaLax} with the 
substitution $-i\hbar\partial_{x_m}$ for $p_m$,
 we have a kernel identity
\begin{equation}\label{DK}
	(D(x)-D(-y))\mathcal{K}(x,y) =0,
\end{equation}
with the kernel function $\cK(x,y)$  given by \eqref{cK}. This identity also holds true when $\cK(x,y)$ is replaced by $\cK(\sigma(x),y)$, with $\sigma$ any cyclic permutation.
\end{theorem}
\begin{proof}
The limit \eqref{ScK} is still valid when $\eta$ in the kernel functions $S^{\pm}(x,y)$ is replaced by \eqref{etapm} instead of \eqref{etag}. Indeed, this change gives rise to shifts of the arguments of the functions $G_R$ and $G_L$ in \eqref{Sp}--\eqref{Sm} by $\pm i\hbar \beta/4$. This amounts to a shift of $z$ by $\pm ia_{-}/4$ in the limit formula \eqref{GLRLimit}. Since the convergence in this formula is uniform on compact subsets of the plane, the limit is unchanged. In view of the above lemma and \eqref{periodicTidp}--\eqref{periodicTidm} we can now deduce \eqref{DK}. The kernel identities in Theorem~2.7 also hold for cyclic transforms of $S^{\pm}(x,y)$, so the same is true for \eqref{DK}. Alternatively, one can invoke the cyclic invariance of the PDOs $\hat{\Sigma}_{1,{\rm nr}},\ldots,\hat{\Sigma}_{N,{\rm nr}}$, which follows from that of the A$\De$Os $A_j^{(\pm)}$ by the lemma.
\end{proof}

\subsubsection{The nonperiodic Toda case}

As in the periodic case, we first reparametrize $a_{+},a_{-}$ and $\eta$ via \eqref{aT} and \eqref{etag}. Then it follows as before (cf.~\eqref{Ulim1}--\eqref{ScK}) that the $U$-function \eqref{TUnp} has $\beta \to 0$ limit 1, whereas the nonrelativistic  limit of the nonperiodic kernel functions $\cS^{\pm}(x,y)$(cf.~\eqref{Spl}--\eqref{Sml})  is given by 
\beq\label{Knr}
\cK(x,y)= \exp \left(-\frac{g}{\hbar}\left( e^{\mu(y_N-x_N)}+\sum_{m=1}^{N-1} \left(e^{\mu(x_{m+1}-y_{m})}+e^{\mu(y_m-x_m)}\right)\right)\right).
\eeq
Proceeding as in the periodic case, we get again \eqref{cAp1}--\eqref{defgam} with the nonperiodic Toda convention \eqref{npconv}. Hence we deduce the $\beta$-expansion \eqref{bexp}, now with the defining Hamiltonian
\beq\label{HTnp}
H_{{\rm nr}}(x)=-\frac{\hbar^2}{2}\sum_{m=1}^N\partial_{x_m}^2+a^2\sum_{m=1}^{N-1} e^{\mu(x_{m+1}-x_m)},\ \ \ \ a=\mu g,
\eeq
of the nonrelativistic nonperiodic Toda system.  

By contrast to the relativistic case, there seems to be no limit transition leading from the periodic kernel function given by \eqref{cK} to the nonperiodic one \eqref{Knr}. On the other hand, it is again easy to obtain the nonperiodic Hamiltonian  \eqref{HTnp} from the periodic one given by \eqref{HTnr}. It is also straightforward to verify directly the kernel function property
\beq
(H_{\rm nr}(x)-H_{\rm nr}(y))\cK(x,y)=0.
\eeq
To obtain the kernel property for the higher order commuting PDOs, we proceed as in the periodic case. Thus, we use a nonrelativistic Lax matrix
\beq\label{Lnrnp}
(L_{\text{nr}})_{mn} = \delta_{mn}p_m + \delta_{m,n-1} + a^2\delta_{m,n+1}e^{\mu(x_m-x_{m-1})},
\eeq
and a relativistic one $L$~\eqref{LTp}, with $b_m$ given by
\beq
b_1=\left( 1+\gamma^2 e^{\mu(x_2-x_1)}\right)^{1/2} e^{\beta p_1},\ \ \ \ 
b_N=\left( 1+\gamma^2 e^{\mu(x_N-x_{N-1})}\right)^{1/2} e^{\beta p_N},
\eeq
and  by \eqref{bm} for $m=2,\ldots,N-1$, and with the matrix $E$ given by
\beq
	E_{mn} = 1,\quad n-m=N-1,\ldots,1,0,
	\eeq
	\beq
	E_{m,m-1} = \frac{\gamma^2e^{\mu(x_m-x_{m-1})}}{1+\gamma^2e^{\mu(x_m-x_{m-1})}},\ \ \ m=2,\ldots,N,
	\eeq
	\beq\label{E3}
	E_{mn} = 0,\quad n-m=-2,\ldots,-N+1.
\eeq

It is not difficult to check that $L_{\text{nr}}$ and $L$ are again related by~\eqref{Lexp}, and that  the symmetric functions of $L$ are given by~\eqref{SkT} (with the convention~\eqref{npconv} in effect). Defining the auxiliary A$\De$Os $A^{(\pm)}_k$ via \eqref{etapm} and \eqref{Apmk}, and then $D^{(\pm)}_k$ by \eqref{Dkt} with the coefficients $c_j$ set equal to 1, we obtain \eqref{DSig} by an easy adaptation of the proof of Lemma~4.2. The remark below the proof applies here, too, and now the following analog of Theorem~4.3 readily follows.

\begin{theorem}
Setting
\begin{equation}
	D(x) = \det\left({\bf 1}_N+\lambda\hat{L}_{\rm nr}(x)\right),
\end{equation}
where $\hat{L}_{\rm nr}(x)$ is defined by \eqref{Lnrnp} with the 
substitution $-i\hbar\partial_{x_m}$ for $p_m$,
 we have a kernel identity
\begin{equation}\label{DKnp}
	(D(x)-D(-y))\mathcal{K}(x,y) =0,
\end{equation}
with the kernel function $\cK(x,y)$  given by \eqref{Knr}. 
\end{theorem}

Finally, we obtain a counterpart of Corollary~2.9, namely, kernel identities that relate the PDOs $\hat{\Sigma}_{k,{\rm nr}}$ in $N$ variables $x=(x_1,\ldots,x_N)$ to the PDOs $\hat{\Sigma}_{k,{\rm nr}}$ in $N-1$ variables $y=(y_1,\ldots,y_{N-1})$ for $k<N$.

\begin{corollary} Setting
\beq\label{K1nr}
\cK_1(x,y)= \exp \left(-\frac{g}{\hbar} \sum_{m=1}^{N-1} \left(e^{\mu(x_{m+1}-y_{m})}+e^{\mu(y_m-x_m)}\right)\right),
\eeq
we have an eigenfunction identity
\beq
\hat{\Sigma}_{N,{\rm nr}}(x)\cK_1(x,y)=0,
\eeq
and kernel identities
\beq
(\hat{\Sigma}_{k,{\rm nr}}(x) -\hat{\Sigma}_{k,{\rm nr}}(-y))\cK_1(x,y)=0,\ \ \ \ k=1,\ldots, N-1.
\eeq
\end{corollary}
\begin{proof}
The kernel function $\cK_1((x_1,\ldots,x_N),(y_1,\ldots,y_{N-1}))$ is obtained from the kernel function $\cK(x,y)$ given by~\eqref{Knr} upon substituting $y_N\to y_N-\Lambda$ and then taking $\Lambda$ to infinity. Doing so in the Lax matrix $\hat{L}_{\rm nr}(y)$, the matrix element 
$(\hat{L}_{\rm nr})_{N,N-1}$ vanishes in the limit. Noting $\partial_{y_N}$ annihilates $\cK_1(x,y)$, the assertions now follow from Theorem~4.4 upon expansion of the determinants.
\end{proof}

\subsubsection{The dual nonperiodic Toda case}

To obtain the nonrelativistic limits of the quantities in Subsection~2.5, we proceed in the same way as for the dual hyperbolic case. Thus we
reparametrize $a_{+},a_{-}$ via \eqref{aT} and revert to the spectral variables $\hat{p}$ via \eqref{phsub}. Omitting the hats on $p$, we wind up with A$\De$Os (cf.~\eqref{defDl} and \eqref{Hd2})
\beq\label{Dkr}
\hat{A}_{\pm k,+}=(\mp i)^{k(N-k)}\sum_{ |I|=k}\prod_{\substack{m\in I\\ n\notin I}}\frac{1}{2\sinh(\beta(p_m-p_n)/2)}
\prod_{m\in I}\exp(\mp i\hbar\mu\partial_{p_m}),
\eeq
\bea\label{Hkr}
\hat{H}_{\pm k,+}  &  =  &  \sum_{ |I|=k}\prod_{\substack{m\in I\\ n\notin I}}\left|\frac{1}{2\sinh(\beta(p_m-p_n)/2)}\right|^{1/2}
\prod_{m\in I}\exp(\mp i\hbar\mu\partial_{p_m})
\nonumber \\
&  &  \times \prod_{\substack{m\in I\\ n\notin I}}\left|\frac{1}{2\sinh(\beta(p_m-p_n)/2)}\right|^{1/2},
\eea
related by the weight function
\beq
\hat{W}=\prod_{1\le j<k\le N}4\sinh(\beta(p_j-p_k)/2)\sinh(\pi(p_j-p_k)/\hbar\mu),
\eeq
cf.~\eqref{Wd}--\eqref{DWH}.

Clearly, we need only multiply \eqref{Dkr}--\eqref{Hkr}  by $(\hbar\mu\beta)^{k(N-k)}$ and take $\beta$ to 0 to get the commuting A$\De$Os
\beq\label{Dknr}
\hat{A}_{\pm k,{\rm nr}}(p)=(\mp i)^{k(N-k)}\sum_{ |I|=k}\prod_{\substack{m\in I\\ n\notin I}}\frac{\hbar\mu}{p_m-p_n}
\prod_{m\in I}\exp(\mp i\hbar\mu\partial_{p_m}),
\eeq
\beq\label{Hknr2}
\hat{H}_{\pm k,{\rm nr}}(p)    =    \sum_{ |I|=k}\prod_{\substack{m\in I\\ n\notin I}}\left|\frac{\hbar\mu}{p_m-p_n}\right|^{1/2}
\prod_{m\in I}\exp(\mp i\hbar\mu\partial_{p_m})
 \prod_{\substack{m\in I\\ n\notin I}}\left|\frac{\hbar\mu}{p_m-p_n}\right|^{1/2}.
\eeq
They are related by~\eqref{AHW},
where
\beq\label{WhT}
\hat{W}_{\rm nr}(p)=\prod_{1\le j<k\le N}2((p_j-p_k)/\hbar\mu)\sinh(\pi(p_j-p_k)/\hbar\mu)
\eeq
is the nonrelativistic dual weight function.

Next, we obtain the nonrelativistic limit of the kernel identities in Subsection~2.5 by arguing as in the dual hyperbolic case (cf.~the paragraph containing~\eqref{Gkap}), using also that we may multiply the kernel function by $c_1\exp(c_2\sum(p_j-q_j))$ without losing the kernel property.

\begin{theorem}
Letting $l\in\{ \pm 1,\ldots,\pm N\}$ and $\sigma\in\{ \pm 1\}$, we have the dual kernel function identities
\begin{equation}\label{DcK}
	\big(\hat{A}_{l,{\rm nr}}(p) - \hat{A}_{-l,{\rm nr}}(q)\big)\hat{\cK}(p,q)^{\sigma} = 0,
\end{equation}
\beq\label{HdK}
\big(\hat{H}_{l,{\rm nr}}(p) - \hat{H}_{-l,{\rm nr}}(q)\big)\hat{W}_{\rm nr}(p)^{1/2}\hat{W}_{\rm nr}(q)^{1/2}\hat{\cK}(p,q)^{\sigma} = 0,
\eeq
where
\beq\label{cKhT}
\hat{\cK}(p,q)=\prod_{j,k=1}^N \Gamma (i(p_j-q_k)/\hbar\mu ).
\eeq
\end{theorem}

Each term of the A$\De$Os occurring in~\eqref{DcK} and~\eqref{HdK} shifts $|l|$ coordinates by $\pm i \hbar\mu$, so the kernel property is preserved upon multiplication by products of functions that are $i\hbar\mu$-antiperiodic. Hence we can derive the kernel property of $1/\hat{\cK}(p,q)$ from that of $\hat{\cK}(-p,-q)$ by using the reflection equation of the gamma function, in the form
\beq
\Gamma(iz+1/2)\Gamma(-iz+1/2)=\pi/\cosh(\pi z).
\eeq

We proceed with the following counterpart of Theorem~2.12.

\begin{theorem}
Define kernel functions
\begin{equation}\label{NRDTKernelFunc}
	\hat{\cK}_{\ell}(p,q) = \prod_{m=1}^N\prod_{n=1}^{N-\ell}\Gamma\big(1-i(p_m-q_n)/\hbar\mu\big),
\end{equation}
where $\ell=0,1,\ldots,N-1$. For any $k\in\{1,\ldots,N-\ell\}$ we have
\beq\label{Dlid}
	\hat{A}_{k,\rm {nr}}(p_1,\ldots,p_N)\hat{\cK}_\ell (p,q)  = \hat{A}_{-k,{\rm nr}}(q_1,\ldots,q_{N-\ell})\hat{\cK}_\ell (p,q).
	\eeq
\end{theorem}

\begin{proof} We bypass a laborious derivation from Theorem~2.12 by adapting its proof, as follows. First, the analogs of the equations \eqref{Ssm}--\eqref{Ssp} read
\beq\label{Kp}
\hat{\cK}_{\ell}(p,q)^{-1}\exp(-i\hbar\mu\partial_{p_m}) \hat{\cK}_{\ell}(p,q)=\left(i\hbar\mu\right)^{N-\ell}\prod_{n=1}^{N-\ell}\frac{1}{p_m-q_n},
\eeq
\beq\label{Kq}
\hat{\cK}_{\ell}(p,q)^{-1} \exp(i\hbar\mu\partial_{q_m}) \hat{\cK}_{\ell}(p,q)=
\left(i\hbar\mu\right)^{N}\prod_{n=1}^{N}\frac{1}{p_n-q_m}.
\eeq
Second, we take $\de =+$ in the identities \eqref{sId2} and  substitute (recall $s_{+}(z)=\sinh(\pi z/a_{+})$)
\beq
v=ta_{+}p/\pi,\ \ \ w=ta_{+}q/\pi.
\eeq
If we now multiply both sides of \eqref{sId2} by $t^{k(2N-k-\ell)}$ and  take  $t\to 0$, then we obtain the identities
\begin{multline}\label{DTIds}
	\sum_{\substack{I\subset\lbrace 1,\ldots,N\rbrace\\ |I|=k}}\prod_{\substack{m\in I\\ n\notin I}}\frac{1}{p_n-p_m}\prod_{\substack{m\in I\\ n\in\lbrace 1,\ldots,N-\ell\rbrace}}\frac{1}{p_m-q_n}\\ =\sum_{\substack{I\subset\lbrace 1,\ldots,N-\ell\rbrace\\ |I|=k}}\prod_{\substack{m\in I\\ n\notin I}}\frac{1}{q_m-q_n}\prod_{\substack{m\in I\\ n\in\lbrace 1,\ldots,N\rbrace}}\frac{1}{p_n-q_m}.
\end{multline}
Finally, using \eqref{Dknr} and \eqref{Kp}--\eqref{Kq}, it easily follows that the kernel identities \eqref{Dlid} amount to \eqref{DTIds}.
\end{proof}

Introducing 
\beq
\tilde{\cK}_{\ell}(p,q)\equiv \hat{\cK}_{\ell}(-p,-q),
\eeq
we can mimic the reasoning leading to Corollary~2.13 to obtain the following nonrelativistic analog, which concludes this subsection.

\begin{corollary}
We have eigenfunction identities
\beq
\hat{A}_{\ell,{\rm nr}}(v_1,\ldots,v_N)\tilde{\cK}_\ell(v,w) = \tilde{\cK}_\ell(v,w),
\eeq
and kernel identities
\beq
\hat{A}_{k,{\rm nr}}(v_1,\ldots,v_N)\tilde{\cK}_\ell(v,w) = \hat{A}_{-(k-\ell),{\rm nr}}(w_1,\ldots,w_{N-\ell})\tilde{\cK}_\ell(v,w),\ \ \ k=\ell+1,\ldots,N.
\eeq
\end{corollary}

\subsection{B\"acklund transformations}

\subsubsection{The elliptic case}

To obtain the nonrelativistic versions of the results in Subsection~3.1, we can proceed in two distinct ways. First, we can adapt the reasoning based on the expected relation \eqref{PsiF} to the nonrelativistic kernel function $\Psi_{\rm nr}(x,y)$ given by~\eqref{Psinr}. The second way is to take $\beta$ to 0 in the relevant formulas in Subsection~3.1. This yields the same results, provided we replace $\rho$ again by $i\beta g$, as we did in Subsection~4.1, cf.~\eqref{rhog}. Since we want to view $g$ as a real coupling constant, it follows from the limits of \eqref{pj} and \eqref{qj} that we wind up with purely imaginary momenta $p_j$ and $q_j$. Thus, we run into the same problem as we had in Subsection~3.1. We improved the situation in the relativistic elliptic setting by requiring $\beta$ to be purely imaginary, while keeping $\rho\in i(0,\alpha)$. From this it is clear that we can emulate this improvement by replacing $\rho$ by $\beta g$ before taking $\beta$ to 0, keeping $g$ real.

Doing so, we wind up with a generating function
\begin{equation}\label{defFnr}
	F_{\rm nr}(x,y) =F_{W_{\rm nr}}(x) + F_{W_{\rm nr}}(y) + F_{\cK}(x,y),
\end{equation}
where
\begin{equation}\label{FWnr}
	F_{W_{\rm nr}}(x) = \frac{g}{2}\sum_{1\le j< k\le N}\ln (R(x_j-x_k+i\alpha/2)R(x_j-x_k-i\alpha/2)),
\end{equation}
\begin{equation}\label{FcK}
	F_{\cK}(x,y) =-g \sum_{j,k=1}^N \ln (R(x_j-y_k)).
\end{equation}
Obviously, this can also be obtained via \eqref{PsiF} and \eqref{Psinr}, provided we replace $g$ by $-ig$ in \eqref{cSnr} and \eqref{Wnr}.

In order to show the B\"acklund property, we use the relation between the elliptic relativistic and nonrelativistic Lax matrices from~\cite{Rui94}. The relativistic one is defined by 
\beq\label{Lell}
L_{jk}=\exp(\beta p_j)\prod_{l\ne j}f(x_j-x_l)\cdot \frac{s(x_j-x_k+\lambda)s(\rho)}{s(\lambda)s(x_j-x_k+\rho)},
\eeq
where the function $f(z)$ is given by \eqref{f}, and $\lambda\in \C$ is a spectral parameter. Its symmetric functions $\Sigma_k$ are proportional to the Hamiltonians $S_k$~\eqref{Hamiltonians}. Specifically, 
\beq\label{SigS}
\Sigma_k(x,p)=s(\lambda)^{-k}s(\lambda -\rho)^{k-1}s(\lambda +(k-1)\rho)S_k(x,p),\ \ \ k=1,\ldots,N.
\eeq
The nonrelativistic Lax matrix is defined by
\begin{equation}\label{Kric}
	(L_{\rm nr})_{jk} = \delta_{jk}p_j + ig(1-\delta_{jk})\frac{s(x_j-x_k+\lambda)}{s(\lambda)s(x_j-x_k)}.
\end{equation}
Up to a similarity transformation, it coincides with the elliptic Lax matrix introduced by Krichever \cite{Kri80}. (At this point a physicist reader might worry about non-matching dimensions, inasmuch as $p$ has dimension [momentum], whereas the coupling constant $g$ has dimension [position]$\times$[momentum]. But $\lambda$ has dimension [position], and so does $s(\lambda)$. Thus the dimensions work out.)

Clearly, the Lax matrices $L$ with $\rho=i\beta g $ and $L_{\rm nr}$ are related in the same way as in the periodic Toda case, cf.~\eqref{Lexp}. Hence their symmetric functions are related by
\begin{equation}\label{SSnr}
	\Sigma_{k,\text{nr}}(x,p) = \lim_{\beta\to 0}\beta^{-k}\sum_{l=0}^k (-1)^{k+l}\binom{N-l}{N-k}\Sigma_k(x,p),\ \ \ \ k=1,\ldots,N.
\end{equation}
The same relation holds for the symmetric functions of the matrices $L$ with $\rho=\beta g$ and $L_{\rm nr}$ with $g$ replaced by $-ig$. Therefore, the nonrelativistic B\"acklund property
\beq\label{Bnr}
\Sigma_{\rm nr}(x,p)=\Sigma_{\rm nr}(y,q),
\eeq
follows from its relativistic counterpart \eqref{SS} by using  \eqref{SigS} and \eqref{SSnr}.

Although this reasoning is formally impeccable, it skirts the existence problem already discussed in Subsection~3.1. Moreover, the noncompleteness of the flows is now even more conspicuous due to the \lq wrong sign\rq\ of the coupling. Indeed, for the defining Hamiltonian associated with $L_{\rm nr}$~\eqref{Kric} we get
\beq
H_{\rm nr}=\frac{1}{2}{\rm Tr}\,  L_{\rm nr}^2=\frac{1}{2}\sum_{j=1}^Np_j^2 +g^2\sum_{1\le j<k\le N}\wp(x_j-x_k) -\frac{1}{2}g^2N(N-1)\wp(\lambda).
\eeq
Thus, taking $g\to -ig$ leads to a negative coupling, so that the singularities at coinciding positions cannot be avoided. (Recall we require that  $g$  be real to avoid imaginary momenta resulting from the nonrelativistic versions of \eqref{pj} and \eqref{qj}, cf.~\eqref{defFnr}--\eqref{FcK}.)

\subsubsection{The hyperbolic case and its dual}

Proceeding as in the elliptic case, we obtain from~\eqref{cSnrh}--\eqref{Wnrh} once more \eqref{defFnr}--\eqref{FcK}, with $R(z)$ replaced by $2\cosh (\mu z/2)$. In the hyperbolic versions of the Lax matrices \eqref{Lell} and \eqref{Kric} we can take $\lambda$ to infinity. Omitting a similarity factor and reparametrizing via~\eqref{aT}--\eqref{tau} and \eqref{rhog}, we obtain
\beq\label{Ltau}
L(\tau)_{jk}=\exp(\beta p_j)\prod_{l\ne j}\left( 1+\frac{\sin^2(\tau)}{\sinh^2(\mu(x_j-x_l)/2)}\right)^{1/2}
\cdot \frac{i\sin(\tau)}{\sinh(i\tau+\mu(x_j-x_k)/2)},
\eeq
\beq
L_{\rm nr}(g)_{jk} = \delta_{jk}p_j + (1-\delta_{jk})\frac{i\mu g}{2\sinh(\mu(x_j-x_k)/2)},
\eeq
in accord with~\cite{Rui88} (where a parameter $z$ is used instead of $i\tau$).
Then \eqref{Lexp} holds true again, so it still holds for the Lax matrices $L(-i\beta\mu g/2)$ and $L_{\rm nr}(-ig)$, where we now think of $\beta$ being purely imaginary and $\mu,g$ real.
Hence the symmetric functions of the latter Lax matrices are again related by \eqref{SSnr}, and accordingly the negative coupling B\"acklund property \eqref{Bnr} results.  As in the relativistic case, its precise interpretation within the confines of global analysis/symplectic geometry remains to be determined.

Turning to the dual case, the relevant dual Lax matrices are~\cite{Rui88}
\beq\label{Ldtau}
\hat{L}(\tau)_{jk}=\exp(\mu \hat{x}_j)\prod_{l\ne j}\left( 1+\frac{\sin^2(\tau)}{\sinh^2(\beta(\hat{p}_j-\hat{p}_l)/2)}\right)^{1/2}
\cdot \frac{i\sin(\tau)}{\sinh(i\tau-\beta(\hat{p}_j-\hat{p}_k)/2)},
\eeq
\beq
\hat{L}_{\rm nr}(g)_{jk}=\exp(\mu \hat{x}_j)\prod_{l\ne j}\left( 1+\frac{(\mu g)^2}{(\hat{p}_j-\hat{p}_l)^2}\right)^{1/2}
\cdot \frac{i\mu g}{i\mu g-(\hat{p}_j-\hat{p}_k)},
\eeq
with $\hat{L}_{\rm nr}(g)$ being the $\beta \to 0$ limit of $\hat{L}(\beta\mu g/2)$. Thus we can expect the nonrelativistic B\"acklund property
\beq\label{nrB}
S_{k,{\rm nr}}(\hat{p},\hat{x})
=S_{k,{\rm nr}}(\hat{q},\hat{y}),
\eeq 
as a limit of the relativistic one. However, it is not straightforward to obtain the nonrelativistic generating function $F_{\rm nr}(\hat{p},\hat{q})$ as a limit of $F(\hat{p},\hat{q})$ (given by \eqref{defFhyp}--\eqref{FShyp} with $\beta$ and $\mu$ interchanged), since the renormalizations mentioned below 
\eqref{Gkap} must be taken into account.

Instead, we start directly from the dual kernel function
\beq
\Psi_{\rm nr}(\hat{p},\hat{q})=\hat{W}_{\rm nr}(\hat{p})^{1/2}\hat{W}_{\rm nr}(\hat{q})^{1/2}\hat{\cK}(\hat{p},\hat{q}),
\eeq
where $\hat{\cK}$ is defined by \eqref{cKd} with $\ell=0$, and $\hat{W}_{\rm nr}$ by \eqref{Whypd}. As before, one might expect \eqref{PsiF} to yield the desired generating function $F_{\rm nr}(\hat{p},\hat{q})$ as a limit
\beq\label{nol}
\lim_{\hbar\downarrow 0}i\hbar \ln \Psi_{\rm nr}(\hat{p},\hat{q}).
\eeq
In fact, however, this limit does not exist, as is obvious from the following lemma.

\begin{lemma}
For $\hat{p}$ and $\hat{q}$ in $\hat{G}_{\rm hyp}$~\eqref{dconfighyp}, we have classical limits
\beq
\lim_{\hbar\downarrow 0}i\hbar \ln (\exp[N^2(g/\hbar)\ln (1/\hbar)] \hat{\cK}(\hat{p},\hat{q}))=i\sum_{j,k=1}^N
\int_{i(\hat{p}_j-\hat{q}_k)/\mu +g/2}^{i(\hat{p}_j-\hat{q}_k)/\mu -g/2}dw\ln w,
\eeq
\beq
\lim_{\hbar\downarrow 0}i\hbar \ln (\exp[-N(N-1)(g/\hbar)\ln (1/\hbar)]
\hat{W}_{\rm nr}(\hat{p}))=i\sum_{j<k}
\int_{i(\hat{p}_j-\hat{p}_k)/\mu -g}^{i(\hat{p}_j-\hat{p}_k)/\mu +g }dw\ln w,
\eeq
where the integration paths stay away from the cut $(-\infty, 0]$.
\end{lemma}

\begin{proof}
We recall Stokes' formula, in the form
\beq\label{Sto}
\lim_{\Lambda\to\infty}\frac{1}{\Lambda}\ln \left(\exp[(d-u)\Lambda\ln \Lambda]\frac{\Gamma(\Lambda u)}{\Gamma(\Lambda d)}\right)=\int_d^udw \ln w,\ \ \ \ u,d\notin (-\infty, 0].
\eeq
Inspecting the definitions \eqref{cKd} (with $\ell=0$) and \eqref{Whypd}, the limits readily follow from this.
\end{proof}

Instead of using \eqref{nol}, it is now clear that we need to define the generating function by
\beq
F_{\rm nr}(\hat{p},\hat{q})=\lim_{\hbar\downarrow 0}i\hbar \ln (\exp[N(g/\hbar)\ln(1/\hbar)] \Psi_{\rm nr}(\hat{p},\hat{q})).
\eeq
Hence $\hat{q}(\hat{x},\hat{p})$ is to be determined from the equations
\bea\label{hatq}
\hat{x}_j  & = & -\frac{\partial F_{\rm nr}}{\partial
 \hat{p}_j}
\nonumber \\ 
  &  =  &  \frac{1}{2\mu}\sum_{k\ne j} \ln \left( \frac{\hat{p}_j-\hat{p}_k-i\mu g}{\hat{p}_j-\hat{p}_k+i\mu g}\right)+
  \frac{1}{\mu}\sum_{k=1}^N  \ln \left( \frac{\hat{p}_j-\hat{q}_k+i\mu g/2}{\hat{p}_j-\hat{q}_k-i\mu g/2}\right),
\eea
and then $\hat{y}(\hat{x},\hat{p})$ is given by
\bea
\hat{y}_j  &  =   &  \frac{\partial F_{\rm nr}}{\partial \hat{q}_j}
\nonumber \\ 
  &  =  &  \frac{1}{2\mu}\sum_{k\ne j} \ln \left( \frac{\hat{q}_j-\hat{q}_k+i\mu g}{\hat{q}_j-\hat{q}_k-i\mu g}\right)+
  \frac{1}{\mu}\sum_{k=1}^N  \ln \left( \frac{\hat{p}_k-\hat{q}_j+i\mu g/2}{\hat{p}_k-\hat{q}_j-i\mu g/2}\right).
\eea
Clearly, when we retain the physical choice $\mu, g>0$, then we get a contradiction from assuming that for given $(\hat{x},\hat{p})\in\hat{\Omega}_{\rm hyp}$ the implicit equations \eqref{hatq} yield a solution $\hat{q}\in\hat{G}_{\rm hyp}$. (Indeed, it would follow that $\hat{x}$ is purely imaginary.) We can only avoid this snag by keeping $\mu$ positive, while requiring that $g$ be purely imaginary.

Accepting this and assuming (possibly complex) solutions, we deduce as before (cf.~Subsection~3.1) that the B\"acklund property is equivalent to the functional identities
\begin{multline}
	\sum_{\substack{I\subset\lbrace 1,\ldots,N\rbrace\\ |I|=k}}\prod_{\substack{m\in I\\ n\notin I}}\frac{\hat{p}_m-\hat{p}_n-i\mu g}{\hat{p}_m-\hat{p}_n}\prod_{\substack{m\in I\\ n=1,\ldots,N}}\frac{\hat{p}_m-\hat{q}_n+i\mu g/2}{\hat{p}_m-\hat{q}_n-i\mu g/2}\\ = \sum_{\substack{I\subset\lbrace 1,\ldots,N\rbrace\\ |I|=k}}\prod_{\substack{m\in I\\ n\notin I}}\frac{\hat{q}_m-\hat{q}_n+i\mu g}{\hat{q}_m-\hat{q}_n}\prod_{\substack{m\in I\\ n=1,\ldots,N}}\frac{\hat{p}_n-\hat{q}_m+i\mu g/2}{\hat{p}_n-\hat{q}_m-i\mu g/2}.
\end{multline}
These are easily deduced from \eqref{Bsc}, and so \eqref{nrB}  follows.

\subsubsection{The periodic Toda case}

Comparing the expected asymptotic relation \eqref{PsiF} to the nonrelativistic kernel function $\cK(x,y)$ given by \eqref{cK}, it becomes clear that in this case no limit is needed. The resulting generating function,
\beq\label{Tgen}
-ig\sum_{m=1}^N \left(e^{\mu(x_{m+1}-y_{m})}+e^{\mu(y_m-x_m)}\right),
\eeq
would yield purely imaginary momenta, but like in Subsection~3.3 this disease can be cured by an analytic continuation
\beq
x_m\to x_m-i\pi/2\mu,\ \ \ \ m=1,\ldots,N.
\eeq
Then we obtain the generating function
\beq
F_{\rm nr}(x,y)=g\sum_{m=1}^N \left(-e^{\mu(x_{m+1}-y_{m})}+e^{\mu(y_m-x_m)}\right),
\eeq
which gives rise to
\beq\label{pmnr}
p_m  = \mu g \left( e^{\mu(y_m-x_m)}+e^{\mu (x_m-y_{m-1}}\right),
   \eeq
   \beq\label{qmnr}
q_m  =  \mu g \left( e^{\mu(y_m-x_m)}+e^{\mu (x_{m+1}-y_{m}}\right).
   \eeq
These equations can be regarded as the nonrelativistic limit of the equations \eqref{pm}--\eqref{qm}. Indeed, replacing $\gamma$ by $\beta \mu g$ and shifting $x_m\to x_m-\ln(\beta\mu g)/\mu$ in the latter equations (in accord with \eqref{defgam} and \eqref{xac}), it is obvious that their $\beta \to 0$ limit yields \eqref{pmnr}--\eqref{qmnr}. 

Next, we recall that we already detailed a Lax matrix $L_{\rm nr}$ \eqref{periodicTodaLax} for the nonrelativistic periodic Toda system, and obtained its symmetric functions $\Sigma_{k,\text{nr}}$ as  limits of appropriate linear combinations of the relativistic Hamiltonians $S_k$, cf.~\eqref{SymFuncsPeridodicToda}. Therefore the B\"acklund property
\beq
\Sigma_{k,\rm nr}(x,p)=\Sigma_{k,\rm nr}(y,q),\ \ \ \ k=1,\ldots,N,
\eeq
follows from its relativistic counterpart. For the defining Hamiltonian
\beq\label{Hcl}
H_{\rm nr}(x,p)=\frac{1}{2}\sum_{m=1}^N p_m^2 +a^2\sum_{m=1}^N e^{\mu(x_{m+1}-x_m)},\ \ \ a=\mu g,
\eeq
it is of course easily checked directly from \eqref{pmnr}--\eqref{qmnr}.

\subsubsection{The nonperiodic Toda case and its dual}

The nonperiodic Toda counterparts of the formulas \eqref{Tgen}--\eqref{Hcl} will be obvious by now: we need only insist on the convention \eqref{npconv} for $x$ and $y$. Also, the B\"acklund property follows by using the Lax matrices given by  \eqref{Lnrnp}--\eqref{E3}.

A study of the dual system is less straightforward. Of course, from \eqref{WhT} we get the same result
\beq
\lim_{\hbar\downarrow 0}i\hbar \ln 
\hat{W}_{\rm nr}( \hat{p})=\frac{i\pi}{\mu}\sum_{1\le m<n\le N}
(\hat{p}_m-\hat{p}_n),
\eeq
as in the relativistic case, cf.~\eqref{hWlim}. For the kernel functions $\hat{\cK}(\hat{p},\hat{q})^{\sigma}$ given by \eqref{cKhT}, however, the relevant limit vanishes for $\sigma =-1$ and does not exist for $\sigma =1$. On the other hand, in this case we have even more freedom to modify kernel functions than indicated in the paragraph containing \eqref{ambig}, since we can also multiply by products of arbitrary $i\hbar\mu$-antiperiodic functions. Exploiting this, one can probably obtain the same generating function as we now shall arrive at by starting from the modified relativistic generating function $\tilde{F}(\hat{p},\hat{q})$~\eqref{Fmod}.

The crux is that when we add the function 
\beq
-\frac{N}{\mu}\ln(\beta) \sum_{m=1}^N (\hat{p}_m-\hat{q}_m),
\eeq
 to $\tilde{F}$, discard a constant, and then take $\beta$ to 0, we get the limit function
\beq
\tilde{F}_{\rm nr}(\hat{p},\hat{q})=  \frac{i\pi}{2\mu}\sum_{m=1}^N\Big( (N-2m+2)\hat{p}_m+(N-2m)\hat{q}_m\Big)+ \frac{1}{\mu}\sum_{m,n=1}^N\int_0^{\hat{p}_m-\hat{q}_n}dw\ln (-iw)  .
\eeq
As in Subsection~3.4, the corresponding equations
\begin{equation}\label{xdTnr}
	\hat{x}_m=-\frac{\partial\tilde{F}_{\rm nr}}{\partial \hat{p}_m} = -\frac{i\pi}{2\mu}(N-2m+2)-\frac{1}{\mu}\sum_{n=1}^N\ln\big(-i(\hat{p}_m-\hat{q}_n)\big),
\end{equation}
\begin{equation}\label{ydTnr}
	\hat{y}_m=\frac{\partial\tilde{F}_{\rm nr}}{\partial \hat{q}_m} = \frac{i\pi}{2\mu}(N-2m)-\frac{1}{\mu}\sum_{n=1}^N\ln\big(-i(\hat{p}_n-\hat{q}_m)\big),
\end{equation}
might well yield a solution $\hat{q}\in\hat{G}$, $\hat{y}\in\R^N$, with the interlacing property \eqref{interl}. In any case, the B\"acklund property
\beq
\hat{H}_{k,{\rm nr}}(\hat{x},\hat{p}) =\hat{ H}_{k,{\rm nr}}(\hat{y},\hat{q}),\ \quad k=1,\ldots,N,
\end{equation}
for the dual Hamiltonians
\beq
\hat{H}_{k,{\rm nr}}(\hat{x},\hat{p})= \sum_{\substack{I\subset\lbrace 1,\ldots,N\rbrace\\ |I|=k}}
 \prod_{\substack{m\in I\\ n\notin I}}\frac{1}{|\hat{p}_m-\hat{p}_n|}
\prod_{l\in I}\exp(\mu \hat{x}_l),\quad k=1,\ldots,N,
\eeq
can now be shown in the same way as in Subsection~3.4.

\begin{appendix}

\section{Elliptic and hyperbolic gamma functions}\label{gammaFuncsAppendix}
The elliptic and hyperbolic gamma functions were introduced and studied in~\cite{Rui97} as so-called minimal solutions of certain first order analytic difference equations. (The hyperbolic gamma function has various differently-named cousins, as detailed in Appendix~A of~\cite{Rui05}.) In this appendix we review features of these gamma functions that are relevant for the present paper.

The following material concerning the elliptic gamma function $G(r,a_+,a_-;z)$ can all be found in Subsection III~B of~\cite{Rui97}. To begin with,
the elliptic gamma function
can be defined by the product representation
\beq\label{Gell}
	G(r,a_+,a_-;z) = \prod_{m,n=0}^\infty \frac{1-\exp \big(-(2m+1)ra_+-(2n+1)ra_--2irz\big)}{1-\exp \big(-(2m+1)ra_+-(2n+1)ra_-+2irz\big)}.
\eeq
Here and below, we require that the parameters satisfy
\beq
r,a_{+},a_{-}>0.
\eeq
It is obvious from~\eqref{Gell} that the elliptic gamma function is meromorphic in $z$, with poles and zeros that can be read off. In particular, for $z$ in the strip 
\beq\label{strip}
|\Im (z)|<a,\ \ \ \ a=(a_{+}+a_{-})/2,
\eeq
no poles and zeros occur, so that we have
\beq\label{Gg}
G(z)=\exp(ig(z)),\ \ \ \ |\Im (z)|<a,
\eeq
with the function $g(z)$ being analytic in the strip. (We often suppress parameters when no ambiguity arises.) In fact, 
it is explicitly given by
\beq\label{grep}
g(r,a_{+},a_{-};z)=\sum_{n=1}^{\infty}\frac{\sin(2nrz)}{2n\sinh(nra_{+})\sinh(nra_{-})},\ \ \ \ |\Im (z)|<a.
\eeq
Both from this series representation and from~\eqref{Gell}, the following properties are clear:
\beq\label{refl}	
G(-z) = 1/G(z),\ \ \ ({\rm reflection\ equation}),
\eeq
\beq\label{modinv}
	G(a_-,a_+;z) = G(a_+,a_-;z),\ \ \  ({\rm modular\ invariance}),
	\eeq
	\beq\label{scale}
	G(\lambda^{-1}r,\lambda a_+,\lambda a_-;\lambda z) = G(r,a_+,a_-;z),\quad \lambda\in(0,\infty),\ \ \ ({\rm scale\ invariance}).
\eeq

The elliptic gamma function arises as a minimal solution of analytic difference equations that involve a right-hand side function defined by
\beq\label{defR}
R(r,\alpha;z) =\prod_{k=1}^{\infty}[1-\exp(2irz-(2k-1)\alpha r)][1-\exp(-2irz-(2k-1)\alpha r)].
\eeq
Specifically, setting
\beq\label{Rdel}
R_\delta(z)= R(r,a_\delta;z),\ \ \ \ \de=+,-,
\eeq
it satisfies
\begin{equation}\label{ellipticGDiffEq}
	\frac{G(z+ia_\delta/2)}{G(z-ia_\delta/2)} = R_{-\delta}(z),\quad \delta=+,-.
\end{equation}
It is clear from~\eqref{defR} and~\eqref{Rdel} that the functions $R_{\pm}$ are entire, even and $\pi/r$-periodic, and satisfy
\begin{equation}\label{RDiffEq}
	\frac{R_\delta(z+ia_\delta/2)}{R_\delta(z-ia_\delta/2)} = -\exp(-2irz),\ \ \ \de=+,-.
\end{equation}

For the classical and nonrelativistic limits in the main text we need to invoke two related zero step size limits of the elliptic gamma function. First, for $z$ and $w$ staying away from cuts given by
\beq\label{cuts}
 \pm i[\alpha/2,\infty)+k\pi/r,\ \ \ k\in\Z,
\eeq
we have
\beq\label{Gellcl}
\lim_{a_{-}\downarrow 0}a_{-} g(r,\alpha,a_{-};z) = -\int_0^z\, dw\ln R(r,\alpha;w),
\eeq
where the logarithm takes real values for $w$ real; moreover, this limit is uniform on compact subsets of the cut plane. In particular, when we have an upper limit $z$ with $|\Im(z)|<\alpha/2$, we can choose a path along which $|\Im (w)|<\alpha/2$ and use the representation
\beq\label{Rrep}
\ln R(r,\alpha;z)=-\sum_{n=1}^{\infty}\frac{\cos (2nrz)}{n\sinh(nr\alpha)},\ \ \ |\Im (z)|<\alpha/2,
\eeq
which follows from \eqref{ellipticGDiffEq} and \eqref{grep}. The second limit reads
\beq\label{Gellnr}
\lim_{a_{-}\downarrow 0}\frac{G(r,\alpha, a_{-};z+iua_{-})}{G(r,\alpha, a_{-};z+ida_{-})}
=\exp((u-d)\ln R(r,\alpha;z)),\ \ \ \ u,d\in\R,
\eeq
uniformly on compact subsets of the cut plane. Note that for $u-d$ integer this limit readily follows from~\eqref{ellipticGDiffEq}.

In the main text we make extensive use of the functions
\beq\label{sdel}
s_{\de}(z)=s(r,a_{\de};z),\ \ \ \de=+,-,
\eeq
defined by
\beq\label{ssigrel}
s(r,\alpha;z)=\exp(-\eta rz^2/\pi)\sigma(z;\pi/2r,i\alpha/2),
\eeq
where $\sigma$ denotes the Weierstrass sigma function. 
Hence these functions are entire, odd, $\pi/r$-antiperiodic and satisfy
\begin{equation}
	\frac{s_\delta(z+ia_\delta/2)}{s_\delta(z-ia_\delta/2)} = -\exp(-2irz),\ \ \ \ \de=+,-.
\end{equation}
They are related to $R_{\pm}$ via the formula
\beq\label{sR}
s_{\de}(z)=\frac{1}{2ir}\prod_{k=1}^\infty\frac{1}{\big(1-\exp(-2ka_{\de}r)\big)^2}\cdot \exp(irz) R_{\de}(z+ia_{\de}/2).
\eeq
Also, the well-known product representation for the Weierstrass sigma function entails
\begin{equation}\label{sProductRep}
	s_\delta(z) = \frac{a_\delta}{\pi}e_\delta (-rz^2/\pi)\sinh(\pi z/a_\delta)\prod_{l=1}^\infty\frac{\big(1-e_\delta(2z-2\pi l/r)\big)\big(z\to -z\big)}{\big(1-e_\delta(-2\pi l/r)\big)^2},
\end{equation}
where we have introduced the notation
\beq
e_{\de}(z)=\exp(\pi z/a_{\de}),\ \ \ \ \de=+,-.
\eeq
Clearly, this implies
\begin{equation}\label{sToSinhLim}
	\lim_{r\downarrow 0}s_\delta(z) = \frac{a_\delta}{\pi}\sinh(\pi z/a_{\de}),
\end{equation}
the limit being uniform on compact subsets of $\C$. 

We proceed to discuss the hyperbolic gamma function $G(a_+,a_-;z)$, cf.~Subsection~III~A of~\cite{Rui97}. It can be defined as the unique minimal solution of the analytic difference equations
\begin{equation}\label{hyperbolicgammaDiffEq}
	\frac{G(z+ia_\delta/2)}{G(z-ia_\delta/2)} = 2\cosh(\pi z/a_{-\delta}),\quad \delta=+,-,
\end{equation}
satisfying $G(0)=1$.
It arises from the elliptic gamma function by the following limit:
\begin{equation}\label{ellTohypgammaLim}
	\lim_{r\downarrow 0}G(r,a_+,a_-;z)\exp \left(\frac{\pi^2z}{6ira_+a_-}\right) = G(a_+,a_-;z).
\end{equation}
It is meromorphic in $z$, and for $z$ in the strip~\eqref{strip} it has neither poles nor zeros. Thus it can be written as in \eqref{Gg}, with $g(z)$ analytic in the strip. Explicitly, $g(z)$ has the integral representation
\beq\label{ghyp}
g(a_{+},a_{-};z) =\int_0^\infty\frac{dy}{y}\left(\frac{\sin 2yz}{2\sinh(a_{+}y)\sinh(a_{-}y)} - \frac{z}{a_{+}a_{-} y}\right),\ \ \ \ |\Im (z)|<a.
\eeq
From this it is clear that the hyperbolic gamma function also satisfies the reflection equation~\eqref{refl} and has the modular invariance property~\eqref{modinv}, whereas the counterpart of~\eqref{scale} reads
	\beq\label{hypsc}
	G(\lambda a_+,\lambda a_-;\lambda z) = G(a_+,a_-;z),\quad \lambda\in(0,\infty),\ \ \ ({\rm scale\ invariance}).
\eeq

In the main text we need several zero step size limits of the hyperbolic gamma function. The first one yields the relation to the Euler gamma function:
\beq\label{GhGr}
\lim_{\kappa\downarrow 0}G(1,\kappa;
 \kappa z+i/2)\exp \big(iz\ln(2\pi\kappa)-\ln(2\pi)/2\big) = 1/\Gamma(iz+1/2).
\eeq
The second and third one are needed for the classical and nonrelativistic limits, resp.: For $z$ and $w$ staying away from  cuts given by $\pm i[\alpha/2,\infty)$, we have 
\beq\label{Ghypcl}
\lim_{a_{-}\downarrow 0}a_{-} g(\alpha,a_{-};z) = -\int_0^z\, dw\ln(2\cosh (\pi w/\alpha)),
\eeq
\beq\label{Ghypnr}
\lim_{a_{-}\downarrow 0}\frac{G(\alpha, a_{-};z+iua_{-})}{G(\alpha, a_{-};z+ida_{-})}
=\exp((u-d)\ln (2\cosh(\pi z/\alpha))),\ \ \ \ u,d\in\R,
\eeq
uniformly on compact subsets of the cut plane.  

In the relativistic Toda setting it is expedient to switch to two slightly different hyperbolic gamma functions given by 
\beq\label{GRDef}
	G_R(a_+,a_-;z) = G(a_+,a_-;z)\exp \left(i\chi + \frac{i\pi z^2}{2a_+a_-}\right),
	\eeq
\beq\label{GLDef}
	G_L(a_+,a_-;z) = G(a_+,a_-;z)\exp \left(-i\chi - \frac{i\pi z^2}{2a_+a_-}\right),
\eeq
where
\begin{equation}
 \chi = \frac{\pi}{24}\left(\frac{a_+}{a_-} + \frac{a_-}{a_+}\right).
\end{equation}
These functions are the unique minimal solutions of the analytic difference equations
\beq\label{GRDE}
	\frac{G_R(z+ia_{-\delta}/2)}{G_R(z-ia_{-\delta}/2)} = 1+e_\delta(-2z),
	\eeq
\beq\label{GLDE}
	\frac{G_L(z+ia_{-\delta}/2)}{G_L(z-ia_{-\delta}/2)} = 1+e_\delta(2z),
\eeq
with asymptotic behavior
\beq\label{GRLAsA}
	G_{\substack{R \\ L}}(z) = 1 + \mathcal{O}\left(\exp(-r |\Re(z)|)\right),\quad \Re(z)\to\pm\infty,
	\eeq
\beq\label{GRLAsB}
	G_{\substack{R \\ L}}(z) = \exp \big(\pm i\left(2\chi+\pi z^2/a_+a_-\right)\big)\big(1 + \mathcal{O}(\exp(-r |\Re(z)|))\big),\quad \Re(z)\to\mp\infty,
\eeq
where the decay rate $r$ can be any positive number satisfying 
\beq
r <2\pi\min(a_+,a_-)/a_+a_-.
\eeq
Furthermore, they are related by
\beq\label{GRGL}
G_R(z)G_L(-z)=1.
\eeq
 The properties of the functions $G_R$ and $G_L$ just stated are easy to infer from the corresponding properties of the hyperbolic gamma function. See also Appendix A in~\cite{Rui05}, where functions $S_R$ and $S_L$ were introduced that differ from $G_R$ and $G_L$ by the shift $z\to z-ia$. 
 
 Finally, we have occasion to use  the limits
\beq\label{GLRLimit}
	\lim_{a_-\to 0}G_{\substack{R\\L}}\left(a_{+},a_{-};z\pm \lambda\frac{a_+}{2\pi}\ln\frac{1}{a_-}\right) = 
	\left\lbrace \begin{array}{ll}	\exp \left(\pm\frac{ia_+}{2\pi}e_+(\mp 2z)\right), & \lambda=1,\\  1 , & \lambda>1,\end{array}\right.
	\eeq
which hold uniformly on compact subsets of $\C$. These limits are proved in Appendix~B of~\cite{Rui11}.

\section{A family of Hilbert-Schmidt operators}

In this appendix we reconsider the periodic Toda kernel functions $S^{\pm}$ given by \eqref{Sp} and \eqref{Sm}. More precisely, shifting $y_n$ by $\xi$ in $S^{+}(x,y)$,
we focus on the resulting function
\beq\label{Sxi2}
S_{\xi}(x,y)= \prod_{n=1}^N\frac{G_R(y_{n}-x_{n+1}-ia/2-\eta/2+\xi)}{G_L(y_n-x_n+ia/2+\eta/2+\xi)},\ \ \eta\in\R,\ \ \ \xi\in\C,
\eeq
noting that basically the same results apply to $S^{-}(x,y)$.
Letting $\xi$ vary over the strip
\begin{equation}\label{xir}
	| \Im(\xi)| < a/2,
\end{equation}
we shall associate a family of Hilbert-Schmidt operators to the kernel functions~\eqref{Sxi2}. To this end we need some preliminaries.

We denote by $E$ the Euclidean space
\begin{equation}
	E = \lbrace x\in\mathbb{R}^N\mid x_1+\cdots+x_N=0\rbrace,
\end{equation}
with inner product given by the restriction of the standard inner product on $\mathbb{R}^N$. Furthermore, we denote the Lebesgue measure on $E$ by $\lambda_E$ and the corresponding Hilbert space by $L^2(E)$. Next, we introduce a change of coordinates
\begin{equation}\label{coords}
	s = \frac{1}{N}(x_1 +\cdots + x_N),\quad r_n = x_n - x_{n+1},\quad n=1,\ldots,N-1,
\end{equation}
with inverse given by
\begin{equation}\label{xsr}
	x_m = s - \frac{1}{N}\sum_{n=1}^{N-1}nr_n + \sum_{n=m}^{N-1}r_n,\quad m=1,\ldots,N.
\end{equation}
The associated Jacobian determinant equals 1, as is readily verified. Viewing $r_1,\ldots,r_{N-1}$ as coordinates on $E$ by taking $s=0$ in~\eqref{xsr}, we deduce
\beq\label{lamr}
d\lambda_E=N^{-1/2}dr_1\cdots dr_{N-1}.
\eeq

To explain the relevance of these coordinates for the Toda
 \adiffops~$\cA^{+}_{ l,\delta}(x)$ (cf.~\eqref{cApp}-\eqref{cApm}),  we point out that they factorise as a product of a center-of-mass operator and a reduced operator:
\begin{equation}
	\cA^{+}_{ l,\delta} = \cA^{+,\text{cm}}_{ l,\delta}A^{+,\text{r}}_{ l,\delta},\quad \cA^{+,\text{cm}}_{ l,\delta} = \exp(- ila_{-\delta}\partial_s/N).
\end{equation}
Here, the reduced \adiffops~$\cA^{+,\text{r}}_{ l,\delta}$ depend only on the variables $r_1,\ldots,r_{N-1}$, so they commute with the center-of-mass operators. It is easy to verify that the kernel identity \eqref{periodicTidp} remains valid if we substitute for $\cA^+_{\pm l,\delta}$ either $\cA^{+,\text{r}}_{\pm l,\delta}$ or $\cA^{+,\text{cm}}_{\pm l,\delta}$. Hence, viewing $L^2(\mathbb{R}^N)$ as a tensor product
\begin{equation}
	L^2(\mathbb{R}^N)\simeq L^2(\mathbb{R})\otimes L^2(E),
\end{equation}
the findings of this appendix can be used to study Hilbert space aspects of the A$\De$Os, but this is beyond our present scope.

We are now prepared to state the main result of this appendix. 

\begin{theorem}\label{HilbertSchmidtPeriodicThm}
The operator on $L^2(E)$ defined by
\begin{equation}
	(\hat{S_{\xi}}f)(x)\equiv \int_E S_\xi(x,y)f(y)d\lambda_E(y),\quad f\in L^2(E),
\end{equation}
with $S_{\xi}(x,y)$ given by~\eqref{Sxi2}, is a Hilbert-Schmidt operator, provided $\xi$ satisfies~\eqref{xir}.
\end{theorem}

To prove this theorem we need to show that the $\xi$-restriction entails
\begin{equation}\label{HilbertSchmidtInEquality}
	I_{\xi}\equiv \int_{E\times E}|S_\xi(x,y)|^2d\lambda_E(x)d\lambda_E(y) < \infty.
\end{equation}
As it stands, this integral is hard to estimate. Using the above coordinates $r$ on $E(x)$ and the same coordinates $q$ on $E(y)$ we have from~\eqref{lamr}
\beq\label{Ixi2}
I_{\xi}=N^{-1} \int_{E\times E}|S_\xi(x(r),y(q))|^2drdq,
\eeq
but at face value this seems no improvement.

We shall therefore introduce new coordinates on $E\times E$, in which the function $S_\xi(x,y)$ takes a particularly simple form. First, we set
\begin{equation}\label{udef}
	u_{2n-1} = x_n-y_n,\quad u_{2n} = y_n-x_{n+1},\quad n=1,\ldots,N-1.
\end{equation}
Now on $E\times E$ we have $x_1+\cdots+x_N=y_1+\cdots+y_N=0$, so that
\begin{align}
	y_N-x_1 &= -u_2-u_4-\cdots-u_{2N-2},\\
	y_N-x_N &= u_1+u_3+\cdots+u_{2N-3}.
\end{align}
For $(x,y)\in E\times E$, we thus have 
\begin{equation}\label{SInTermsOfU}
\begin{split}
	S_\xi(x,y) &= \frac{G_R(-u_2-u_4-\cdots-u_{2N-2}-ia/2-\eta/2+\xi)}{G_L(u_1+u_3+\cdots+u_{2N-3}+ia/2+\eta/2+\xi)}\\ &\quad\times\prod_{n=1}^{N-1}\frac{G_R(u_{2n}-ia/2-\eta/2+\xi)}{G_L(-u_{2n-1}+ia/2+\eta/2+\xi)}.
\end{split}
\end{equation}
With a slight abuse of notation, we shall write $S_\xi(u)$ for the right-hand side of this expression. 

To check that $u_1,\ldots,u_{2N-2}$ yield well-defined coordinates on $E\times E$, we first note that~\eqref{udef} entails
\beq\label{ru}
	r_n = x_n - x_{n+1} = u_{2n-1} + u_{2n},\quad n=1,\ldots,N-1,
	\eeq
	\beq
	q_n = y_n - y_{n+1} = u_{2n} + u_{2n+1},\quad n=1,\ldots,N-2,
\eeq
\begin{equation}\label{qu}
	q_{N-1} = y_{N-1} - y_N = u_{2N-2} - (u_1+u_3+\cdots+u_{2N-3}).
\end{equation}
From this we readily deduce
\beq
\frac{\partial(r_1,q_1,\ldots,r_{N-1},q_{N-1})}{\partial(u_1,\ldots,u_{2N-2})}=N,
\eeq
so the linear transformation given by~\eqref{ru}--\eqref{qu} maps $\R^{2N-2}$ onto $\R^{2N-2}$.
Furthermore, from~\eqref{Ixi2} we get
\beq
I_{\xi}=\int_{\mathbb{R}^{N-2}}|S_\xi(u)|^2du.
\eeq
Therefore, the inequality \eqref{HilbertSchmidtInEquality}, and hence the validity of Theorem \ref{HilbertSchmidtPeriodicThm}, are clear from the following lemma.

\begin{lemma}\label{SEstimateLemma}
Assuming $\xi$ satisfies~\eqref{xir}, there exist constants $B,b>0$ such that
\begin{equation}
	|S_\xi(u)|\leq B\exp(-b||u||),\quad \forall u\in\mathbb{R}^{2N-2}.
\end{equation}
\end{lemma}

In order to prove this lemma we shall make use of the following observation.

\begin{lemma}\label{TuplesLemma}
Let $\alpha = (\alpha_1,\ldots,\alpha_n)\in\mathbb{R}^n$ be such that $||\alpha||=1$. Let $P$ and $N$ be the sum of the positive and negative $\alpha_j$, respectively. Then, at least one of the numbers $P+N$ and $-N$ is greater than or equal to $1/(2\sqrt{n})$.
\end{lemma}

\begin{proof}
If $N\leq -1/(2\sqrt{n})$, then the statement is obviously true. Next, assume $N$ is greater than $-1/(2\sqrt{n})$. Clearly, $|\alpha_j|\geq 1/\sqrt{n}$ for at least one $j$. By the assumption on $N$, all such $\alpha_j$ must be positive. It follows that $P\geq 1/\sqrt{n}$, so that
\begin{equation}
	P + N > \frac{1}{\sqrt{n}} - \frac{1}{2\sqrt{n}} = \frac{1}{2\sqrt{n}}.
\end{equation}
\end{proof}

\begin{proof}[Proof of Lemma \ref{SEstimateLemma}]
We observe that $S_\xi(u)$ is a meromorphic function of $u$ whose poles are located at 
\begin{equation}
	u_{2n} =\eta/2-\xi - ia/2  - i(ka_+ + la_-) ,\quad	u_{2n-1} =\eta/2 + \xi - ia/2 - i(ka_+ + la_-),
\end{equation}
with $n=1,\ldots,N-1$, and at
\begin{gather}
	u_2 + u_4 +\cdots + u_{2N-2} = -\eta/2 +\xi+ia/2 + i(ka_+ + la_-),\\
	u_1 + u_3 +\cdots +u_{2N-3} =-\eta/2  - \xi +ia/2 +i(ka_+ + la_-),
\end{gather}
where $k,l\in\mathbb{N}$. Since we have $-a/2<\Im(\xi)<a/2$, it follows that $S_\xi(u)$ is real-analytic for all $u\in\mathbb{R}^{2N-2}$, and hence bounded on any compact subset of $\mathbb{R}^{2N-2}$. 

We proceed to study the asymptotic behavior of $S_\xi(u)$. To this end we assume $u\neq 0$, and let $\beta = (\beta_1,\ldots,\beta_{2N-2})$ be the corresponding unit vector $u/||u||$. We observe that the numerator and denominator in \eqref{SInTermsOfU} depend only on $\beta_n$ with $n$ even and odd, respectively. It is therefore convenient to introduce
\begin{align}
	\beta^o &= (\beta_1,\beta_3,\ldots,\beta_{2N-3}),\\
	\beta^e &= (\beta_2,\beta_4,\ldots,\beta_{2N-2}).
\end{align}
It follows from Lemma \ref{TuplesLemma} and \eqref{GRLAsA}--\eqref{GRLAsB} that the numerator of \eqref{SInTermsOfU} is bounded above by
\begin{equation}
	C\exp \left(\frac{\pi||\beta^e||\big(\Im(\xi)-a/2\big)}{a_+a_-\sqrt{N-1}}||u||\right)
\end{equation}
for some constant $C>0$. Similarly, the denominator is bounded below by
\begin{equation}
	D\exp \left(\frac{\pi||\beta^o||\big(\Im(\xi)+a/2\big)}{a_+a_-\sqrt{N-1}}||u||\right)
\end{equation}
for some constant $D>0$. This clearly implies the statement. In fact, these bounds imply that we can choose
\beq
b= \frac{\pi}{a_+a_-\sqrt{N-1}}\min(a/2-\Im(\xi),\Im(\xi)+a/2).
\eeq
\end{proof}

\end{appendix}

\bibliographystyle{amsalpha}

\end{document}